\documentclass[twoside,11pt]{article}
\usepackage{blindtext}

\usepackage[preprint]{jmlr2e}

\usepackage[utf8]{inputenc} 

\usepackage{lipsum}
\usepackage{subfiles}
\usepackage{booktabs}
\usepackage{chngcntr}

\usepackage [english]{babel}
\usepackage [autostyle, english = american]{csquotes}
\MakeOuterQuote{"}

\let\proof\relax

\usepackage{amsmath,amsthm,amssymb}
\usepackage{amsfonts} 

\usepackage{natbib}
\usepackage{tikz} 
\usepackage{caption,subcaption}
\usepackage{graphicx}
\usepackage{amsmath, amsthm, amsfonts, microtype,booktabs}
\usepackage{wrapfig}
\usepackage{lscape}
\usepackage{rotating}

\usepackage{algorithm}
\usepackage{algpseudocode}

\usepackage[T1]{fontenc}    
\usepackage{hyperref}       
\usepackage{url}            
\usepackage{booktabs}       
\usepackage{amsfonts}       
\usepackage{nicefrac}       
\usepackage{microtype}      
\usepackage{lipsum}		
\usepackage{graphicx}
\usepackage{doi}
\usepackage{lscape}
\usepackage{chngcntr}
\usepackage{caption,subcaption}
\usepackage{setspace}

\usepackage{algorithm}
\usepackage{algpseudocode}

\usepackage{microtype}
\usepackage{graphicx}
\usepackage{booktabs} 

\makeatletter
\newcommand\footnoteref[1]{\protected@xdef\@thefnmark{\ref{#1}}\@footnotemark}
\makeatother

\usepackage{microtype}
\usepackage{graphicx}
\usepackage{float}
\usepackage{nicefrac,xfrac}
\usepackage{booktabs} 
\usepackage{commath}
\usepackage{bbm}
\usepackage{color}

\usepackage{harpoon}
\definecolor{editcolor}{rgb}{0, 0, 0} 

\newcommand{\edit}[1]{\textcolor{editcolor}{#1}}

\newcommand{\R}{\mathbb{R}}

\newcommand{\prob}{\mathbb{P}}
\newcommand{\N}{\mathbb{N}}

\newcommand{\E}{\mathbb{E}}

\newcommand{\hilb}{\mathcal{H}}

\newcommand{\var}{\text{var}}

\newcommand{\loss}{\mathcal{L}}
\newcommand{\normal}{\mathcal{N}}
\newcommand{\iid}{\overset{\text{i.i.d.}}{\sim}}
\newcommand{\disteq}{\overset{\text{d}}{=}}

\newcommand{\noisevar}{\sigma_{\text{noise}}^2}

\newcommand{\indep}{\perp \!\!\! \perp}
\newcommand{\tens}{%
  \mathbin{\mathop{\otimes}\limits}%
}

\newcommand\inner[2]{\langle #1, #2 \rangle}

\usepackage{mathtools}
\usepackage{lipsum,graphicx,multicol}

\usepackage{newfloat}

\usepackage{mathtools}
\usepackage[capitalize, sort]{cleveref}

\DeclareFloatingEnvironment[fileext=lop, name=Table, within=none]{Table}

\crefname{Table}{Table}{Tables}

\usepackage{thmtools}
\theoremstyle{plain}
\newtheorem{nthm}{Theorem}

\newtheorem{nprop}{Proposition}
\newtheorem{nlem}{Lemma}
\newtheorem{ncor}{Corollary}

\theoremstyle{definition}
\newtheorem{ndefn}{Definition}
\newtheorem{nexa}{Example}
\newtheorem{nassum}{Assumption}

\crefname{nassum}{Assumption}{Assumptions}

\crefformat{footnote}{#1\footnotemark[#2]#3}

\ShortHeadings{The SKIM-FA Kernel}{Agrawal and Broderick}
\firstpageno{1}

\usepackage{lastpage}
\jmlrheading{23}{2022}{1-\pageref{LastPage}}{11/21; Revised
9/22}{11/22}{21-1403}{Raj Agrawal and Tamara Broderick}
\ShortHeadings{The SKIM-FA Kernel}{Agrawal and Broderick}

\begin{document}


\title{The SKIM-FA Kernel: High-Dimensional Variable Selection and Nonlinear Interaction Discovery in Linear Time}

\author{\name Raj Agrawal \email r.agrawal@mit.edu \\
       \addr Department of Electrical Engineering and Computer Science\\
       Massachusetts Institute of Technology\\
       Cambridge, MA 02139-4307, USA
       \AND
       \name Tamara Broderick \email tbroderick@mit.edu \\
       \addr Department of Electrical Engineering and Computer Science\\
      Massachusetts Institute of Technology\\
      Cambridge, MA 02139-4307, USA}

\editor{Daniela Witten}

\maketitle

\begin{abstract}%
Many scientific problems require identifying a small set of covariates that are associated with a target response and estimating their effects. Often, these effects are nonlinear and include interactions, so linear and additive methods can lead to poor estimation and variable selection. Unfortunately, methods that simultaneously express sparsity, nonlinearity, and interactions are computationally intractable --- with runtime at least quadratic in the number of covariates, and often worse. In the present work, we solve this computational bottleneck. We show that suitable interaction models have a kernel representation, namely there exists a "kernel trick" to perform variable selection and estimation in $O$(\# covariates) time. Our resulting fit corresponds to a sparse orthogonal decomposition of the regression function in a Hilbert space (i.e., a functional ANOVA decomposition), where interaction effects represent all variation that cannot be explained by lower-order effects. On a variety of synthetic and real data sets, our approach outperforms existing methods used for large, high-dimensional data sets while remaining competitive (or being orders of magnitude faster) in runtime.
\end{abstract}


\section{Introduction} \label{sec:intro}
 Many scientific and decision-making tasks require learning complex relationships between a set of $p$ covariates and target response from $N$ observed datapoints
with $N \ll p$. For example, in genomics and precision medicine, researchers would like to identify a small set of genetic and environmental factors (out of potentially thousands or millions) associated with diseases and quantify their effects \citep{missing_herit, gwas_interaction, epistatis_gene2, gpu_epistasis}. Estimating these effects can be challenging, however,  without sufficiently flexible models. Blood sugar levels, for example, could vary sinusoidally with the time of day (e.g., depending on when an individual has breakfast, lunch, and dinner). In other instances, effects can be challenging to estimate due to multiplicative interactions. A particular drug could help individuals with certain genetic characteristics but harm others. To learn such nuances in our data for  better decision-making, we need statistical methods that can model nonlinear and interaction effects. We also need computationally efficient methods that can scale to large-$p$ settings. Unfortunately, as we detail below, existing methods suffer in at least one of these three categories.

Sparse linear regression methods (e.g., the Lasso) are typically fast but do not have the flexibility to learn nonlinear or interaction effects \citep{atom_pursuit, dantzig, pruning_lasso}. SpAM extends the Lasso to model nonlinear effects but assumes additive effects \citep{spam}. Conversely, the hierarchical Lasso models interactions but assumes linearity, and its runtime scales quadratically with dimension \citep{lass_heirch}. Recently, \citet{kit} developed a kernel trick to learn interactions in time linear in dimension, but this method assumes linear effects. Black-box models, such as neural networks and random forests, often include interactions and nonlinear effects for the sake of prediction. However, it is  unclear how to access the effects from the  fitted prediction model.

The \emph{hierarchical functional ANOVA} \citep{stone1994}, which includes many of the models
described above as special cases,  provides a general framework to jointly model interactions and nonlinear effects through a variance decomposition of the regression function. As long as the response has finite variance and the covariates vary over a compact set, the functional ANOVA decomposition exists. Such a variance decomposition, which includes classical ANOVA decompositions of contingency tables and generalized additive models as special cases, has been widely used in applications due to desirable interpretability properties. For example, in genetic applications, biologists use ANOVA decompositions to isolate the marginal effects of particular genetic or environmental factors on disease in a population \citep{anova_bio, missing_herit, gwas_interaction}. 
 Unfortunately, existing functional ANOVA methods, which are primarily kernel-regression based, do not scale well with dimension \citep{ss_anova, cosso, supanova}; these methods use kernels that take $O(p^Q)$ time to evaluate, where $Q$ equals the size of the highest order interaction. Hence, running kernel ridge regression for inference takes $O(p^QN^2 + N^3)$ time.  

\noindent \textit{Contributions.} \edit{We consider two interconnected tasks: (1) high dimensional variable selection and (2) estimation of nonlinear additive and interaction effects. We define a new class of kernels called "model selection kernels" to simultaneously solve each of these tasks via kernel ridge regression. Model selection kernels have the flexibility to select a sparse subset of covariates that drive the response, and estimate nonlinear effects among the selected covariates.  However, model selection kernels are  computationally intractable to compute in general. Hence, we propose SKIM-FA
kernels, a type of model selection kernel that also enjoys
computational efficiency.} We show how to compute SKIM-FA kernels in $O(pQ)$ time by exploiting special low-dimensional structure. We motivate this structure from the perspective of hierarchical Bayesian modeling. Then, we use equivalences between kernel ridge regression, Gaussian processes, and conjugate Bayesian  regression to  develop our efficient inference procedure. 
 
 \noindent \textit{Outline.} We start by describing how to model nonlinear interaction effects and encode sparsity using hierarchical Bayesian modeling in \cref{sec:prelims}. In \cref{sec:two_kern_tricks}, we \edit{define model selection kernels and} develop two kernel tricks to perform inference more efficiently when the covariates are independent. Then, we extend our procedure to the general covariate case in \cref{sec:dependent_covariates}. We defer implementation details of our final algorithm to \cref{sec:alg_implement}. We conclude by discussing related work in \cref{sec:related_anova_work} and benchmarking our method against other methods often used to model high-dimensional data in \cref{sec:experiments}. 

\section{A Framework for \edit{Modeling} Nonlinear \edit{Additive and  Interaction Effects} and \edit{Inducing} Sparsity} \label{sec:prelims}
\noindent \textit{Problem statement.}
Suppose we collect data $D = \{(x^{(n)}, y^{(n)})\}^N_{n=1}$ with covariates $x^{(n)} \in \R^p$ and continuous scalar responses $y^{(n)}$. We model $y^{(n)} =  f^*(x^{(n)}) + \epsilon^{(n)}$, where $ \epsilon^{(n)} \iid \normal(0, \noisevar)$, $x^{(n)} \iid \mu$, and the unknown regression function $f^*$ belongs to some class of functions $\hilb$. Using only noisy realizations of $f^*$, we would like to identify which covariates  $f^*$ depends on\footnote{\edit{When the derivative exists, this set equals
 equals all covariates with non-zero derivatives; that is, all $x_j$ with $j \in \{1,\ldots,p\}$ such that $\| \partial f^* / \partial x_j \| \neq 0$. See also \cref{lem:kern_var_select}.}} (i.e., perform variable selection), and recover \edit{main effects and} interaction effects. For example, a biologist might seek to identify a small set of genes \edit{(out of tens of thousands of possible genes)} associated with a disease — e.g., to design new gene-based drug targets. Understanding the relationship between the selected genes and disease response could help the biologist properly administer the drug.  

\edit{ To perform variable selection and estimation, existing methods often assume the majority of effects equal zero. Problematically, when there are interactions present, sparsity in effects does not guarantee that a sparse set of covariates is selected. For example, suppose we include all additive and pairwise interaction effects in our model. A method that selects $p$ non-zero effects might be considered sparse in the effects since $p \ll p^2$. But the selected effects could correspond to $p$ (or nearly $p$) selected covariates, so the burden of data collection is not reduced relative to the original problem; see \citet{lass_heirch} for further
motivation of sparsity in covariates, by contrast to sparsity in
effects. To ensure sparsity in covariates, many interaction methods impose a "hierarchy" or "heredity" constraint \citep{lass_heirch, radchenko_var, haris_herid}. Such a constraint allows interactions to be present only among selected additive effects. If the additive effects are weak, then this constraint will lead to poor inference; see also our extended discussion in \cref{A:lit_review}.  }


 To perform variable selection and estimation \edit{without requiring a heredity constraint}, we use penalized regression:
\begin{equation} \label{eq:pen_ls}
	\hat{f} = \arg \min_{f \in \hilb} \sum_{n=1}^N \loss(y^{(n)}, f(x^{(n)})) + J(f),
\end{equation}
where $ \loss(\cdot, \cdot)$ and $J(f)$ denote some loss function and penalty on model complexity, respectively. This paper focuses on four subproblems resulting from  \cref{eq:pen_ls}: (P1) picking $\hilb$ to model interactions, (P2) selecting $\loss(\cdot, \cdot)$ and $J(f)$ to induce sparsity (i.e., to identify the small subset of covariates that influences the response), (P3) tractably solving \cref{eq:pen_ls} for our choice of sparsity-inducing $J(f)$, and (P4) efficiently reporting effects in $\hat{f}$.

\subsection{Our Contributions: An Overview} 
We describe, at a high level, our solutions to subproblems P1 through P4, and what parts of our solutions are new. Our solutions to P3 and P4 are our core contributions. \\

\noindent \textit{P1: Constructing $\hilb$.} Our construction of $\hilb$ in \cref{sec:construct_space} is based on \citet{projection_anova}. We use the hierarchical functional ANOVA introduced in \citet{stone1994} to make recovering interaction effects a well-defined inference task (i.e., statistically identifiable).  \\

\noindent \textit{P2: Selecting the loss and penalty.} We select the loss and penalty from a hierarchical Bayesian modeling point of view in \cref{sec:general_prior_class}. In particular, we choose the loss $\loss$ to correspond to a negative log-likelihood of the data and the penalty $J(f)$ to correspond to a negative log prior of $f$. Existing sparse Bayesian interaction methods do not work at our level of generality. Nevertheless, our proposed class of priors is heavily influenced by existing sparse Bayesian interaction models. \\

\noindent \textit{P3: Solving \cref{eq:pen_ls}.} We solve \cref{eq:pen_ls} in time linear in $p$ by using two kernel tricks to (1) reduce the cost of modeling nonlinear functions and (2) avoid summing over a combinatorial number of interaction terms. The first kernel trick, described in \cref{sec:two_kern_tricks}, is based on the foundational smoothing spline ANOVA (SS-ANOVA) work by  \citet{ss_anova}. To make the connection to SS-ANOVA, we show that there exists a duality between our class of hierarchical models (see P2) and reproducing kernel Hilbert spaces induced by what we call \emph{model selection kernels}. Our model selection kernels generalize the kernels used in \citet{ss_anova} by removing the requirement that all covariates be independent. Our second kernel trick, which allows us to avoid summing over a combinatorial number of interaction terms,  applies to a subset of model selection kernels. We call this subset of kernels \emph{SKIM-FA} kernels and prove that SKIM-FA kernels have desirable statistical properties from the hierarchical Bayesian modeling point of view.  \\

\noindent \textit{P4: Reporting effects.} For the case of independent covariates, we report effects using the procedure in \citet{ss_anova}. Our new contribution, provided in \cref{sec:dependent_covariates}, is developing an efficient algorithm to report effects for the non-independent case. 

\subsection{Interactions and Identifiability for Nonlinear Functions} \label{sec:construct_space}
We construct $\hilb$ in \cref{eq:pen_ls} by considering functions on $\R^p$ that can be written as a sum of lower-dimensional functions (i.e., interaction effects) that depend on at most $Q$ covariates with $Q < p$. Our goal is to estimate these interaction effects. Unfortunately, as we detail below, such an expansion is not unique, and therefore not a valid target of inference. To make inference over $\hilb$ well-defined, we use the \emph{hierarchical functional ANOVA} \citep{stone1994}. \\  

\noindent \textit{Modeling Interactions.}  Let $\hilb =\hilb_Q \coloneqq \bigoplus_{V: |V| \leq Q} \hilb_V$, where $\hilb_V$ belongs to the space of all square-integrable functions of $x_V$  (with respect to the probability measure $\mu$) and $V \subset [p] \coloneqq \{1, \cdots, p \}$. Then, for $f_{\emptyset}$ in the space of constant functions $\hilb_{\emptyset} = \{\theta:  \theta \in \R \}$,
\begin{equation} \label{eq:general_mod_space}
	\begin{split}
		\bigoplus_{V: |V| \leq Q}  \hilb_V &= \left \{ f: f = \sum_{V: |V| \leq Q} f_V(x_V), \ f_V \in \hilb_V \right \} \\
		& = \left \{ f: f= f_{\emptyset} + \sum_{i=1}^p f_{\{i\}}(x_i) +  \sum_{i<j}^p f_{\{i, j\}}(x_i, x_j) + \cdots + \sum_{V: |V| = Q} f_{V}(x_V)  \right\}.
	\end{split}
\end{equation}
Similar to additive models, $f_{\{i\}}(x_i)$  has the interpretation as the main or marginal effect of covariate $x_i$ on $y$. Similarly, $ f_{\{i, j\}}(x_i, x_j)$ has the interpretation as the two-way or pairwise effect of $x_i$ and $x_j$ on $y$. Unfortunately, the components in \cref{eq:general_mod_space} are not identifiable without further constraints.  For example, if $f^*(x) = f_{\{1\}}(x_1) + f_{\{2\}}(x_2) + f_{\{1,2\}}(x_1, x_2)$, then $f^*$ also decomposes as $f_{\{1\}}(x_1) + [f_{\{2\}}(x_2) + 5] + [f_{\{1,2\}}(x_1, x_2) - 5]$. \\ 

\noindent \textit{Identifiability with the Functional ANOVA.} To resolve identifiability issues, we construct a smaller space of functions $\hilb_V^o \subset \hilb_V$, where  $\hilb_V^o $ includes only functions whose variation cannot be explained by lower-order effects of $x_V$:
\begin{equation} \label{eq:orth_hilb}
	\hilb_V^o= \{f_V \in \hilb_V: \forall A \subsetneq V,  \ \forall f_A \in \hilb_A,  \  \inner{f_V}{f_A}_{\mu} = 0 \},
\end{equation}
where $\inner{\cdot}{\cdot}_{\mu}$ is an inner product on $L^2$. That is,  $\inner{f_A}{f_B}_{\mu} = \E_{x \sim \mu}[f_A(x_A) f_B(x_B)]$.
\begin{nthm} \citep{stone1994, projection_anova} \label{thm:unique_fanova}
Suppose $f \in \hilb_Q$ and $\mu$ is absolutely continuous with respect to Lebesgue measure. Further, suppose that the domain of functions in $\hilb_Q$ is $\mathcal{X}$, and $\mathcal{X}$ is a compact set of $\R^p$.  Then, there exist ($\mu$-almost everywhere) unique functions $f_V \in  \hilb_V^o$ such that $f =  \sum_{V: |V| \leq Q} f_V$.
\end{nthm}
\begin{ndefn} \label{def:fanova}
Suppose $f = \sum_{V: |V| \leq Q} f_V$ where $f_V \in \hilb_V^o$. Then, $ \sum_{V: |V| \leq Q} f_V$  is called the \emph{functional ANOVA decomposition} of $f$ with respect to $\mu$.
\end{ndefn}
In light of \cref{thm:unique_fanova}, we assume compactness throughout to have a well-defined target of inference (i.e., the functional ANOVA decomposition of $f$ in \cref{def:fanova}). By the orthogonality constraints in \cref{eq:orth_hilb}, the effect $f_{\{i,j\}}(x_i, x_j)$ in \cref{def:fanova} represents, for example, the variation that cannot be explained by 1D functions of $x_i$ and $x_j$ and an intercept. When the covariates are independent, then the signal variance decomposes as 
\begin{equation} \label{eq:sig_var_decomp}
	\begin{split}
		\var(f) = \var(f_{\{\emptyset\}}) + \sum_{i} \var( f_{\{i\}}) + \sum_{i, j} \var( f_{\{i, j\}}) + \cdots \var(f_{\{1,2, \cdots, p\}}(x_1, \cdots, x_p)),
	\end{split}
\end{equation}
where $\var(f) = \inner{f}{f}_{\mu}$. Hence, \cref{eq:sig_var_decomp} allows us to \emph{analyze} how the \emph{variance} of the \emph{function} is distributed across the interactions of different orders. Hence, the name functional analysis of variance or functional ANOVA. When all the covariates are categorical, then the functional ANOVA reduces to the classical ANOVA decomposition of a contingency table.

\subsection{How to Achieve Sparsity for Nonlinear Functions} \label{sec:general_prior_class}
To complete our specification of \cref{eq:pen_ls}, we still need to pick a loss and penalty function on $\hilb_Q$. We motivate our choice of loss and penalty from a Bayesian point of view. That is, we view $\loss(\cdot, \cdot)$ as the negative log-likelihood function, $J(f)$ as the negative log prior on $f$, and $\hat{f}$ as the maximum a priori (MAP) estimate under our proposed Bayesian model. \\

\noindent \textit{Our loss.} Since the noise terms are Gaussian (see "Problem Statement" \edit{at the start of} \cref{sec:prelims}), the negative log-likelihood is quadratic: $\loss(y, f(x)) = (y - f(x))^2$  (i.e., squared-error loss). \\

\noindent \textit{Our penalty.} We are primarily interested in the case when $f^*$ is sparse, i.e., when $f^*$ depends on a small number of covariates. So $J(f)$ should promote such sparsity. To that end, we first take a basis expansion of each component space, and then place a sparsity prior on the basis coefficients. We assume that for all  $V \subset [p]$ and $1 \leq |V| \leq Q$, there exists a $B_V \in \N \cup \{\infty\}$ and feature map $\Phi_V: \R^{|V|} \mapsto \R^{B_V}$ such that the components of $ \Phi_V$ form a basis of $\hilb_V^o$. Then, for any $f_V \in \hilb_V^o$, there exists $\Theta_{V} \in \R^{B_V}$ such that $f_V(x_V) = \Theta_{V}^T \Phi_V(x_V)$. Hence, if we can estimate $\Theta_V$, we can estimate the functional ANOVA decomposition of $f^*$ by \cref{thm:unique_fanova}. 
 
To obtain a MAP estimate of $\Theta_V$, we draw each $\Theta_V \sim \normal(0,  \theta_{V} \cdot I_{B_V})$, where $\theta_V \in \R$ is a non-negative auxiliary parameter drawn from a sparsity prior (e.g., a Laplace prior) \edit{and $I_{B_V}$ denotes the $B_V \times B_V$ identity matrix}; see \cref{sec:alg_implement} for our particular choice of prior. If $\theta \coloneqq \{\theta_V\}$ is sparse, then we claim that $\{f_V \}$ is sparse. To understand why, suppose $\theta_V = 0$. Then, the prior variance of $f_V$ equals 0. Hence, $f_V$ will equal 0. Thus, a prior that induces sparsity in $\theta$ enables us to get sparsity in the number of effects selected. However, sparsity in \emph{effects} does not automatically guarantee that a sparse subset of \emph{covariates} is selected, \edit{as we discussed at the start of \cref{sec:prelims}.} 

To get sparsity in covariates  \edit{without requiring a heredity constraint}, we draw $\theta$ from a hierarchical sparsity prior; see \cref{sec:trick_two} for details. Since $\Phi_V$ is a basis of $\hilb_V^o$ and our prior on $\Theta_V$ has full support on $\R^{B_V}$, our choice of likelihood and prior allows us to model any $f \in \hilb_Q$ as summarized below: 
\begin{subequations} \label{eq:gen_bayes_model}
	\begin{align}
			\Theta_{V} \mid \theta_{V} &\sim \normal(0,  \theta_{V} \cdot \edit{I_{B_V}}), \ V \subset [p], \ |V| \leq Q, \ \theta_V \geq 0 \label{eqn:top_gen_bayes_model}  \\
				y^{(n)} \mid x^{(n)}, \Theta, \noisevar &\sim \normal(f(x^{(n)}), \noisevar), \ f = \sum_{V: |V| \leq Q} \Theta_V^T \Phi_V(\cdot), \ n \in [N], \label{eqn:bot_gen_bayes_model}
	\end{align}
\end{subequations}
 where the likelihood in the first equation corresponds to $\exp(-J(f))$ and $\exp(-\loss(y^{(n)}, f(x^{(n)})))$ corresponds to the likelihood in the last equation.\footnote{\edit{While \cref{eq:gen_bayes_model} has the flexibility to induce sparsity in both covariates and effects, it does not lead to sparsity in the basis expansion of a selected effect (e.g., if $f_V$ is selected, then $\Theta_V$ will be a dense vector with probability one). Since our goal is not to learn a sparse representation of $f_V$, our choice of a Gaussian prior (or $L_2$ regularization) is not very limiting since irrelevant basis components will just be shrunk close to zero. However, if there are many irrelevant basis components in $\Phi_V$, then a Laplace prior (or $L_1$ penalty) might be preferable. Other methods, such as sparse additive models, also use $L_2$ regularization to penalize the basis expansion coefficients \citep{spam}.}} While other likelihoods and priors exist to model interactions and sparsity, many existing sparse Bayesian methods are instantiations of \cref{eq:gen_bayes_model}, and have desirable statistical shrinkage properties; see, for example,  \citet{sparse_bayes_pairwise,Curtis13fastbayesian, heirc_sparisty, kit, chipman_bayes_glm, george1993variable}. In the next section, we exploit the special Gaussian and interaction structure in \cref{eq:gen_bayes_model} for faster inference.

\section{Using Two Kernel Tricks to Reduce Computation Cost} \label{sec:two_kern_tricks}
In principle, we can analytically compute the MAP estimate of $\Theta_V$ in \cref{eq:gen_bayes_model} (and hence solve \cref{eq:pen_ls} in closed-form);  conditional on $\theta$, \cref{eq:gen_bayes_model} reduces to conjugate Bayesian regression. Unfortunately, unless $p$ is very small or $Q=1$, computing this closed-form solution is typically computationally intractable, for reasons we describe next. To remedy this computational intractability, we show how to make inference scale linearly with $p$ by exploiting special model structure in  \cref{sec:trick_one} and \cref{sec:trick_two}.\\ 
 
\noindent \textit{Intractability of Conjugate Bayesian Regression.} Our model in \cref{eq:gen_bayes_model} has $B_Q \coloneqq \sum_{V: |V| \leq Q} B_V$ parameters. In general, computing the MAP estimate of these $B_Q$ parameters requires inverting a $B_Q \times B_Q$ covariance matrix \citep[Chapter 2]{gp_book}. So the computational cost of MAP inference scales as $O(B_Q^3 + NB_Q^2)$. $B_Q$ may be prohibitively large for two reasons. First, $B_Q$ is large if any basis-expansion size (i.e., any $B_V$) is large. For example, if $\hilb_V^o$ is infinite-dimensional (e.g., if $\hilb_V$ equals the space of all square-integrable functions of $x_V$), then $B_V = \infty$. Even if all the $\hilb_V^o$ are finite-dimensional (e.g., if $\hilb_V$ is generated from a finite polynomial basis), $B_V$ typically grows exponentially with the size of $|V|$; see, e.g., \citet{projection_anova}. $B_Q$ may also be large due to the combinatorial sum over interactions; even if all of the $B_V$ equal 1, $B_Q$ still has on the order of $O(p^Q)$ terms. Hence, without additional structure, the computation time for conjugate Bayesian regression is lower bounded by $\Omega(p^{3Q} + p^{2Q}N)$. Fortunately, due to unique structure in our problem, we show how to avoid the cost of explicitly generating the basis expansion ("Trick 1" in \cref{sec:trick_one}), and summing over all $O(p^Q)$ interactions ("Trick 2" in \cref{sec:trick_two}). In what follows, we assume $\theta$ is fixed. Then, we show how to estimate $\theta$ in \cref{sec:alg_implement}.  

\subsection{Trick 1: Represent and Access Sparsity Without Basis Expansion} \label{sec:trick_one}

We show how to remove the computational dependence on the size of $B_V$ through a kernel trick. Our kernel generalizes the one used in \cite{ss_anova}, which assumes independent covariates, to the case of general covariate distributions. In order to prove the existence of a kernel trick, we make the following assumption:
\begin{nassum} \label{assum:base_rkhs}
Each $\hilb_V$ is a reproducing kernel Hilbert space (RKHS). 
\end{nassum}
Given that there exist reproducing kernels that can approximate any continuous function arbitrarily well, \cref{assum:base_rkhs} is a mild condition \citep{univ_kernel}. The non-trivial part is proving the existence of a kernel to induce $\hilb_V^o$, which is not immediate due to the orthogonality constraints in \cref{eq:orth_hilb}. 

\begin{nprop} \label{prop:kern_exist}
(existence of a kernel trick) Under \cref{assum:base_rkhs}, there exists a positive-definite kernel $k_V$ such that $k_V(x, \tilde{x}) = \inner{\Phi_V(x)}{\Phi_V(\tilde{x})}$, where the components of $\Phi_V \in \R^{B_V}$ form a countable basis of $\hilb_V^o$. 
\end{nprop}

We prove \cref{prop:kern_exist} in \cref{A:proof_exists}. In \cref{sec:trick_two}, we show how to efficiently evaluate $k_V$ without explicitly computing the feature maps.  In light of \cref{prop:kern_exist}, we introduce \emph{model selection kernels} to rewrite the model in \cref{eq:gen_bayes_model} as a Gaussian process. We then show how this reparametrization allows us to perform inference more efficiently. 
\begin{ndefn} \label{def:model_select_kernel}
A kernel $k_{\theta}$ is a \emph{model selection kernel} if it can be written as $\sum_{V: |V| \leq Q} \theta_V k_V$, where $k_V$ is the reproducing kernel for $\hilb_V^o$ and $k_{\emptyset}(x, \tilde{x}) = 1$ (i.e., the kernel $k_{\emptyset}$ induces the space of constant functions $\hilb_{\emptyset})$.
\end{ndefn}
\begin{nlem} \label{prop:model_select_bayes} 
\edit{Let $\{y^{(n)} \}_{n=1}^N$ be generated according to the model in \cref{eq:gen_bayes_model}. Suppose that $\{\tilde{y}^{(n)} \}_{n=1}^N$ is generated according to the model below: 
\begin{equation}
	\begin{split}
		f &\sim GP(0, k_{\theta}) \\
		\tilde{y}^{(n)} \mid f, x^{(n)} &\sim \normal(f(x^{(n)}), \noisevar), \quad n \in [N],
	\end{split}
\end{equation}
where $k_{\theta}$ is defined in \cref{def:model_select_kernel}. Then,  $\{y^{(n)} \}_{n=1}^N \mid X \disteq \{\tilde{y}^{(n)} \}_{n=1}^N \mid X$ , where $\disteq$ denotes equality in distribution.}
\end{nlem}
Based on the reparametrization in \cref{prop:model_select_bayes} (see \cref{A:reparam_proof} for the proof), $J(f)$ equals the penalty induced by the kernel $k_{\theta}$. Hence, the solution to \cref{eq:pen_ls} reduces to kernel ridge regression (or equivalently equals the posterior predictive mean of the Gaussian process) by \citet[Chapter 2]{gp_book}:
\begin{equation} \label{eq:kern_ridge}
	\hat{f}(x) = \bar{f}_{\theta}(x) \coloneqq \sum_{n=1}^N \hat{\alpha}_n k_{\theta}(x_n, x), \quad \hat{\alpha} =  (K_{\theta} + \noisevar I_{N \times N})^{-1} Y,
\end{equation}
where $Y$ is a column vector with $n$th component $Y_n = y^{(n)}$ and $[K_{\theta}]_{nm} = k_{\theta}(x^{(n)}, x^{(m)})$. 

Unlike the "weight-space view" in \cref{sec:general_prior_class} where $f_V = \Theta_V^T \Phi_V(\cdot)$, it is not clear how to actually recover the effects $f_V$ from the prediction function $\bar{f}_{\theta}$. For general kernels, accessing $f_V$ (and consequently computing the functional ANOVA of  $\bar{f}_{\theta}$) lacks an analytical form. Fortunately, we can easily recover $f_V$ from $\bar{f}_{\theta}$  for model selection kernels: 
\begin{nlem}  \label{lem:kern_select}
Let $k_{\theta}$ be a model selection kernel and $f^{(M)}(x) = \sum_{m=1}^M \alpha_m k_{\theta}(x_m, x)$ for $\alpha_m \in \R$ and $x_m \in \R^p$. Then,  $f^{(M)}(x)  = \sum_{V: |V| \leq Q} f_V$, where $f_V =  \theta_V \sum_{m=1}^M \alpha_m k_{V}(x_m, x) \in \hilb_V^o$.
\end{nlem}
It follows from \cref{lem:kern_select} that model selection kernels enable easy variable selection; we just need to examine the sparsity pattern of $\theta$. For general kernels,  we would need to search over the entire domain of the fitted regression function (a $p$-dimensional space) to perform variable selection.
\begin{ncor} (nonlinear variable selection)  \label{lem:kern_var_select}
Suppose $f^{(M)}(x)  = \sum_{m=1}^M \alpha_m k_{\theta}(x_m, x)$. Then, $f^{(M)}(x)$ functionally depends on the set of covariates $\{i: \exists V \subset [p], i \in V \text{ s.t. } \theta_V \neq 0\}$.
\end{ncor}

We have avoided the cost of generating the basis expansion to solve \cref{eq:pen_ls}, but \cref{eq:kern_ridge} is still computationally intractable; $k_{\theta}$ sums over $O(p^Q)$ kernels. Hence, the cost to compute the kernel matrix $K_{\theta}$ and invert $ (K_{\theta} + \lambda I_{N \times N})$  takes $O(N^2p^Q)$ and $O(N^3)$ time, respectively. 

\subsection{Trick 2: A Recursion to Avoid a Combinatorially Large Summation Over Interactions Given Covariate Independence}  \label{sec:trick_two}

We show how to compute $k_{\theta}$ in $O(pQ)$ time (and hence solve  \cref{eq:kern_ridge} in $O(pQN^2 + N^3)$ time) for a particular subset of model selection kernels that we call \emph{SKIM-FA} kernels.  
In what follows, we start by motivating SKIM-FA kernels from the hierarchical Bayesian characterization of model selection kernels in \cref{eq:gen_bayes_model}. We show that, when the covariates are independent, we can compute SKIM-FA kernels much more efficiently using a second kernel trick.  In \cref{sec:dependent_covariates}, we generalize to the non-independent covariate case by building on the procedure described in this section. \\

\noindent \textit{The Sparse Kernel Interaction Model for Functional ANOVA (SKIM-FA).}  \edit{To motivate SKIM-FA, suppose for the moment we know what covariates $f$ functionally depends on. Let the binary vector $\kappa \in \{0, 1\}^p$ encode this knowledge, where $\kappa_j = 0$ if and only if $f$ does not depend on covariate $x_j$. Then, $f_V = 0$ if there there  exists some $j \in V$ such that $\kappa_j = 0$. Equivalently, it suffices to check that $\prod_{j \in V} \kappa_j = 0$. Hence, if we knew $\kappa$, we should adjust our prior variance for $f_V$ from $\theta_V$ to $\theta_V \prod_{j \in V} \kappa_j$. This update to the prior enforces support on only selected covariates, and is the same as fitting the model without the irrelevant covariates included. Since we do not know $\kappa$ in advance, however, we treat $\kappa$ as learnable parameter, and use it to (1) induce sparsity in effects via the product structure above, and (2) perform variable selection. We propose the following prior on $\Theta_V$ in \cref{eq:gen_bayes_model}, which generalizes the prior for linear pairwise interaction models in \citet{kit}}:
\begin{equation} \label{eq:skim_fa_weight}
	\begin{split}
			\Theta_{V} \mid \eta, \kappa &\sim \normal\left(0,  \eta_{|V|}^2 \prod_{i \in V} \kappa_i^2 \cdot I_{B_V \times B_V} \right),
	\end{split}
\end{equation}
for non-negative random vectors $\kappa \in \R^p_+$ and $\eta \in \R^{Q + 1}_+$. We do not restrict $\kappa \in \{0, 1\}^p$ so that we can leverage gradient-based techniques for learning $\kappa$ more easily; see \cref{sec:alg_implement} for details. \\ 


\noindent \textit{SKIM-FA interpretation.}  In \cref{eq:skim_fa_weight}, $\eta_{|V|}^2$ quantifies the overall strength of $|V|$-way interactions by modifying the prior variance of all effects of order $|V|$. \edit{Hence, $\eta_{|V|}$ plays an analogous role to the "global scale" in sparse Bayesian linear models; see, for example, \citet{finnish_prior, horse_prior, kit}}. $\kappa_i$ plays the role of a "variable importance" measure for covariate $x_i$ by affecting the prior variance of all effects involving covariate $x_i$. Hence, if it turns out an effect involving $x_i$ is strong, the posterior of $\kappa_i$ will place high probability at large values (i.e., indicating that covariate $x_i$ has high "importance"). Notice that if $\kappa_i = 0$, then the prior variance of $\Theta_V$ equals 0 whenever $i \in V$. Consequently, all effects involving $x_i$ will equal 0. Hence, we can perform variable selection in $O(p)$ time by just examining the sparsity pattern of $\kappa$ instead of in $O(p^Q)$ time using \cref{lem:kern_var_select}. In \cref{sec:alg_implement}, we show how we select our sparsity prior on $\kappa$.  Finally, note that while we added more structure to the prior, we have not lost modeling flexibility; as long as $\prob(\kappa_i > 0)$ does not equal $0$, then the prior variance of $\Theta_V$ will be non-zero.\footnote{\edit{SKIM-FA considers all interactions of order $Q$ among selected covariates (i.e., does not assume sparsity in interactions between selected covariates).  Specifically, suppose $\{x_1, x_2, x_3\}$ are selected and $Q=2$. Then, SKIM-FA considers all additive and pairwise effects between the first three covariates. If $f_{\{1, 2\}} = 0$, for example, then the posterior of $f_{\{1, 2\}}$ is non-zero because the prior on $\Theta_{\{1, 2\}}$ is drawn from a Gaussian distribution with non-zero variance. However, as $N$ increases, the posterior of $f_{\{1, 2\}}$ will be close to 0.}} Hence, our prior will have support on all of $\hilb_Q$.




\begin{ndefn}  \label{def:skim_kern} A \emph{SKIM-FA kernel} is a  model selection kernel that can be written as 
\begin{equation*}
	k_{\text{SKIM-FA}}(x, \tilde{x}) = \sum_{V: |V| \leq Q} \left[  \eta_{|V|}^2 \prod_{i \in V} \kappa_i^2 \right] k_V(x, \tilde{x}).
\end{equation*}
for some $\kappa \in \R^p$ and $\eta \in \R^{Q + 1}$. 
\end{ndefn}

\begin{nprop} \label{cor:skim_kern_ridge_equiv} For a SKIM-FA kernel, \cref{eqn:top_gen_bayes_model} can be replaced by \cref{eq:skim_fa_weight} in \cref{prop:model_select_bayes}. 
\end{nprop}

\begin{proof}
Set $\theta_V =  \eta_{|V|}^2 \prod_{i \in V} \kappa_i^2$ in \cref{prop:model_select_bayes}.
\end{proof}


\noindent \textit{Efficient evaluation of SKIM-FA kernels.} Recall that $k_i$ is the reproducing kernel for $\hilb_i^o$. Suppose, for the moment, that the reproducing kernel $k_V$ for $\hilb_V^o$ equals $\prod_{i \in V} k_i$; we will shortly show that this condition holds when the covariates are independent. Then, by \cref{thm:skim_fast_form} and \cref{corr:skim_kern_time} below, we can compute SKIM-FA kernels orders of magnitude faster by not explicitly summing over all $O(p^Q)$ interactions in \cref{def:skim_kern}. 
\begin{nthm} \label{thm:skim_fast_form}
Suppose  $k_V(x, \tilde{x}) = \prod_{ i \in V} k_i(x_i, \tilde{x}_i)$. Then,  
\begin{equation} \label{eq:skim_fast_form}
	\begin{split}
	& k_{\mathrm{SKIM-FA}}(x, \tilde{x}) = \sum_{q=1}^Q \eta_q^2 \bar{k}_q(x, \tilde{x}) \quad \mathrm{s.t.} \\
	& \bar{k}_q(x, \tilde{x}) = \frac{1}{q} \sum_{s=1}^q (-1)^{s+1} \bar{k}_{q-s}(x, \tilde{x}) k^s(x, \tilde{x}),  \quad \bar{k}_0(x, \tilde{x}) = 1, \\
	& k^s(x, \tilde{x}) = \sum_{i=1}^p \kappa_i^{2s} [k_i(x_i, \tilde{x}_i)]^s.
	\end{split}
\end{equation}
\end{nthm}
As we show in \cref{A:kern_trick_two_proof}, the key to proving \cref{thm:skim_fast_form} is an old recursive kernel formula provided in \citet[pg.\ 199]{vapnik95}. From \cref{thm:skim_fast_form}, we have two corollaries. The first requires a short inductive argument; see \cref{A:recur_proof}. The second follows immediately by setting $Q=2$ into \cref{eq:skim_fast_form}.

\begin{ncor} \label{corr:skim_kern_time}
$k_{\mathrm{SKIM-FA}}(x, \tilde{x})$ takes $O(pQ)$ time to evaluate on a pair of points.
\end{ncor}

\begin{ncor}
Suppose $Q=2$. Then, $k_{\text{SKIM-FA}}(x, \tilde{x})$ equals
\begin{equation}
	0.5 \eta_2^2 \left [\left( \sum_{i=1}^p \kappa_i^2 k_i(x_i, \tilde{x}_i) \right)^2 - \sum_{i=1}^p \kappa_i^4 [k_i(x_i, \tilde{x}_i)]^2  \right] + \eta_1^2 \sum_{i=1}^p \kappa_i^2 k_i(x_i, \tilde{x}_i) + \eta_0^2.
\end{equation}
\end{ncor}
To see why "trick 2" in \cref{eq:skim_fast_form} indeed acts as another kernel trick, consider the linear interaction case when $\hilb_Q$ consists of interactions of the form $\prod_{i \in V} x_i$. Suppose further that $\kappa$ and $\eta$ are equal to the ones vector. Then, $k_{\mathrm{SKIM-FA}}(x, \tilde{x}) = \sum_{V: |V| \leq Q} \prod_{i \in V} x_i \tilde{x}_i$, which explicitly generates and sums over the interactions $\prod_{i \in V} x_i$. However, it is well known that polynomial kernels implicitly generate interactions, and hence can be used instead to avoid summing over all interactions. The core idea in \cref{eq:skim_fast_form} is similar; the kernel $k^s$, which sums over kernels raised to the $s$ power in \cref{eq:skim_fast_form}, implicitly generates interactions of order equal to $s$ just like a polynomial kernel. However, instead of generating interactions of the form $\prod_{i \in V} x_i$, $k^s$ operates on one-dimensional kernels $k_i$, where its product with $\bar{k}_{q-s}$ generates interactions of the form $\prod_{i \in V} k_i$. Since $k_V = \prod_{i \in V} k_i$ by assumption, these "interactions" of kernels span $\hilb_V^o$ by the product property of kernels. 

To understand when  $k_V = \prod_{i \in V} k_i$ in \cref{thm:skim_fast_form} holds, we provide sufficient conditions based a result from \citet{ss_anova}. We leave our construction of $k_i$ to \cref{A:model_implement_details}. 
\begin{nassum} \label{assum:tens_space} (Tensor product space) For all $V \subset [p]$ and $1 \leq |V| \leq Q$, $\hilb_V =  \bigotimes_{i \in V} \hilb_i$.
\end{nassum}
\begin{nprop} \label{prop:ss_anova_construct} \citep{ss_anova} Suppose $\mu = \mu_{\tens}$, where $\mu_{\tens}(x) \coloneqq \mu_1(x_1) \tens \mu_2(x_2) \cdots \tens \mu_p(x_p)$ and $\mu_j$ is the marginal distribution of $x_j$. Then, under \cref{assum:base_rkhs} and \cref{assum:tens_space}, $k_V = \prod_{i \in V} k_i$.
\end{nprop}

Since any Hilbert space of square-integrable functions of $x_V$ can be approximated arbitrarily well by taking tensor products of  one-dimensional Hilbert spaces by \citet{stone1994, projection_anova}, \cref{assum:tens_space} is a mild assumption. The more problematic assumption is that all covariates are independent (i.e., that $\mu = \mu_{\tens}$).

\section{How to Get Sparsity, Interactions, and Fast Inference When Covariates Are Dependent} \label{sec:dependent_covariates}
Here we extend to the general $\mu$ case. We start by motivating this extension in \cref{sec:issue_prod_measure}. In  \cref{sec:change_basis}, we develop a change-of-basis formula to take the functional ANOVA decomposition of $\bar{f}_{\theta}$ with respect to $\mu_{\tens}$ to one with respect to $\mu$. In this section, we assume $Q=2$. We defer the general $Q$ case to \cref{A:skim_extend}. 

\subsection{Practical Problems From Assuming Independent Covariates} \label{sec:issue_prod_measure}
Since $\mu_{\tens} = \mu_1(x_1) \tens \mu_2(x_2) \cdots \tens \mu_p(x_p)$,  $\mu_{\tens}$ has the same 1D marginal distributions as  $\mu$. Nevertheless, we prove that the functional ANOVA decomposition of an $f \in \hilb$ can be \emph{arbitrarily} different depending on if $\mu_{\tens}$ or $\mu$ is selected. We prove this claim by showing something stronger, namely that the intercepts between two functional ANOVA decompositions can be arbitrarily far apart (see \cref{A:far_intercept_proof} for the proof of \cref{lem:arb_far_intercept}).

\begin{nprop} \label{lem:arb_far_intercept}
For any $\Delta > 0$, there exists a probability measure $\mu$ and square-integrable $f$ such that the relative difference
\begin{equation*}
	\frac{| f_{\emptyset}^{\mu} - f_{\emptyset}^{\mu_{\tens}} |}{ |f_{\emptyset}^{\mu}|} > \Delta, \quad \text{where} \quad f_{\emptyset}^{\mu} =  \E_{\mu} f(X) \quad \text{and} \quad f_{\emptyset}^{\mu_{\tens}} =  \E_{\mu_{\tens}} f(X).
\end{equation*}
\end{nprop}
\begin{figure}
        \centering
        \begin{subfigure}[b]{0.29\textwidth}
            \centering
            \includegraphics[width=\textwidth]{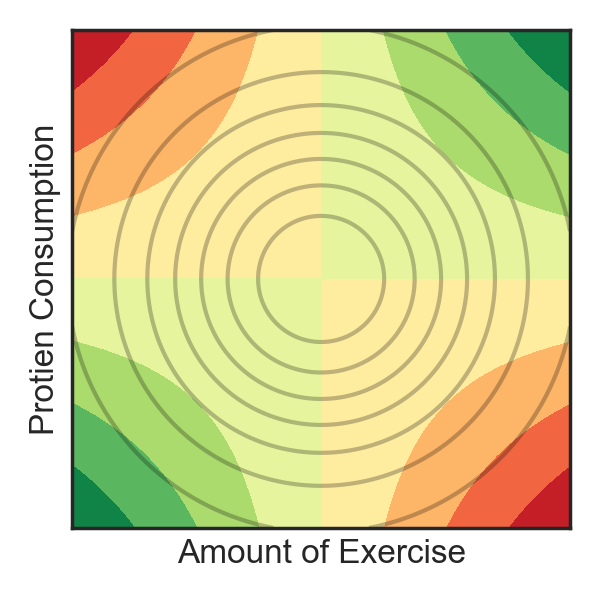}
            \caption{Product measure} \label{fig:toy_prod_measure_issue}  
        \end{subfigure}
        \begin{subfigure}[b]{0.29\textwidth}  
            \centering 
            \includegraphics[width=\textwidth]{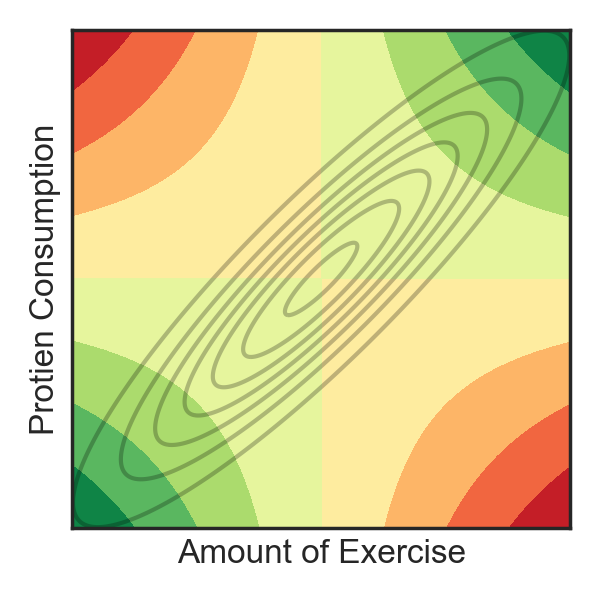}
		 \caption{Covariate measure}  \label{fig:toy_prod_measure_issue_two}  
        \end{subfigure}
        \begin{subfigure}[b]{0.39\textwidth}  
            \centering 
            \includegraphics[width=\textwidth]{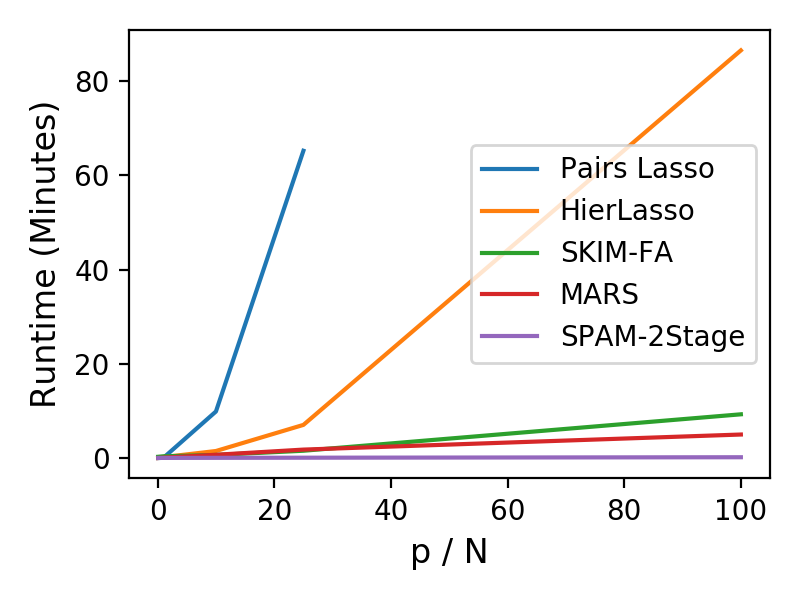}
            \caption{Runtimes on Simulated Data} \label{fig:syn_runtime}
        \end{subfigure}
                	\caption{\emph{Left and middle}: the colors denote the contour plot of the function $f^*(x_1, x_2) =100 x_1 x_2 - 50$. Darker green indicates larger positive values while darker red indicates larger negative values. The gray solid lines in the left and right hand figures represent the density contours of $\mu_{\tens}$ and $\mu$ in \cref{ex:prod_measure_issue}, respectively. \emph{Right}:  runtime comparisons of different methods as $p / N$ increases; see \cref{sec:experiments} for details.}
\end{figure}
To build intuition for \cref{lem:arb_far_intercept}, and motivate using $\mu$ instead of $\mu_{\tens}$ to compute the functional ANOVA decomposition of $\bar{f}_{\theta}$, consider the following toy example.
\begin{nexa} \label{ex:prod_measure_issue}
Suppose $f^*(x_1, x_2) = 100 x_1 x_2 - 50$, where $x_1$ could represent exercise, $x_2$ protein consumption, and $f^*(x_1, x_2)$ the expected percent decrease in body mass index after taking a weight-loss drug for an individual who consumes $x_1$ grams of protein and exercises $x_2$ minutes per week. Suppose exercise and protein consumption are positively correlated and that $\mu$ corresponds to a multivariate Gaussian distribution with mean zero, unit \edit{variance}, and correlation equal to $0.9$. Then, $\mu_{\tens}$ corresponds to a multivariate Gaussian distribution with mean zero, unit covariance but correlation equal to $0$. Suppose we  report the intercept $f_{\emptyset}$ to summarize the typical decrease in body mass index in a population of people who might take the weight-loss drug (e.g., after drug approval). In the functional ANOVA decomposition of $f^*$ with respect to $\mu$, $f_{\emptyset} = \E_{\mu}[f^*] =  \E_{\mu}[f^* + \epsilon] = \E_{\mu}[y] \edit{= 40}$. Hence, this intercept says that, on average, people in this population should decrease their body mass index by \edit{40}\% if they take the drug. If we instead use $\mu_{\tens}$, then $f_{\emptyset} = \E_{\mu_{\tens}}[ f^*] = -50 \neq  \E_{\mu}[f^*]$, suggesting that the drug \emph{increases} body mass index. In the $\mu_{\tens}$ case, it is not clear how to interpret the intercept; $\mu_{\tens}$ averages the regression surface $f^*$ over individuals who rarely occur in the actual population (e.g., those who exercise very frequently but do not consume much protein); see also \cref{fig:toy_prod_measure_issue} and \cref{fig:toy_prod_measure_issue_two} for a visualization. 
\end{nexa}

\subsection{ A Change of Basis to Handle Covariate Dependence} \label{sec:change_basis}
We generalize to the non-independent case through a change-of-basis formula provided in \cref{thm:theo_proj}. Our formula allows us to re-express the effects estimated using the kernel in \cref{sec:trick_two}, which assumes independent covariates, to one with respect to the actual distribution $\mu$. Our idea is similar to ideas in numerical linear algebra; we use one parameterization of a vector space, in our case the space of functions $\hilb_Q$, that makes computation "nice." Once we finish computation in the "nice" parameterization, we use a change-of-basis formula to report the actual quantity we care about in the original parameterization of the space, namely reporting the functional ANOVA decomposition of our fit $\bar{f}_{\theta}$ with respect to $\mu$. 

To make this idea mathematically precise, suppose we can write $\hilb_Q$ using two different parameterizations, one that uses $\mu_{\tens}$ in \cref{eq:orth_hilb} (denoted as $\hilb^o_{V, \mu_{\tens}}$) and the other that uses $\mu$ in \cref{eq:orth_hilb} (denoted as $\hilb^o_{V, \mu}$). Then,
%
%
%
\begin{align}
		\hilb_Q =  \underbrace{\bigoplus_{V: |V| \leq Q} \hilb^o_{V, \mu_{\tens}}}_{\mathrm{(a)}} =  \underbrace{\bigoplus_{V: |V| \leq Q} \hilb^o_{V, \mu}}_{\mathrm{(b)}}. \label{eq:diff_param}
\end{align}
If these equalities indeed hold, then we can use \cref{thm:skim_fast_form} to estimate $f^*$ in $O(pQN^2 + N^3)$ time. Hence, it suffices to show how to take this estimate of $f^*$ and compute its functional ANOVA decomposition with respect to $\mu$ instead of $\mu_{\tens}$ (i.e., move from the parameterization in \cref{eq:diff_param}(a) to the one in \cref{eq:diff_param}(b)). We show how to compute this change-of-basis when all the $\hilb_{\{i\}}$ are finite-dimensional.


%
%

\begin{nassum} \label{assum:finite_dim}
For all $i \in [p]$ there exists a  $B_i < \infty$ and linearly independent set of continuous functions $\{ \phi_{ib} \}_{b=1}^{B_i}$ such that $\hilb_{\{i\}} =\text{span}\{1, \phi_{i1}, \cdots, \phi_{iB_i} \}$ and $\Phi_i = [ \phi_{i1}, \cdots,  \phi_{iB_i}]^T$.
\end{nassum}

\cref{assum:finite_dim} is a mild condition since we can approximate any function arbitrarily well by setting $B_i$ sufficiently large given that $\hilb_{\{i\}}$ is separable; see \citet{projection_anova} for rates of convergence for different finite-basis approximations. Under this assumption, \cref{eq:universal_space} implies that $ \hilb^o_{V, \mu_{\tens}} =  \hilb^o_{V, \mu}$. Hence, a change-of-basis formula exists. We provide the change-of-basis formula for $Q=2$ in \cref{thm:theo_proj}.

\begin{nlem} \label{eq:universal_space}
Under \cref{assum:tens_space}, \cref{assum:finite_dim}, and compactness of the domain of $f$, any $f \in \hilb$ is square-integrable with respect to any probability measure.  
\end{nlem}

\begin{nthm} \label{thm:theo_proj}
Suppose $Q=2$ and that \cref{assum:tens_space,assum:finite_dim} hold. For $f \in \hilb$, let
\begin{equation*}
	\begin{split}
		f &= f_{\emptyset}^{\mu_{\tens}} + \sum_{i=1}^p f_{\{ i\}}^{\mu_{\tens}} + \sum_{i,j=1}^p f_{\{ i, j \}}^{\mu_{\tens}} \\
			&= f_{\emptyset}^{\mu} + \sum_{i=1}^p f_{\{ i\}}^{\mu} + \sum_{i,j=1}^p f_{\{ i, j \}}^{\mu}
	\end{split}
\end{equation*}
be the functional ANOVA decompositions of $f$ with respect to $\mu_{\tens}$ and $\mu$, respectively. Then, there exist unique coefficients, $\Psi_{ij}^i \in \R^{1 \times B_i}, \Psi_{ij}^j  \in \R^{1 \times B_j}, \Psi_{ij}^{0} \in \R$, such that
\begin{equation}
	 \begin{split} \label{eq:change_basis}
	 	 f_{\{ i, j \}}^{\mu}(x_i, x_j) &=  f_{\{ i, j \}}^{\mu_{\tens}}(x_i, x_j) - [\Psi_{ij}^i \Phi_i(x_i) +  \Psi_{ij}^j \Phi_j(x_j) + \Psi_{ij}^{0}] \\
	 	f_{\{ i \}}^{\mu}(x_i) &=f_{\{ i \}}^{\mu_{\tens}}(x_i) + \sum_{p \geq j > i} \Psi_{ij}^i \Phi_i(x_i) + \sum_{1 \leq j < i} \Psi_{ji}^{i} \Phi_i(x_i)  \\
	 	f_{\emptyset}^{\mu} &= f_{\emptyset}^{\mu_{\tens}} + \sum_{1 \leq i < j \leq p} \Psi_{ij}^{0},
	 \end{split}
\end{equation}
where $\Phi_i$ denotes the (finite-dimensional) feature map in \cref{def:model_select_kernel}.
\end{nthm}
We prove \cref{thm:theo_proj} in \cref{A:theo_proj_proof}. By \cref{corr:skim_kern_time}, we can estimate $f_{\{ i \}}^{\mu_{\tens}}(x_i)$  and $f_{\{ i, j \}}^{\mu_{\tens}}(x_i, x_j)$ in time linear in $p$. Hence, it remains to show how we can actually compute each $\Psi_{ij}^i$ in \cref{thm:theo_proj}. In \cref{sec:alg_implement}, we show how to estimate $\Psi_{ij}^i$ arbitrarily well using a Monte Carlo approach.

\section{Final Algorithm and Implementation Details} \label{sec:alg_implement}
We start by describing and motivating our choice of sparsity prior on $\kappa$. Then, we show how we  fit $\kappa$ and $\eta$ using cross-validation and our computational tools in \cref{sec:two_kern_tricks}. We conclude by showing how we compute $\Psi_{ij}^i$  in \cref{thm:theo_proj} via Monte Carlo. 

\begin{algorithm}[h]
    \caption{Learn SKIM-FA Kernel Hyperparameters and Kernel Ridge Weights}
    \label{algo:learn_skim_hyp}
    \begin{algorithmic}[1] 
         \Procedure{LearnHyperParams}{$M$, $\gamma$, $T$, $u_{\text{init}}$, $c$} \edit{where $M$ equals the cross-validation fold size, $\gamma$ the learning rate, $T$ the number of gradient descent steps, $u_{\text{init}}$ initializer for $\tilde{U}$, and $c$ the value selected in \cref{eq:kappa_def}}

		\State \edit{Initialize $\tilde{U}^{(0)} = (u_{\text{init}}, \cdots, u_{\text{init}}) \in \R^p$ }
		\State  \edit{Initialize $\eta = (1, \cdots, 1) \in \R^{Q+1}$ \Comment{These are the global scale parameters in \cref{eq:skim_fa_weight}}}
		
		
		\State $\sigma^{(0)}_{\text{noise}} = \sqrt{0.5\var(Y)}$ \Comment{Initialize noise variance as half of the response variance}
		
		\State $\tau^{(0)} = (\tilde{U}^{(0)}, \eta^{(0)},\sigma^{(0)}_{\text{noise}})$  \Comment{\edit{Collect all parameters into a single vector}}

            \For{$t \in 1:T$} \Comment{\edit{Make $T$ gradient step updates}}
            		\State \edit{Sample $A \sim \pi$ \Comment{$\pi$ is uniform distribution over all $N-M$ subsets of $[N]$}}
            		\State \edit{Collect $N-M$ covariates in \edit{$X_A \in \R^{(N-M) \times p}$ and responses in $Y_A \in \R^{(N-M)}$}}
				
				\For{$i \in 1:p$}            		
            		
            		\State \edit{$U_i^{(t-1)} = [\tilde{U}_i^{(t-1)}]^2 / \left([\tilde{U}_i^{(t-1)}]^2 + 1\right)$}
            		
            		\State $\kappa_i^{(t-1)} = \max\left( U_i^{(t-1)}- c, 0\right)$ \label{algo1:kappa_val}

            		 \EndFor
            		
            		\State Compute kernel matrix \edit{$K_{\tau}^A \in \R^{(N-M) \times (N-M)}$, where $[K_{\tau}^A]_{ij} = k_{\text{SKIM-FA}}([X_{A}]_i, [X_{A}]_j)$ via \cref{eq:skim_fast_form} and $[X_{A}]_i, [X_{A}]_j \in \R^p$}
            		
            		\State Let $f_A$ equal the solution of \cref{eq:kern_ridge} with $\lambda = [\sigma^{(t)}_{\text{noise}}]^2$, $K = K_{\tau}^A$, $Y = Y_A$
            		
            		\State $L = \frac{1}{M} \sum_{n \in [N] \setminus A} (y^{(n)} - f_A(x^{(n)}))^2$  \Comment{\edit{Cross-validation loss in \cref{eq:exact_leave_out}}}
            		
            		\State $\tau^{(t)} = \tau^{(t-1)}  - \gamma \nabla_{\tau^{(t-1)}} L$ \Comment{\edit{Gradient update to parameters via autodiff library}}
            		            		               
            \EndFor 
            
            \State Compute $\alpha^{(T)}$, the kernel ridge regression weights found by solving \cref{eq:kern_ridge} using all $N$ datapoints with SKIM-FA hyperparameters equal to $\kappa^{(T)}, \eta^{(T)}, \sigma^{(T)}_{\text{noise}}$
                                 
            \State \textbf{return} $\kappa^{(T)}, \eta^{(T)}, \sigma^{(T)}_{\text{noise}}, \alpha^{(T)}$
            
        \EndProcedure
    \end{algorithmic}
\end{algorithm}

\noindent \textit{Our sparsity inducing prior on $\kappa$.} To induce sparsity in $\kappa$ for variable selection, we pick a prior on $\kappa_i$ that equals the mixture of a discrete point mass at 0 and a Uniform(0, 1) random variable. Similar to a \emph{spike-and-slab} prior \citep{george1993variable}, the point mass at $0$ allows us to achieve exact sparsity. Unlike a spike-and-slab prior, however, we construct our prior so that we can still take gradients (and hence use continuous optimization techniques like gradient descent). Our construction involves introducing another random variable $U_i$ so that
\begin{equation} \label{eq:kappa_def}
	\kappa_i = \frac{1}{1 - c} \max(U_i - c, 0), \ U_i \sim \text{Uniform}(0, 1).
\end{equation}
Then, $\prob(\kappa_i = 0) = c$. Otherwise, with probability $1 - c$, $\kappa_i \sim \text{Uniform}(0, 1)$. Hence, $c$ plays a similar role as a prior inclusion probability in a spike-and-slab prior. Since the gradient of $\kappa_i$ equals 0 when $U_i < c$, this zero gradient property is key for inducing sparsity; see our proof of \cref{prop:grad_descent}.\footnote{\edit{At $U_i = c$, the derivative of $\max(U_i - c, 0)$ is undefined. Since $\max(U_i - c, 0)$ is a convex function, the set of all subgradients at $U_i = c$ is $[0, 1]$. We let the subgradient equal 0 at $U_i = c$.}}\\ 

\noindent \textit{Cross-validation loss and optimization.} Given the empirical success of cross-validation and its use in other functional ANOVA methods (e.g., as in \citet{ss_anova, cosso}), we also use cross-validation to fit the SKIM-FA kernel hyperparameters $\kappa$ and $\eta$. Specifically, we would like to pick $U, \eta, \noisevar$ (where $\kappa_i =  \frac{1}{1 - c} \max(U_i - c, 0)$) by minimizing a leave-$M$-out cross validation loss:
\begin{equation} \label{eq:exact_leave_out}
	\begin{split}
	L(U, \eta, \noisevar) &= \frac{1}{{N \choose M}} \sum_{\substack{A: A \subset [N] \\|A| = N - M}} \left[ \frac{1}{M} \sum_{m \in A} (y^{(m)} - \bar{f}_A(x^{(m)}))^2 \right] \\
									&=\E_{A \sim \pi} \left[  \frac{1}{M} \sum_{m \in [N] \setminus A} (y^{(m)} - \bar{f}_A(x^{(m)}))^2 \right],
	\end{split}
\end{equation}
where $\bar{f}_A$ equals the kernel ridge regression fit in \cref{eq:kern_ridge} using the subset of datapoints in $A$ and $\pi$ equals the uniform distribution over all $N-M$ sized subsets $A$ of $[N]$. 

Since the gradient of $L(U, \eta, \noisevar)$ exists when all $U_i \neq c$, and the subgradient when some $U_i = c$, we can minimize \cref{eq:exact_leave_out} using gradient descent.  However, this loss is computationally intensive; we need to solve \cref{eq:kern_ridge} ${N \choose M}$ times in order to take a single gradient descent step. Instead, we approximate \cref{eq:exact_leave_out} by using stochastic gradient descent. Specifically, we randomly draw a single $A$ from $\pi$ in \cref{eq:exact_leave_out} and use the mean-squared prediction error of $\bar{f}_A$ to estimate \cref{eq:exact_leave_out}. Then, this estimate leads to an unbiased estimate of \cref{eq:exact_leave_out}, and hence an unbiased estimate of the gradient of $L(U, \eta, \noisevar)$. We summarize our full procedure in \cref{algo:learn_skim_hyp}, and prove that it leads to sparsity below.\footnote{Note that in \cref{algo:learn_skim_hyp} we do not minimize over $U$ but instead over $\tilde{U}$, where $U_i = \frac{\tilde{U}_i^2}{\tilde{U}_i^2 + 1}$.  Since the range of $ \frac{\tilde{U}_i^2}{\tilde{U}_i^2 + 1}$ equals $(0, 1)$ when $\tilde{U}_i$ varies over all of $\R$, we can optimize $\tilde{U}_i$ over an unconstrained domain. Since we only care about estimating the $\kappa_i$, it does not matter that $U_i$ is not a 1-1 function of $\tilde{U}_i$.}
 
\begin{nprop} \label{prop:grad_descent}
Suppose $\kappa_i^{(t)} = 0$ at some iteration $t$ in \cref{algo:learn_skim_hyp}. Then, for all subsequent iterations $t' \geq t$, $\kappa_i^{(t')} = 0$. 
\end{nprop}

Based on \cref{prop:grad_descent}, we may  view \cref{algo:learn_skim_hyp} as a gradient-based analogue of backward stepwise regression; we start with the model that includes all covariates by initializing all $U_i > c$ (and consequently all $\kappa_i > 0$). Then, we keep pruning off covariates the longer we run gradient descent. We demonstrate empirically in \cref{sec:experiments} that the actual data-generating covariates remain while the irrelevant covariates get pruned off. Once we have found the kernel hyperparameters from \cref{algo:learn_skim_hyp}, \cref{algo:var_select} and \cref{algo:orth_effects} show how to perform variable selection and recover the effects, respectively. Both \cref{algo:var_select} and \cref{algo:orth_effects} follow directly from \cref{lem:kern_var_select} and \cref{lem:kern_select}. \edit{In \cref{A:add_algo_details}, we discuss additional algorithmic details such as how to select $c$ in \cref{algo:learn_skim_hyp}.}

\noindent \textit{Estimating $\Psi_{ij}^i$ for change-of-basis formula in \cref{thm:theo_proj}.} Our change-of-basis formula in  \cref{eq:change_basis}  requires computing $\Psi_{ij}^i$. As we show in our proof of \cref{eq:change_basis}, $\Psi_{ij}^i$ has the interpretation as the basis coefficients associated with an $L^2$ projection. Since this $L^2$ projection requires a high-dimensional integration, we use Monte Carlo to estimate $\Psi_{ij}^i$ in \cref{algo:change_basis}. We prove that our Monte Carlo estimate  converges to the true projection coefficients in \cref{prop:change_basis_correct} below.
\begin{nprop} \label{prop:change_basis_correct}
Let $W \rightarrow \infty$ in \cref{algo:change_basis}. Then, the components returned from \cref{algo:change_basis} converge to the decomposition in \cref{eq:change_basis}.
\end{nprop}

 \begin{algorithm}[h]
    \caption{SKIM-FA Variable Selection}
    \label{algo:var_select}
    \begin{algorithmic}[1] 
         \Procedure{VarSelect}{$\kappa$}
            \State \textbf{return} $ \{i: \kappa_i \neq 0 \} $
        \EndProcedure
    \end{algorithmic}
\end{algorithm}

 \begin{algorithm}[h]
    \caption{Estimated functional ANOVA effect $\bar{f}_V$ of $\bar{f}_{\theta}$ with respect to $\mu_{\tens}$}
    \label{algo:orth_effects}
    \begin{algorithmic}[1] 
         \Procedure{OrthEffects}{$V$, $\alpha$, $\kappa$, $\eta$, $\alpha$} 

		\State $\theta_V = \eta_{|V|}^2 \prod_{i \in V} \kappa_i^2$
                      
            \State \textbf{return} $\bar{f}_V(\cdot) = \theta_V \sum_{n=1}^N \alpha_n k_V(x^{(n)}, \cdot)$ 
        \EndProcedure
    \end{algorithmic}
\end{algorithm}

\begin{algorithm}[h]
    \caption{Change of Basis Formula for Finite Dimensional Model Selection Kernels}
    \label{algo:change_basis}
    \begin{algorithmic}[1] 
         \Procedure{ReExpressEffect}{$\alpha$, $k_{\theta}$, $W$, $\mu$}  
         	\State Compute $ f_{\{ i, j \}}^{\mu_{\tens}} , f_{\{ i \}}^{\mu_{\tens}}, f_{\emptyset}^{\mu_{\tens}}$ using \cref{algo:orth_effects}
			\State For $1 \leq w \leq W$ randomly sample $x^{(w)} \iid \mu$    
			\State Compute $X_{ij} = [\mathbf{1} \ \Phi_i(x^{(1)}_i) \cdots \Phi_i(x^{(W)}_i) \ \Phi_j(x^{(1)}_j) \cdots \Phi_j(x^{(W)}_j)]^T$, $\mathbf{1} = (1, \cdots, 1) \in \R^{W}$
			\State Compute  $f_{ij, W}^{\mu_{\tens}} = [f_{\{ i, j \}}^{\mu_{\tens}}(x^{(1)}_i,  x^{(1)}_j) \cdots f_{\{ i, j \}}^{\mu_{\tens}}(x^{(W)}_i,  x^{(W)}_j)]^T$
			\State Compute $[\hat{\Psi}_{ij}^{0} \ \hat{\Psi}_{ij}^i \ \hat{\Psi}_{ij}^j]^T = (X_{ij}^T X_{ij})^{-1} X_{ij}^T f_{ij, W}^{\mu_{\tens}}$ \Comment{Least-squares projection}
			\State Compute $\hat{f}_{\{ i, j \}}^{\mu} =  f_{\{ i, j \}}^{\mu_{\tens}} - [\hat{\Psi}_{ij}^i \Phi_i^T(\cdot) +  \hat{\Psi}_{ij}^j \Phi_j^T(\cdot) + \Psi_{ij}^{0}]$
			\State Compute $\hat{f}_{\{ i \}}^{\mu} =f_{\{ i \}}^{\mu_{\tens}} + \sum_{j > i} \hat{\Psi}_{ij}^i \Phi_i(\cdot) + \sum_{j < i} \hat{\Psi}_{ji}^{i} \Phi_i(\cdot)$ 
			\State Compute $\hat{f}_{\emptyset}^{\mu} = f_{\emptyset}^{\mu_{\tens}} + \sum_{i < j} \hat{\Psi}_{ij}^{0}$
            \State \textbf{return} $\hat{f}_{\{ i, j \}}^{\mu}, \hat{f}_{\{ i \}}^{\mu}, \hat{f}_{\emptyset}^{\mu} $
        \EndProcedure
    \end{algorithmic}
\end{algorithm}

\section{Related Work} \label{sec:related_anova_work}  
 
\edit{Below we compare SKIM-FA to existing functional ANOVA methods, and our previous work for the linear interaction case. We continue our literature review in  \cref{A:lit_review}, where
we also contrast with further methods used for interaction discovery. \\}

\noindent \edit{\textit{Comparison with existing functional ANOVA methods.}}  The foundational work by \citet{ss_anova} used a type of model selection kernel to estimate the functional ANOVA decomposition of $f^*$ with splines. Since the method in \citet{ss_anova} does not lead to sparsity, \citet{supanova, cosso} put an $L_1$ penalty on $\theta$ to achieve sparsity, similar to \emph{multiple kernel learning} techniques \citep{mult_kern_learn}. Adding an $L_1$ penalty does not lead to an analytical solution nor a convex optimization problem. Hence, \citet{supanova, cosso} alternate between minimizing $\theta$ and recomputing $\bar{f}_{\theta}$, similar to \cref{algo:learn_skim_hyp}. Other approaches use cross-validation and gradient descent to iteratively select $\theta$ \citep{ss_anova}. In either case, the computational bottleneck is computing and inverting $ (K_{\theta} + \noisevar I_{N \times N})^{-1} Y$:  $k_{\theta}$ takes $O(p^Q)$ time to compute on a pair of points. Hence, computing and inverting  $K_{\theta} + \noisevar I_{N \times N}$ take $O(p^Q N^2)$ time and $O(N^3)$ time, respectively. 



Many existing functional ANOVA techniques assume that all covariates are independent, i.e., that $\mu$ equals the \emph{product measure}; see, for example, \citet{supanova, cosso, ss_anova, anova_zero_mean_sens}. \citet{hooker_prod} highlighted pathologies that arise when using $\mu_{\tens}$ instead of $\mu$. Specifically, he empirically showed on synthetic and real data that the functional ANOVA decomposition of an $f \in \hilb$ with respect to $\mu$ can be significantly different than the decomposition with respect to $\mu_{\tens}$. This discrepancy arises because $\mu_{\tens}$ can place high probability  in regions where the actual covariate distribution $\mu$ has low probability; see also \cref{sec:issue_prod_measure}.  

Finally, unlike our approach, some functional ANOVA methods  assume sparsity in the effects rather than in the covariates the response depends on; see, for example, \citet{supanova, cosso}). Recall from the discussion and example in \cref{sec:prelims} that sparsity in the covariates is useful for interpretability and downstream applications. But sparsity in the effects need not imply sparsity in the covariates. For example, suppose $Q=2$. Then, there are on the order of $p^2$ interaction effects. A method that selects $p$ non-zero effects might be considered sparse in the effects since $p \ll p^2$. But the selected effects could correspond to $p$ (or nearly $p$) selected covariates, which would not reduce the number of covariates. \\


\noindent \edit{\textit{Comparison with \citet{kit}.} There are five main differences between this work and \citet{kit}:  the method in \citet{kit} (1) assumes linear interaction effects, (2) only considers pairwise interactions, (3) assumes strong-hierarchy (namely that interactions only occur among selected main effects), (4) does not necessarily correspond to an ANOVA decomposition, and (5) does not induce exact sparsity. See \cref{A:lit_review} for further discussion.}

\section{Experiments} \label{sec:experiments}  
\noindent \textit{Summary of experimental results.} In this section, we compare our inference methods in  \cref{sec:alg_implement} against existing procedures in terms of variable selection and estimation performance. We find that when the interaction effects are strong or comparable to the strength of the additive effects, our method outperforms existing methods in terms of variable selection and estimation performance. When the interaction effects are weak our method does not (uniformly) have the best performance but still performs well relative to many of the other methods.\footnote{All results can be re-generated using the data and code provided in \url{https://github.com/agrawalraj/skimfapaper}.}

There are two immediate challenges with our empirical evaluation. The first is that existing methods estimate the functional ANOVA decomposition assuming all covariates are independent (or sometimes do not even specify the measure). Hence, we start our evaluation by assuming the covariates are independent so that we can compare against existing methods in \cref{sec:syn_exp} and \cref{sec:real_exp}. Then, in \cref{sec:anova_sens}, we show why the assumption of independent covariates is problematic to demonstrate the practical utility of \cref{algo:change_basis}. 

The second challenge concerns our performance metrics (detailed in \cref{sec:benchmarks} and \cref{sec:eval_metrics}), which require knowing the ground truth effects. Since we do not know the ground truth effects in  real data, we start in \cref{sec:syn_exp} by evaluating each method on simulated data so that we have ground truth effects. To compare methods on real data, we use a similar evaluation procedure as in \citet{kit} to construct a synthetic ground truth for benchmarking (see \cref{sec:real_exp} for details).

 \subsection{Benchmark Methods} \label{sec:benchmarks}
We compare our method against other methods used to  model high-dimensional data and interactions. We focus on the $Q=2$ case throughout since (1) existing methods typically only work for the pairwise interaction case and (2) higher-order interactions are often difficult to interpret and estimate. Even when $Q=2$, the functional ANOVA methods outlined in \cref{sec:benchmarks}  take $O(p^2N^2 + N^3)$ time, making them computationally intractable for even moderate $p$ and $N$ settings. Instead, we focus on methods that can model interactions and actually scale to moderate-to-large $p$ and $N$ settings. These methods include approximate "two-stage" and greedy forward-stage regression methods, and linear interaction models. We detail these approaches in more depth in \cref{A:lit_review}. The list below summarizes the candidate methods (and software implementations) that we select from each category for empirical evaluation. In \cref{A:model_implement_details} and \cref{A:add_algo_details} we detail the hyperparameters used to fit SKIM-FA. \\

 \noindent \textit{SPAM-2Stage}: we perform variable selection by fitting a sparse additive model (SpAM)  \citep{spam} to the data. We use the \texttt{sam} package in \texttt{R}. Since \texttt{sam}  does not provide a default way to select the $L_1$ regularization strength, we use 5-fold cross-validation. For estimation, we generate all main and interaction effects among the subset of covariates selected by SpAM. We calculate these effects by taking pairwise products of  univariate basis functions generated from a natural cubic spline basis with 5 total knots; see \cref{A:model_implement_details} for details. We estimate the basis coefficients (and hence effects) using ridge regression, where again we use 5-fold cross-validation to pick the $L_2$ regularization strength. \\
  
 	\noindent \textit{Multivariate Additive Regression Splines (MARS)}: we use the \texttt{python} implementation of MARS \citep{mars} in \texttt{py-earth}. We consider two functional ANOVA decompositions of the fitted regression function: (1) \textit{MARS-Vanilla} and (2) \textit{MARS-EMP}. For MARS-Vanilla, the main effect of each covariate equals the sum of all selected univariate basis functions of that covariate (i.e., after the pruning step of MARS). Similarly, each pairwise effect equals the sum of all selected bivariate basis functions of those two covariates. This is the functional ANOVA decomposition originally proposed in \citet{mars} and the one actually implemented in existing MARS software packages.  It is unclear, however, what measure this functional ANOVA decomposition is taken with respect to. To the best of our knowledge, there currently does not exist a procedure to perform the functional ANOVA decomposition of MARS with respect to the empirical distribution of the covariates. We describe how to perform such a decomposition via \textit{MARS-EMP}, which assumes the covariates are jointly independent. This method could be of independent interest and is outlined in \cref{A:mars_anova}. \\

 	\noindent \textit{Hierarchical Lasso (HierLasso)}: we use the implementation of HierLasso \citep{lim_heirch_lasso} in the authors' \texttt{R} package \texttt{glinternet}. Since \citet{lim_heirch_lasso} use cross-validation to pick the $L_1$ regularization strength, we similarly use 5-fold cross-validation. \\
 	
 	\noindent \textit{Pairs Lasso}: we fit the Lasso on the expanded set of features $\{x_i \}_{i=1}^p$ and $\{x_ix_j\}_{i,j=1}^p$. We fit the Lasso using the \texttt{python} package \texttt{sklearn}, and use 5-fold cross-validation to select the $L_1$ regularization strength.

 \subsection{Evaluation Metrics} \label{sec:eval_metrics}
 
  \noindent \textit{Variable selection evaluation metrics.} We consider both the power to select correct covariates and avoid incorrect ones. \emph{\# Correct Selected} counts the number  of covariates correctly selected by the method. Higher is better. \emph{\# Wrong Selected}  counts the number of covariates incorrectly selected by the method (i.e., Type I error). Lower is better. \emph{\# Correct Not Selected} counts the number of covariates that belong to the true model but were not selected by the method (i.e., Type II error). Lower is better. \\

\noindent \textit{Estimation evaluation metrics.} We evaluate how well a method estimates main effects and interaction effects. Instead of  looking only at the total mean squared estimation error, we break this error into multiple buckets to understand what bucket drives the majority of the error. Lower is better for all of the following quantities. \emph{Correct Selected SSE (Main)} takes the sum of squared errors (SSE) between each estimated main effect component and true main effect component. This sum equals $\sum_{i \in S_1} \| f_i^* - \hat{f}_i \|_{\mu}^2$, where $S_1$ is the set of correctly identified main effects, $\hat{f}_i $ is the estimated main effect, and $f_i^*$ is the true main effect. \emph{Correct Not Selected SSE (Main)} takes the sum of squared norms of main effects not selected. This sum equals $\sum_{i \in S_2} \| f_i^* \|_{\mu}^2$, where $S_2$ is the set of correct  main effects not selected. \emph{Wrong Selected SSE (Main)} takes the sum of  squared norms of  main effect components incorrectly selected. This sum equals $\sum_{i \in S_3} \| \hat{f}_i \|_{\mu}^2$, where $S_3$ is the set of incorrect  main effects selected. \emph{Correct Selected SSE (Pair)}, \emph{Correct Not Selected SSE (Pair)}, and \emph{Wrong Selected SSE (Pair)} are the same as the analogous main effect metrics but instead considers interaction effects. \emph{Total SSE} equals the sum of the 6 buckets above and \emph{Total SSE / Signal Variance} equals the relative estimation error, i.e., Total SSE divided by the true signal variance.

 \subsection{Synthetic Data Evaluation} \label{sec:syn_exp}

  %

%


 We randomly generate covariates and responses as follows. For the covariates, we draw each data point and covariate dimension $x^{(n)}_i \iid \text{Uniform}([-1, 1])$. Since $[-1, 1]$ is compact, \cref{thm:unique_fanova} ensures that the functional ANOVA decomposition is unique. We let $y$ depend on the first 5 covariates; the remaining $p-5$ covariates are taken as noise covariates that we do not want to select. To generate responses reflective of what we might expect in real data, we consider the 5 trends shown in \cref{fig:ex_trends}: linear, sine, logistic, quadratic, and exponential. We let the main effects equal the sum of these 5 trends, where the $i$th trend is applied to covariate $i$. For the interactions between the first 5 covariates, we consider all pairwise products of the 5 trends above, resulting in 10 total interactions. 
We select a noise variance such that the $R^2 = \frac{\sigma^2_{\text{signal}}}{\sigma^2_{\text{signal}} + \sigma^2_{\text{noise}}} = 0.8$, where $\sigma^2_{\text{signal}} = \inner{f^*}{f^*}_{\mu}$. We further decompose the signal variance in terms of the total variance explained by main effects and interactions. Similar to the empirical evaluations in \citet{lim_heirch_lasso}, we consider the following three settings:
 \begin{itemize}
  	\item \textbf{Weak Main Effects:} each main effect and pairwise effect has 0.01 * 1/5 and 0.99 * (1 / 10) of the total signal variance, respectively. Hence, the total main effect and pairwise effect variances equal 1\% and 99\% of the total signal variance, respectively.
 	
 	\item \textbf{Equal Main and Interaction Effects:} each main effect and pairwise effect has 0.5 * 1/5 and 0.5 * (1 / 10) of the total  signal variance, respectively. Hence, the total main effect variance equals the total pairwise signal variance.
 	
 		\item \textbf{Main Effects Only:} each of the 5 main effects has 1/5th of the total signal variance, and each pairwise effect has 0 signal variance (i.e., no pairwise interactions).
 	 	
 \end{itemize}

To test the impact of increasing dimensionality on inference quality, we consider $p \in \{250, 500, 1000\}$  and keep $N=1000$ fixed for each setting. For estimation, we compare only the nonlinear methods; linear methods will artificially perform poorly since some of the effects are highly nonlinear by construction. Evaluating estimation performance is trickier than evaluating selection performance since the functional ANOVA decomposition depends on the choice of measure. Unless otherwise stated, the target of inference is finding the functional ANOVA decomposition of $f^*$ with respect to $\mu$ (the joint distribution of the covariates). 

\begin{Table}[t]
\centering
\renewcommand{\arraystretch}{.667}
\resizebox{\textwidth}{!}{%
\begin{tabular}{ccccc}
\hline
\textbf{Method} & \textbf{Setting} & \textbf{\# Correct Selected} & \textbf{\# Wrong Selected} & \textbf{\# Correct Not Selected} \\ \hline

\textbf{SKIM-FA} & Weak Main & 5 & 9 & 0 \\
MARS & Weak Main & 5 & 75 & 0 \\
SPAM-2Stage & Weak Main & 1 & 41 & 4 \\
HierLasso & Weak Main & 5 & 120 & 0 \\
Pairs Lasso & Weak Main & 5 & 144 & 0 \\ \hline

\textbf{SKIM-FA} & Equal & 5 & 0 & 0 \\
MARS & Equal & 5 & 71 & 0 \\
SPAM-2Stage & Equal & 5 & 15 & 0 \\
HierLasso & Equal & 5 & 40 & 0 \\
Pairs Lasso & Equal & 5 & 213 & 0 \\  \hline

\textbf{SKIM-FA} & Main-Only & 3 & 0 & 2 \\
MARS & Main-Only & 5 & 70 & 0 \\ 
SPAM-2Stage & Main-Only & 5 & 15 & 0 \\
HierLasso & Main-Only & 4 & 5 & 1 \\
Pairs Lasso & Main-Only & 4 & 6 & 1 \\ \hline

\end{tabular}
}
\caption{Synthetic Data Variable Selection Performance Results for $p=1000$. The method with the fewest number of incorrect covariates selected is bolded.}
\label{tab:1k_var_select}
\end{Table}



\begin{Table}[t]
\centering
\renewcommand{\arraystretch}{.75}
\resizebox{\textwidth}{!}{%
\begin{tabular}{cccccccccc}
\hline
\textbf{Method} & \textbf{Setting} & \textbf{\begin{tabular}[c]{@{}c@{}}Correct \\ Selected \\ SSE \\ (Main)\end{tabular}} & \textbf{\begin{tabular}[c]{@{}c@{}}Correct \\ Not \\ Selected \\ SSE \\ (Main)\end{tabular}} & \textbf{\begin{tabular}[c]{@{}c@{}}Wrong\\ Selected \\ SSE\\ (Main)\end{tabular}} & \textbf{\begin{tabular}[c]{@{}c@{}}Correct\\ Selected \\ SSE \\ (Pair)\end{tabular}} & \textbf{\begin{tabular}[c]{@{}c@{}}Correct\\ Not \\ Selected \\ SSE\\ (Pair)\end{tabular}} & \textbf{\begin{tabular}[c]{@{}c@{}}Wrong\\ Selected \\ SSE\\ (Pair)\end{tabular}} & \textbf{\begin{tabular}[c]{@{}c@{}}Total \\ SSE\end{tabular}} & \textbf{\begin{tabular}[c]{@{}c@{}}Total SSE\\ ÷ \\ Signal \\ Variance\end{tabular}} \\ \hline

\textbf{SKIM-FA} & Weak Main & 0.72 & 0 & 1.37 & 0.61 & 0 & 0.63 & 3.33 & 0.17 \\
SPAM-2Stage & Weak Main & 0.16 & 0.2 & 6.69 & 0 & 18.33 & 0.31 & 25.69 & 1.28 \\
MARS-EMP & Weak Main & 0.67 & 0 & 5.86 & 3.37 & 0 & 5.63 & 15.52 & 0.78 \\
MARS-VANILLA & Weak Main & 23.62 & 0 & 3.18 & 23.16 & 0 & 15.43 & 65.39 & 3.27 \\ \hline

\textbf{SKIM-FA} & Equal & 1.54 & 0 & 0 & 0.29 & 0 & 0 & 1.82 & 0.09 \\
SPAM-2Stage & Equal & 1.67 & 0 & 1.07 & 0.41 & 0 & 2.16 & 5.31 & 0.27 \\
MARS-EMP & Equal & 0.61 & 0 & 3.84 & 1.7 & 0 & 2.52 & 8.67 & 0.43 \\
MARS-VANILLA & Equal & 454.88 & 0 & 3.16 & 21.46 & 0 & 13.22 & 492.72 & 24.64 \\ \hline

SKIM-FA & Main Only & 2.7 & 8.1 & 0 & 0 & 0 & 0.24 & 11.03 & 0.55 \\
\textbf{SPAM-2Stage} & Main Only & 2.67 & 0 & 0.78 & 0 & 0 & 0.02 & 3.46 & 0.17 \\
MARS-EMP & Main Only & 0.45 & 0 & 2.68 & 0 & 0 & 2.39 & 5.51 & 0.28 \\
MARS-VANILLA & Main Only & 16.14 & 0 & 1.56 & 0 & 0 & 10.33 & 28.02 & 1.4 \\ \hline

\end{tabular}%
}
\caption{Synthetic Data Estimation Performance Results for $p=1000$. The method with the smallest total SSE is bolded.}
\label{tab:1k_est_perform}
\end{Table}

We summarize the variable selection and estimation performances of each method for $p = 1000$ in \cref{tab:1k_var_select} and \cref{tab:1k_est_perform}, respectively; see \cref{A:add_exp} for model performance results for all choices of $p$. As we discuss below, SKIM-FA outperforms all of the other methods (in terms of both variable selection and estimation) in the Weak Main Effects and the Equal Main and Interaction settings. For the Main Effects Only setting, SKIM-FA selects the fewest number of incorrect covariates. Since SKIM-FA does not select two of the correct covariates in this setting, however, its estimation performance is worse than some of the other benchmark methods. \\

\noindent \textit{Weak main effects setting.} In the setting of weak main effects, Spam-2Stage performs worse than the other methods in terms of the power to select correct covariates; Spam-2Stage only selects one correct covariate for $p = 500$ and $p=1000$. This poor variable selection is expected since the signal is locked away in the interactions but SpAM assumes additive effects. In particular, only 1\% of the variance is explained by additive effects (even though additive \emph{and} interaction effects explain 80\% of the variance in the response). MARS, HierLasso, and Pairs Lasso detect all 5 correct covariates but they all pick up many more incorrect covariates relative to SKIM-FA.

MARS selects many incorrect covariates because it can only form an interaction between two covariates if at least one of the covariates has an additive effect (similar to Spam2Stage). In the extreme case of no additive effects, for example, MARS randomly selects covariates to have additive effects. By random chance, MARS will eventually select a correct covariate (i.e., one the response actually depends on) to have an additive effect. Since this covariate has an interaction effect, in the next step MARS will (likely) select the correct interaction effect. Hence, MARS will need to select many incorrect covariates as additive effects before identifying the true interactions.  



In terms of estimation performance, SKIM-FA has the smallest total mean-squared estimation error. Since Spam-2Stage only considers interactions between covariates selected by SpAM, its poor estimation performance is driven by not selecting many of the correct covariates. MARS-VANILLA performs a functional ANOVA decomposition with respect to an unspecified measure. Hence, it is unclear how to interpret its main and interaction effects. One might think (and truthfully what we initially thought) that MARS-VANILLA would still return a functional decomposition close to one with respect to the actual covariate distribution. \cref{tab:est_main_weak}  shows that this intuition is incorrect; the relative estimation error of MARS-VANILLA always exceeds 1! This poor estimation performance stems from not specifying the measure (and hence the target of inference), not MARS's ability in finding a model with good predictive performance. In particular, MARS-EMP, which produces the \emph{exact same predictions} as MARS-VANILLA, yields better  performance because it re-orthogonalizes the fit with respect to the covariate distribution. \\

\noindent \textit{Equal main and interaction effects setting.} In this setting, all methods are able to recover all 5 true covariates. 
For both estimation and variable selection, SKIM-FA performs best. \\
  
\noindent \textit{Main effects only setting.} Each method selects the majority of correct covariates. However, some methods -- namely Pairs Lasso and HierLasso -- have a systematic bias; for all choices of $p$, they never select covariate 3 (the quadratic trend) since a quadratic trend has a weak linear correlation. Since the other methods can model nonlinear relationships, they can pick up this trend. Hence, they have better statistical power to detect correct covariates, improving variable selection performance.
In terms of Type I error, some methods select incorrect covariates much more frequently. For example, MARS consistently selects over 50 incorrect covariates for all choices of $p$. A potential reason for this poor performance is that MARS induces sparsity through a greedy pruning step instead of an actual sparsity inducing penalty as in the other methods. \\

\noindent \textit{Runtime comparisons.} We conclude this section by comparing each method in term of runtime in the high-dimensional setting. The two Lasso methods  take $O(p^2N)$ time while the remaining methods depend only linearly on $p$. When $p > N$ and $\omega = p / N$, our method takes $O(\omega N^3)$ while the two Lasso based methods take $O(\omega^2N^3)$ time. Hence, for higher-dimensional problems, our method will become much faster relative to the Lasso methods. For example, in genome-wide association studies, data sets can have $N$ on the order of $10^3$ and $p$ on the order of $10^7$ \citep{1000_genomes}. Hence, $\omega = 10^4$, which corresponds to a potential $10^4$ computational speedup factor. In \cref{fig:syn_runtime}, we compare the runtimes of each method as we vary $p / N$ on simulated data. We keep $N$ fixed at $100$ and vary $p$ from $10$ to $10^4$.  As expected, as  $p / N$ increases, our method yields substantial computational savings relative to Pairs Lasso and HierLasso. Relative to Spam-2Stage and MARS, our method does not yield better computational scaling. However, based on our synthetic evaluation above (and real data evaluation in \cref{sec:real_exp}), we have better statistical performance.

%

%

\subsection{Evaluation on Real Data: \edit{Bike Sharing \& Concrete Compressive Strength Data sets}}  \label{sec:real_exp}
Evaluating the methods in terms of variable selection and estimation quality is challenging because we typically do not have ground truth main and interaction effects for high-dimensional (real) data. Similar to the evaluation procedure in \citet{kit}, we instead take a low-dimensional data set where $N$ is large and $p$ is small.  We make it high-dimensional by adding synthetic random noise covariates. These two choices have several purposes. First, by fitting a regression function on the original low-dimensional data set, standard $N^{-1/2}$ statistical convergence rates apply. Hence, for large $N$, a maximum-likelihood estimate of the regression function will be close to the true regression function, creating a (near) ground truth for estimation evaluation. For variable selection, the random noise covariates create a ``synthetic control;'' if a method selects any of the random noise covariates as a main or interaction effect, we know the method selected an incorrect covariate.

\begin{Table}[]
\centering
\renewcommand{\arraystretch}{.667}
\begin{tabular}{@{}cccc@{}}
\toprule
\textbf{Method} & \textbf{\# Covariates} & \textbf{\# Original Selected} & \textbf{\# Wrong Selected} \\ \midrule
\textbf{SKIM-FA}         & 1000                   & 3                             & 0                          \\
HierLasso       & 1000                   & 3                             & 5                          \\
SPAM-2Stage     & 1000                   & 3                             & 8                          \\
Pairs Lasso     & 1000                   & 3                             & 76                         \\
MARS            & 1000                   & 3                             & 119                        \\ \bottomrule
\end{tabular}
\caption{Variable Selection Performance for the Bike Sharing Data Set.}
\label{tab:bike_var_perform}
\end{Table}


Based on these ideas, we consider the popular (low-dimensional) Bike Sharing data set, which we downloaded from the UCI Machine Learning Repository. This data set contains 17,389 datapoints and $13$ covariates. We consider 4 continuous variables (hour, air temperature, humidity, windspeed) and use the total number of bikes rented as the response.  We standardize the response by subtracting the mean and dividing by the standard deviation, and min-max standardize the covariates so that each covariate belongs to $[0, 1]$.  For the proxy ground truth, we fit a pairwise interaction model consisting of all 4 main effects and 6 possible pairwise interactions.

\begin{Table}[]
\centering
\renewcommand{\arraystretch}{.667}
\resizebox{\textwidth}{!}{%
\begin{tabular}{@{}ccccccccc@{}}
\toprule
\textbf{Method} & \textbf{\# Noise} & \textbf{\begin{tabular}[c]{@{}c@{}}Correct \\ Selected \\ SSE \\ (Main)\end{tabular}} & \textbf{\begin{tabular}[c]{@{}c@{}}Correct \\ Not \\ Selected \\ SSE \\ (Main)\end{tabular}} & \textbf{\begin{tabular}[c]{@{}c@{}}Wrong\\ Selected \\ SSE\\ (Main)\end{tabular}} & \textbf{\begin{tabular}[c]{@{}c@{}}Correct\\ Selected \\ SSE \\ (Pair)\end{tabular}} & \textbf{\begin{tabular}[c]{@{}c@{}}Correct\\ Not \\ Selected \\ SSE\\ (Pair)\end{tabular}} & \textbf{\begin{tabular}[c]{@{}c@{}}Wrong\\ Selected \\ SSE\\ (Pair)\end{tabular}} & \textbf{\begin{tabular}[c]{@{}c@{}}Total \\ SSE\end{tabular}} \\ \midrule
\textbf{SKIM-FA} & 1000 & 0.145 & 0.002 & 0 & 0.107 & 0.009 & 0 & 0.263 \\
SPAM-2Stage & 1000 & 0.149 & 0.002 & 0.027 & 0.081 & 0.009 & 0.000 & 0.269 \\
MARS-EMP & 1000 & 0.214 & 0.002 & 0.485 & 0.054 & 0.026 & 0.245 & 1.026 \\
MARS-Vanilla & 1000 & 6.556 & 0.002 & 0.796 & 0.947 & 0.026 & 1.882 & 10.209 \\ \bottomrule
\end{tabular}
}
\caption{Estimation Performance for the Bike Sharing Data Set.}
\label{tab:bike_est_perform}
\end{Table}

Similar to our synthetic evaluation, we randomly subsample a total of $N=10^3$ datapoints, and then train each benchmark method on this subsampled data set. To make the inference task high-dimensional we inject $p_{\text{noise}} \in \{250, 500, 1000\}$ random noise covariates, where these noise covariates are drawn iid from a Uniform(0, 1) distribution.  We report on the same variable selection and estimation metrics as in the synthetic experiments for $p_{\text{noise}} = 1000$ in \cref{tab:bike_var_perform} and \cref{tab:bike_est_perform}, respectively; see \cref{tab:bike_var_perform_full} and \cref{tab:bike_est_perform_full} for all choices of $p_{\text{noise}}$. We see that again SKIM-FA has similar or much better estimation and variable selection performance relative to the other methods. Finally, to understand the impact of correlated predictors on performance, we append correlated (real) covariates to the Bike Sharing data set (instead of synthetic ones drawn from a Uniform(0, 1) distribution) in \cref{A:bike_fake_real}. We again find that SKIM-FA has better performance than the other methods. 

\begin{figure}[]
        \centering
        \begin{subfigure}[b]{0.33\textwidth}
            \centering
            \includegraphics[width=\textwidth]{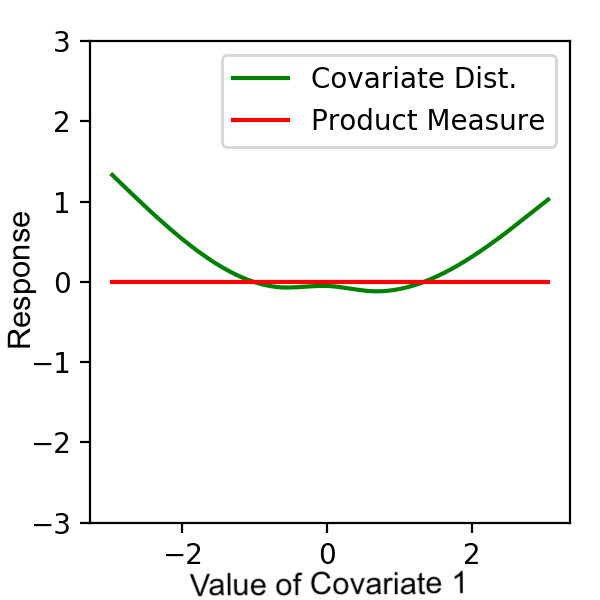}
            \caption[Network2]%
            {{\small $\rho = 0.1$}}    
        \end{subfigure}
        \hspace{2cm}
        \begin{subfigure}[b]{0.33\textwidth}  
            \centering 
            \includegraphics[width=\textwidth]{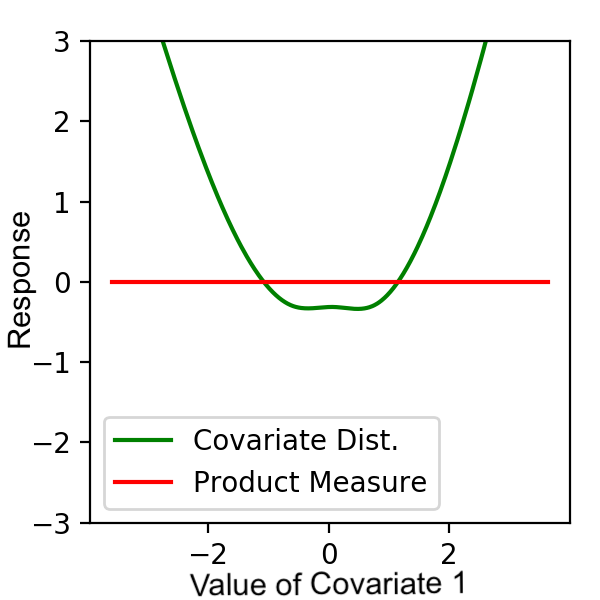}
            \caption[]%
            {{\small $\rho = 0.5$}}    
        \end{subfigure}
                \caption{The left hand and right hand plots show how the additive effect of $x_1$ (in the functional ANOVA decomposition of the function $x_1 x_2$) varies as the correlation between $x_1$ and $x_2$ increases.} \label{fig:corr_anova}

\end{figure}

\subsection{Impact of Correlated Predictors on the Functional ANOVA} \label{sec:anova_sens}
So far we have performed the functional ANOVA decomposition assuming that the covariates are jointly independent; for our synthetic data evaluation in \cref{sec:syn_exp} this independence held by design. Here we show the effect correlated predictors have on the resulting decomposition. Recall that previous functional ANOVA methods assume product measure, but our \cref{algo:change_basis} provides the flexibility to select different measures. We demonstrate the practical utility of this flexibility here. To this end, we consider the simplest possible regression function with interactions: $f(x_1, x_2) = x_1 x_2$. If $x_1 \indep x_2$, then the functional ANOVA decomposition of $f$ with respect to $\mu(x_1, x_2)$ equals $x_1 x_2$. However, if $x_1$ and $x_2$ are correlated, then the functional ANOVA decomposition no longer equals $x_1 x_2$. In particular, as the correlation between $x_1$ and $x_2$ increases, $f$ can be explained better by additive effects (e.g., in the degenerate case when $x_1 = x_2$, then $f(x_1, x_2) = x_1^2)$. To test this empirically, we randomly generate $x_1, x_2$ from a multivariate Gaussian distribution with marginal variances equal to $1$ and pairwise correlation equal to $\rho$. We let $\rho \in \{0.1, 0.5\}$. \cref{fig:corr_anova} shows that when $\rho$ gets stronger, the discrepancy between a functional ANOVA decomposition with respect to $\mu(x_1, x_2)$ versus product measure $\mu_{\tens} = N(0, 1) \tens N(0, 1)$ increases. As expected, as the correlation increases, a quadratic-like function of $x_1$ and $x_2$ explains $f$ increasingly well.

\begin{figure}[]
        \centering
        \begin{subfigure}[b]{0.4\textwidth}
            \centering
            \includegraphics[width=\textwidth]{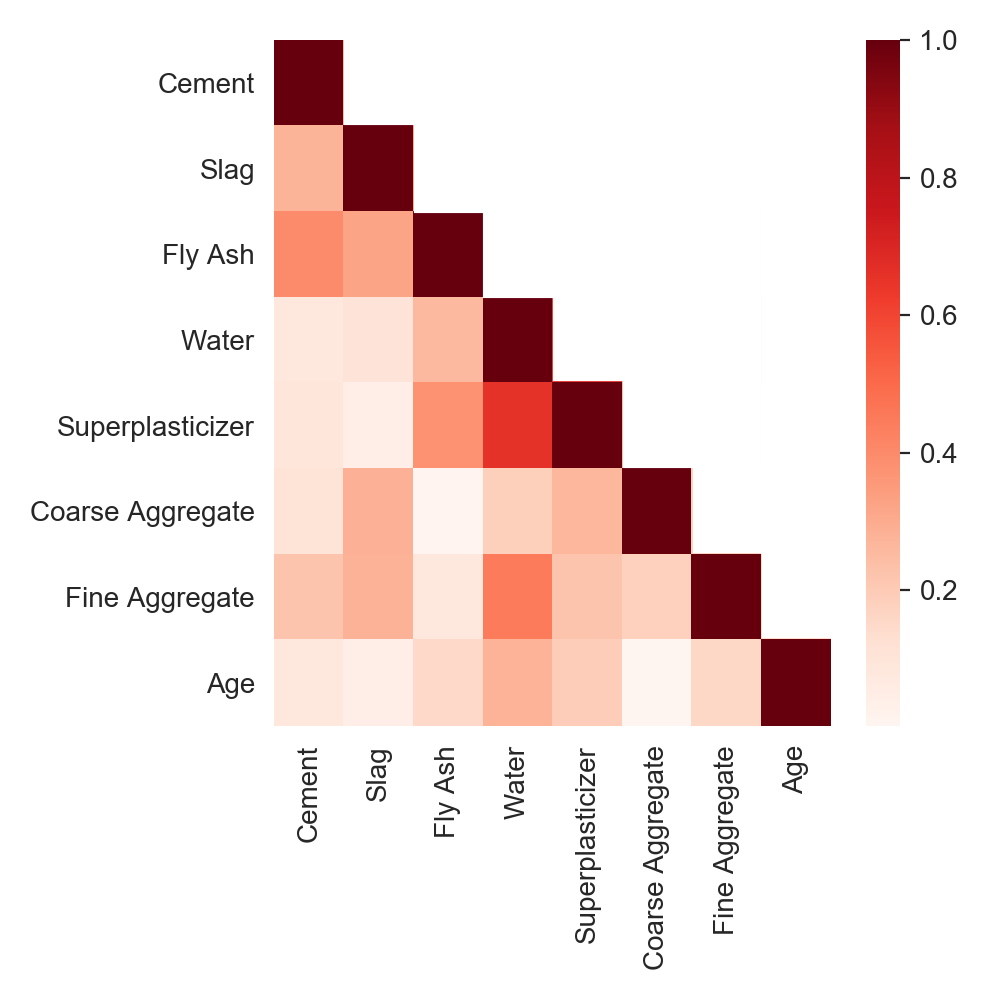}
            \caption[Network2]%
            {{\small Absolute Value of Correlation Matrix}}    
        \end{subfigure}
        \hspace{1cm}
        \begin{subfigure}[b]{0.4\textwidth}  
            \centering 
            \includegraphics[width=\textwidth]{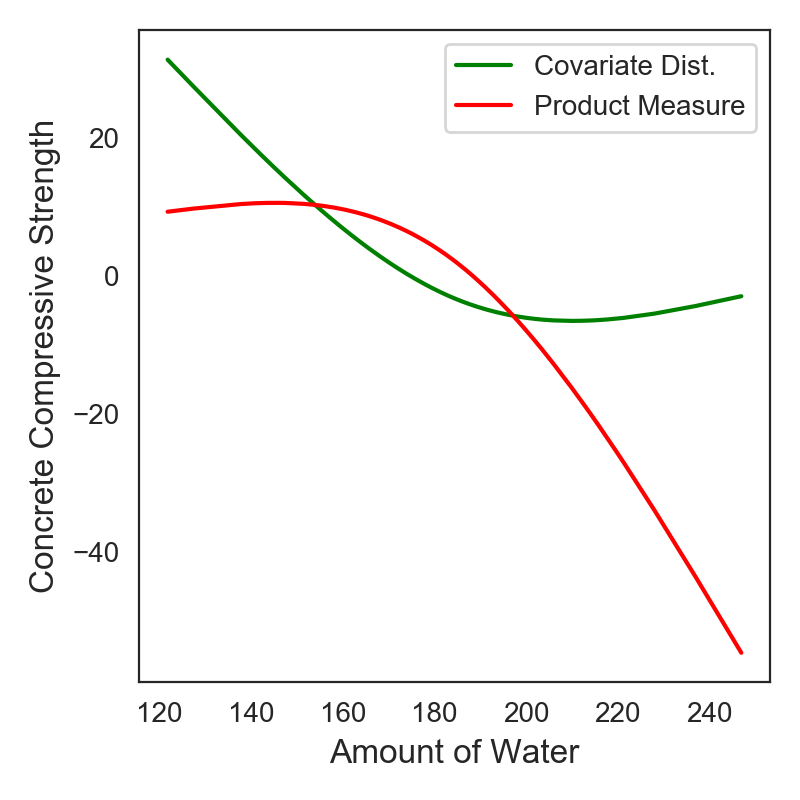}
            \caption[]%
            {{\small Main Effect of Water on Strength}}    
        \end{subfigure}
                \caption{Effect of Correlated Predictors on the Concrete Compressive Strength Data Set} \label{fig:concrete_anova}

\end{figure}

We perform a similar analysis above but for real data, namely the  Concrete Compressive Strength data set from the UCI machine learning repository. In \cref{fig:concrete_anova}, we plot the correlations between the 8 covariates that potentially predict the response (concrete strength). The two most correlated covariates are the amount of water and the amount of superplasticizer. Since the covariates have non-trivial correlations, the functional ANOVA decomposition with respect to $\mu$ and $\mu_{\tens}$ might be different based on \cref{lem:arb_far_intercept}. In \cref{fig:concrete_anova} we see that there indeed is a difference; the (estimated) additive effect for water on concrete strength varies substantially depending on which measure is selected to perform the functional ANOVA decomposition. In \cref{A:corr_pred_bike}, we compare how the functional ANOVA decomposition changes depending on if we use $\mu$ or $\mu_{\tens}$ for the Bike Sharing data set. Unlike the Concrete Compressive Strength data set, however, we do not see a large difference between the two functional ANOVA decompositions for the Bike Sharing data set.

\subsection{Evaluation on Real Data: Obesity Gene-Expression and SNP Data Set}
\edit{We conclude by evaluating each method on a high-dimensional genomics data set where $N \ll p$. Unlike our previous evaluation, however, we do not know the ground truth effects. Hence, we are unable to compute the evaluation metrics in \cref{sec:eval_metrics}. We instead report the mean-squared prediction error of each method on a left-out test set as a proxy. As a qualitative check on inference quality, we interpret the genes selected by SKIM-FA and find that some correspond to genes already flagged as obesity-related based on previous biological studies. Below we summarize the data and our findings in more depth.}

\begin{Table}[]
\centering
\renewcommand{\arraystretch}{.667}
\begin{tabular}{@{}cccr@{}}
\toprule
\textbf{Method} & \textbf{\# Covariates Selected} & \textbf{MSE} & \textbf{$R^2$} \\ \midrule
\textbf{SKIM-FA}         & 31                   & 39.1                             & 0.43                          \\
HierLasso       &    8                & 68.1                             & 0.01                          \\
SPAM-2Stage     & 0                   & 68.9                             & 0.00                          \\
MARS            &    15                & 96.1                             & -0.39 \\
Pairs Lasso     & --                   &  --                             & --                         
                        \\ \bottomrule
\end{tabular}
\caption{Test Set Predictive Performance for
the Obesity Gene-Expression and SNP Data Set.}
\label{tab:obese_perform}
\end{Table}

\edit{We consider the data set kindly provided by \citet{obese_data}, which consists of the body mass index (BMI) of $N=87$ individuals. After using the pre-processing steps in \citet{obese_data}, we also consider 13,276 gene-expression levels, 16 single-nucleotide
polymorphisms (SNPs), and a genetic risk-score feature as covariates for a total of $p=$ 13,293 covariates. Since the number of covariates is more than 100 times the number of observations, and the number of pairwise interactions almost exceeds 100 million, this data set leads to a non-trivial inference task. We report the out-of-sample mean-squared error and out-of-sample $R^2$ for each method in \cref{tab:obese_perform} (15\% of the data is used for testing purposes). \cref{tab:obese_perform} has missing values for  Pairs Lasso since the number of interactions is too large to run on a single machine. SPAM-2Stage does not select any covariates, and hence the $R^2$ is zero. Both HierLasso and MARS seem to overfit given the poor $R^2$ performance, even though each selects 8 and 15 covariates, respectively. SKIM-FA performs the best ($R^2 = 0.43$), and selects 31 covariates; see \cref{A:obese_details} for the names of all genes and SNPs selected.}

\edit{We do not know which of these 31 genes are truly associated with obesity. Nevertheless, we find that several genes SKIM-FA selects are obesity related based on previous studies. SKIM-FA selects IRS2, which is a gene associated with obesity and diabetes risk; see, for example, \citet{irs2_one}. SKIM-FA says that IRS2 has a negative effect on BMI (i.e., a higher expression of IRS2 decreases BMI) which agrees with the findings in \citet{irs2_two} based on experimental data on mice; see \cref{A:obese_details} for details. SKIM-FA also selects a SNP (Rs2112347) and two genes (KISS1R and SKP1) which are obesity related based on \citet{snp_obese, kissr, skp1}. Interestingly, as we discuss in \cref{A:obese_details}, SKP1 does not have strong additive effects, but it has the strongest interactions. Hence, SKP1 might be an
interesting candidate for further study of its interaction
properties.}

\section{Conclusion} 
 In this paper, we developed a new, computationally efficient method to perform sparse functional ANOVA decompositions. The heart of our procedure relied on a new kernel trick to implicitly represent nonlinear interactions (\cref{thm:skim_fast_form}), and a change-of-basis formula (\cref{thm:theo_proj}) to re-express the fit in terms of an arbitrary measure. We compared our method against other methods often used to model high-dimensional data with interactions. We found improved performance on both simulated and real data sets by relaxing assumptions such as linearity and the presence of strong-additive effects while still remaining competitive (or being orders of magnitude faster) in terms of runtime.

There are many interesting future research directions. One involves scaling our method to both the large $N$ and $p$ setting; our current method takes $O(pQN^2 + N^3)$ time which becomes problematic for large $N$.  This cubic dependence, however, is not unique to our method but rather a fundamental obstacle faced by kernel ridge regression and Gaussian processes. Fortunately, many methods already exist to help alleviate these computational challenges with respect to $N$; see, for example, \citet{black_box, Titsias09variationallearning, Quinonero}. Another interesting direction involves applying our method to biological data sets. In particular, an open challenge in genomics has been detecting \emph{epistasis}, or interaction effects between genetic variants, from genome sequencing data  \citep{missing_herit, gwas_interaction, epistatis_gene2, gpu_epistasis}. Detecting epistasis has been statistically and computationally challenging because $p$ is in the millions, so the number of pairwise interactions is on the order of \emph{trillions}. Since our method does not require explicitly generating all interactions, it has the potential to tractably detect interactions in such especially high-dimensional data regimes. 

 \acks{This work was supported in part by the DARPA I2O LwLL program, NSF Award 2029016, an ONR Early Career Grant, and MIT Lincoln Laboratory.}


\newpage

\appendix 
\counterwithin{figure}{section}
\counterwithin{Table}{section}

\noindent \textbf{Supporting Materials}


 \section{Figures and Tables Referenced in \cref{sec:experiments}} 

\begin{figure}[H]
        \centering
        \begin{subfigure}[b]{0.6\textwidth}
            \centering
            \includegraphics[width=\textwidth]{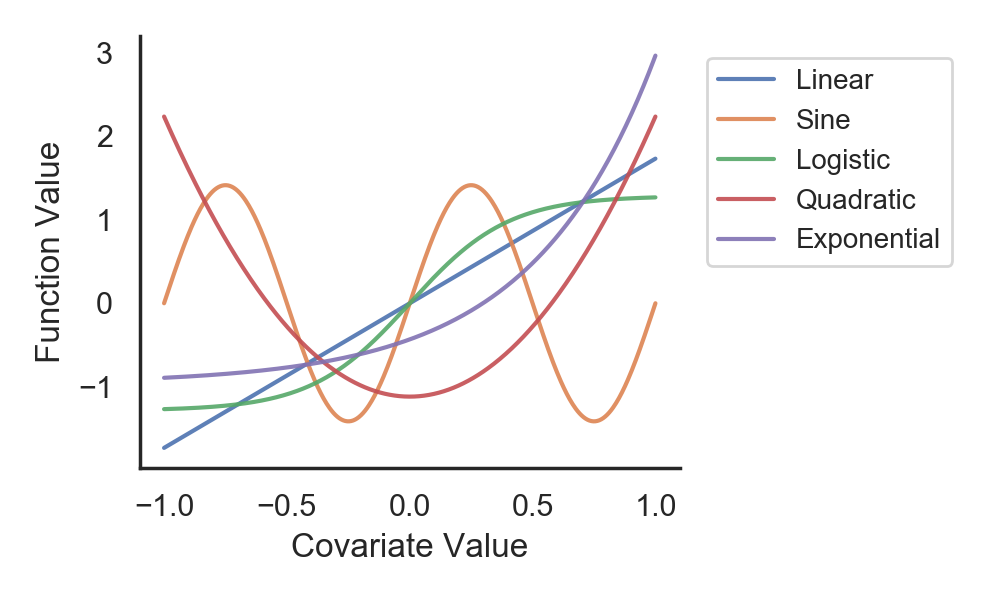}
        \end{subfigure}     
        	\caption{Test functions used to generate synthetic data}  \label{fig:ex_trends}
\end{figure}

\begin{Table}[H]
\centering
\caption{Variable Selection Performance for the Bike Sharing Data Set.}
\label{tab:bike_var_perform}
\renewcommand{\arraystretch}{.667}
\begin{tabular}{@{}cccc@{}}
\toprule
\textbf{Method} & \textbf{\# Covariates} & \textbf{\# Original Selected} & \textbf{\# Wrong Selected} \\ \midrule
\textbf{SKIM-FA}         & 1000                   & 3                             & 0                          \\
HierLasso       & 1000                   & 3                             & 5                          \\
SPAM-2Stage     & 1000                   & 3                             & 8                          \\
Pairs Lasso     & 1000                   & 3                             & 76                         \\
MARS            & 1000                   & 3                             & 119                        \\ \bottomrule
\end{tabular}
\end{Table}

\begin{Table}[H]
\centering
\caption{Estimation Performance for the Bike Sharing Data Set.}
\label{tab:bike_est_perform}
\renewcommand{\arraystretch}{.667}
\resizebox{\textwidth}{!}{%
\begin{tabular}{@{}ccccccccc@{}}
\toprule
\textbf{Method} & \textbf{\# Noise} & \textbf{\begin{tabular}[c]{@{}c@{}}Correct \\ Selected \\ SSE \\ (Main)\end{tabular}} & \textbf{\begin{tabular}[c]{@{}c@{}}Correct \\ Not \\ Selected \\ SSE \\ (Main)\end{tabular}} & \textbf{\begin{tabular}[c]{@{}c@{}}Wrong\\ Selected \\ SSE\\ (Main)\end{tabular}} & \textbf{\begin{tabular}[c]{@{}c@{}}Correct\\ Selected \\ SSE \\ (Pair)\end{tabular}} & \textbf{\begin{tabular}[c]{@{}c@{}}Correct\\ Not \\ Selected \\ SSE\\ (Pair)\end{tabular}} & \textbf{\begin{tabular}[c]{@{}c@{}}Wrong\\ Selected \\ SSE\\ (Pair)\end{tabular}} & \textbf{\begin{tabular}[c]{@{}c@{}}Total \\ SSE\end{tabular}} \\ \midrule
\textbf{SKIM-FA} & 1000 & 0.145 & 0.002 & 0 & 0.107 & 0.009 & 0 & 0.263 \\
SPAM-2Stage & 1000 & 0.149 & 0.002 & 0.027 & 0.081 & 0.009 & 0.000 & 0.269 \\
MARS-EMP & 1000 & 0.214 & 0.002 & 0.485 & 0.054 & 0.026 & 0.245 & 1.026 \\
MARS-Vanilla & 1000 & 6.556 & 0.002 & 0.796 & 0.947 & 0.026 & 1.882 & 10.209 \\ \bottomrule
\end{tabular}
}
\end{Table}

%

%

\clearpage

\section{Proofs} \label{A:proofs}
 
 \subsection{Proof of \cref{prop:kern_exist}} \label{A:proof_exists}
It suffices to prove that $\hilb_V^o$ is an RKHS. First we prove that $\hilb_V^o$ is a Hilbert space. Since $\hilb_V^o \subset \hilb_V$, it suffices to show that  $\hilb_V^o$ is a vector space and complete. To show that $\hilb_V^o$ is a vector space, take arbitrary $f, g \in \hilb_V^o$ and $\alpha, \beta \in R$. We want to show $\alpha f + \beta g \in  \hilb_V^o$. Take an arbitrary $f_A \in \hilb_A$, $A \subsetneq V$. Then,
 \begin{equation*}
 	\begin{split}
 		\inner{\alpha f + \beta g}{f_A}_{\mu} &= \alpha \inner{f}{f_A}_{\mu} + \beta \inner{g}{f_A}_{\mu} \\
 													&= 0
 	\end{split}
 \end{equation*}
 since $f, g \in  \hilb_V^o$. Hence, $\hilb_V^o$ is a vector space. 
 
Suppose towards a contradiction that $\hilb_V^o$ is not complete. Then, since $\hilb_V$ is complete, there exists an $f' \in \hilb_V \setminus \hilb_V^o$ and Cauchy sequence $\{f_n\}_{n=1}^{\infty}$ such that $\lim_{n \rightarrow \infty} \|f' - f_{n} \|_{\hilb_V}$ = 0, where $f_n \in \hilb_V^o$ and $\| \cdot \|_{\hilb_V}$ denotes the induced RKHS norm for $\hilb_V$. Then, there exists an $\epsilon > 0$ and $f_A \in \hilb_A$, $A \subsetneq V$ such that
\begin{equation}
	\begin{split}
			\epsilon & =  \inner{f'}{f_A}_{\mu} \\
						&=  \inner{f'  + f_m - f_m}{f_A}_{\mu} \\
						&= \inner{f'  - f_m}{f_A}_{\mu} + \inner{f_m}{f_A}_{\mu} \\
						&= \inner{f'  - f_m}{f_A}_{\mu} \\
						& \leq \|f' - f_m \|_{\mu} \|f_A\|_{\mu} \quad (\text{by Cauchy-Schwarz}).
	\end{split}
\end{equation}
To reach a contraction, it suffices to show that there exists an $m < \infty$ such that $\|f' - f_m \|_{\mu} < \frac{\epsilon}{\|f_A\|_{\mu}}$. To obtain this inequality, we upper bound $\| \cdot \|_{\mu}$ in terms of $\| \cdot \|_{\hilb_V}$.  Let $r_V$ be the reproducing kernel for $\hilb_V$. Then, for $f \in \hilb_V$,
\begin{equation}
	\begin{split}
		|f(x)|^2 &= |\inner{f}{r_V(x, \cdot)}_{\hilb_V}|^2 \quad (\text{by the reproducing property}) \\
					&\leq \|f\|_{\hilb_V}^2 r_V(x, x)^2 \quad (\text{by Cauchy-Schwarz}).
	\end{split}
\end{equation}
Then, 
\begin{equation}
	\begin{split}
		\|f \|_{\mu}^2 &= \int |f(x)|^2 d\mu \\
					&\leq  \|f\|_{\hilb_V}^2 \int  r_V(x, x)^2 d\mu
	\end{split}
\end{equation}
Since $\hilb_V$ belongs to the space of square integrable functions,  $\int  r_V(x, x)^2 d\mu = M_V < \infty$. Hence,
\begin{equation}
	 \|f' - f_m \|_{\mu} \leq M_V \|f' - f_m \|_{\hilb_V}^2 < \infty.
\end{equation}
Since $\|f' - f_m \|_{\hilb_V}^2 \rightarrow 0$, there exists an $m$ such that $\|f' - f_m \|_{\mu} < \frac{\epsilon}{\|f_A\|_{\mu}}$. Hence, $\hilb_V^o$ is complete.

To complete the proof it suffices to show that the evaluation functional on  $\hilb_V^o$  is a bounded operator. Since $\hilb_V$ is an RKHS there exists an $M_x < \infty$ such that for all $f \in \hilb_V$
\begin{equation}
	|f(x)| \leq M_x  \|f\|_{\hilb_V}. 
\end{equation}
Since $\hilb_V^o \subset \hilb_V$, then for all $g \in \hilb_V^o$,
\begin{equation}
	|g(x)| \leq M_x  \|g\|_{\hilb_V}. 
\end{equation}

\subsection{Proof of \cref{lem:kern_select}}
\begin{equation}
	\begin{split}
		f^{(M)}(x) &= \sum_{m=1}^M \alpha_m k_{\theta}(x_m, x) \\
						&= \sum_{m=1}^M \left ( \sum_{V: |V| \leq Q} \theta_V k_V(x_m, x) \right) \\
						&= \sum_{V: |V| \leq Q} \theta_V \left ( \sum_{m=1}^M   k_V(x_m, x) \right) \\
						&= \sum_{V: |V| \leq Q} f_V(x).
	\end{split}
\end{equation}
It remains to show that $f_V \in \hilb_V^o$. For all $m \in [M]$, $k_V(x_m, \cdot) \in \hilb_V^o$. Hence, $ \theta_V \sum_{m=1}^M   k_V(x_m, x) \in \hilb_V^o$ since $\hilb_V^o$ is a Hilbert space.

\subsection{Proof of \cref{lem:arb_far_intercept}} \label{A:far_intercept_proof}
We prove the claim using a constructive proof with $p=2$ variables. Consider the function
\begin{equation}
	f(x_1, x_2) = 1 + (x_1 - x_2)^{2k} I(|x_1| \leq M)  I(|x_2| \leq M).
\end{equation}
Suppose the joint distribution of  $(x_1, x_2)$ under $\mu$ equals 
\[\begin{pmatrix}
x_1 \\
x_2
\end{pmatrix}\sim N\left(\begin{pmatrix}
0 \\
0
\end{pmatrix},\begin{pmatrix}
1 & \rho \\
\rho & 1
\end{pmatrix}\right).
\]
Then, the joint distribution of $(x_1, x_2)$ under $\mu_{\tens}$ equals
\[\begin{pmatrix}
x_1 \\
x_2
\end{pmatrix}\sim N\left(\begin{pmatrix}
0 \\
0
\end{pmatrix},\begin{pmatrix}
1 & 0 \\
0 & 1
\end{pmatrix}\right).
\]
By symmetry, 
\begin{equation}  \label{eq:tens_bound}
	\begin{split}
\E_{\mu_{\tens}}[f(x_1, x_2)] &= 2 \mu_2(x_2 < 0)  \E_{\mu_{\tens}}[f(x_1, x_2) \mid x_2 < 0] \\
											& \geq 2  \mu_1(x_1 > c) \mu_2(x_2 < 0)  \E_{\mu_{\tens}}[f(x_1, x_2) \mid x_1 > c, x_2 < 0] \\
											&=  \mu_1(x_1 > c) \E_{\mu_{\tens}}[f(x_1, x_2) \mid x_1 > c, x_2 < 0] \\
											& \geq \mu_1(x_1 > c) c^{2k} I(|c| < M).
	\end{split}
\end{equation}
Under $\mu$, we may assume without loss of generality that
\begin{equation*}
	\begin{split}
			x_1 &\sim \normal(0, 1) \\
			\epsilon &\sim \normal(0, 1) \quad \text{s.t.} \quad \epsilon \indep x_1 \\
			x_2 &= \rho x_1 + \sqrt{1 - \rho^2} \epsilon.
	\end{split}
\end{equation*}
Then, 
\begin{equation}
	\begin{split}
			\lim_{\rho \rightarrow 1} \E_{\mu} f(x_1, x_2) &= 1 + \lim_{\rho \rightarrow 1} \int (x_1 - x_2)^{2k} I(|x_1| \leq M)  I(|x_2| \leq M) d\mu(x_1, x_2) \\
																			&= 1 + \lim_{\rho \rightarrow 1} \int (x_1 - \rho x_1 - \sqrt{1 - \rho^2} \epsilon)^{2k} I(|x_1| \leq M)  I(|x_2| \leq M) d\mu_{\tens}(x_1, \epsilon) \\
																			&= 1 + \int \lim_{\rho \rightarrow 1} (x_1 - \rho x_1 - \sqrt{1 - \rho^2} \epsilon)^{2k} I(|x_1| \leq M)  I(|x_2| \leq M) d\mu_{\tens}(x_1, \epsilon) \\
																			&= 1,
	\end{split}
\end{equation}
where the second to last line follows from he Dominated Convergence Theorem since $f(x_1, x_2)$ is uniformly bounded by $(2M)^{2k}$. Since $\E_{\mu} f(x_1, x_2) > 1$ for $0 \leq \rho < 1$, there exists a sequence $\{\rho_k\}_{k=1}^{\infty}$ such that for all $k \in \mathbb{N}$, $1 < \E_{\mu_k} f(x_1, x_2) < 2$ and $0 < \rho_k < 1$, where $\mu_k$ sets $\rho = \rho_k$. Pick $k'$ large enough so that $ f_{\{\emptyset\}}^{\mu_{\tens}} > 2$.  Then, for $k \geq k'$,
\begin{equation} \label{eq:bound}
	\begin{split}
		\frac{|  f_{\{\emptyset\}}^{\mu_{\tens}} - f_{\{\emptyset\}}^{\mu_k} |}{ |f_{\{\emptyset\}}^{\mu_k}|} &\geq \frac{| f_{\{\emptyset\}}^{\mu_{\tens}} - f_{\{\emptyset\}}^{\mu_k} |}{ 2} \\
			& = \frac{ f_{\{\emptyset\}}^{\mu_{\tens}} - f_{\{\emptyset\}}^{\mu_k} }{ 2} \\
			& > \frac{ f_{\{\emptyset\}}^{\mu_{\tens}} -2 }{ 2}
	\end{split}
\end{equation}
Let $k^* = \max\left(k', \left \lceil .5\sqrt[c]{\frac{2(\Delta + 1)}{\mu_1(x_1 > c} } \right\rceil \right)$. Then, by \cref{eq:tens_bound} and \cref{eq:bound}, $\frac{|  f_{\{\emptyset\}}^{\mu_{\tens}} - f_{\{\emptyset\}}^{\mu_{k^*}} |}{ |f_{\{\emptyset\}}^{\mu_{k^*}}|}  > \Delta$.

\subsection{Proof of \cref{prop:model_select_bayes}} \label{A:reparam_proof}
By equation 2.25 of \citet[Chapter 2]{gp_book}, \cref{eq:kern_ridge} equals the posterior predictive mean of the following Bayesian model:
\begin{equation*}
	\begin{split}
		f &\sim GP(0, k_{\theta}) \\
		y \mid f, x &\sim \normal(f(x), \noisevar = \lambda).
	\end{split}
\end{equation*}
We may re-write $k_{\theta}$ as,
\begin{equation*}
	\begin{split}
		k_{\theta}(x, \tilde{x})  &= \sum_{V: |V| \leq Q} \theta_V \Phi_V^T(x) \Phi_V^T(\tilde{x}) \\
										 &=  \sum_{V: |V| \leq Q} \Phi_V^T(x)[ \theta_V I_{B^V \times B^V}  ]\Phi_V^T(\tilde{x}) \\
										 &= \sum_{V: |V| \leq Q} \Phi_V^T(x)\Sigma_V\Phi_V^T(\tilde{x}),
	\end{split}
\end{equation*}
where $\Sigma_V =  \theta_V I_{B^V \times B^V} $. Then, by \citet[Chapter 2.1.2]{gp_book} and the additive property of kernels, $f \sim GP(0, k_{\theta})$ has the same distribution as drawing a set of regression coefficients $\Theta_V \sim \normal(0, \Sigma_V)$ and setting $f = \sum_{V: |V| \leq Q} \Theta_V^T \Phi_V(\cdot)$. Hence, the posterior predictive mean of the Gaussian process at a point $x$ equals $\sum_{V: |V| \leq Q} \hat{\Theta}_V^T \Phi_V(x)$.

\subsection{Proof of \cref{thm:skim_fast_form}} \label{A:kern_trick_two_proof}
\begin{equation} \label{eq:skim_to_vap}
	\begin{split}
		k_{\text{SKIM-FA}}(x, \tilde{x}) &= \sum_{V: |V| \leq Q} \left[  \eta_{|V|}^2 \prod_{i \in V} \kappa_i^2 \right] k_V(x, \tilde{x}) \\
							&= \sum_{V: |V| \leq Q} \left[  \eta_{|V|}^2 \prod_{i \in V} \kappa_i^2 \right] \prod_{i \in V} k_i(x_i, \tilde{x}_i) \\
							&= \sum_{V: |V| \leq Q} \left[  \eta_{|V|}^2 \prod_{i \in V} \kappa_i^2 k_i(x_i, \tilde{x}_i) \right]  \\
							&= \sum_{q=1}^Q \sum_{V: |V| = Q} \left[  \eta_{|V|}^2 \prod_{i \in V} \kappa_i^2 k_i(x_i, \tilde{x}_i) \right] \\
							&= \sum_{q=1}^Q \eta_{q}^2 \sum_{V: |V| = q} \left[ \prod_{i \in V} \kappa_i^2 k_i(x_i, \tilde{x}_i) \right]	 
	\end{split}
\end{equation}
Let $\tilde{k}_i(\cdot, \cdot) = \kappa_i^2 k_i(\cdot, \cdot)$. Then,  \citet[pg. 199]{vapnik95} shows that
\begin{equation} \label{eq:vap_form}
	\bar{k}_q \coloneqq \sum_{V: |V| = q} \prod_{i \in V} \tilde{k}_i = \frac{1}{q} \sum_{s=1}^q (-1)^{s+1} \bar{k}_{q-s} k^s,	
\end{equation}
where $k^s(x, \tilde{x}) = \sum_{i=1}^p [\tilde{k}_i(x_i, \tilde{x}_i)]^s$ and $\bar{k}_0(x, \tilde{x}) = 1$. The result follows from \cref{eq:skim_to_vap} and \cref{eq:vap_form}.

\subsection{Proof of \cref{corr:skim_kern_time}} \label{A:recur_proof}

Computing and storing $k^{1}(x, \tilde{x}), \cdots, k^{Q}(x, \tilde{x})$ takes $O(pQ)$ time and requires $O(Q)$ memory, respectively. After computing and storing $\bar{k}_1(x, \tilde{x}), \cdots, \bar{k}_{q}(x, \tilde{x})$, $\bar{k}_{q+1}(x, \tilde{x})$ takes $O(q+1)$ time. Hence, computing  all $\bar{k}_1(x, \tilde{x}), \cdots, \bar{k}_{Q}(x, \tilde{x})$ terms takes $O(Q^2)$ time given $k^{1}(x, \tilde{x}), \cdots, k^{Q}(x, \tilde{x})$. Since $Q < p$, computing  $k_{\text{SKIM-FA}}(x, \tilde{x})$ takes $O(pQ)$ time.

\subsection{Proof of \cref{prop:grad_descent}}
\begin{equation*}
	\begin{split}
		\frac{\partial L}{\partial \tilde{U}_i^{(t)}} &= \frac{\partial L}{\partial \kappa_i^{(t)}}  \frac{\partial \kappa_i}{\partial U_i^{(t)}} \frac{\partial U_i^{(t)}}{\partial \tilde{U}_i^{(t)}} \\
															&= \frac{\partial L}{\partial \kappa_i^{(t)}}  I(U_i^{(t)} > c) \frac{2 \tilde{U}_i^{(t)}}{(\tilde{U}_i^{(t)} + 1)^2}.		
	\end{split}
\end{equation*}
Since $\kappa_i^{(t)} = 0$, that implies $U_i^{(t)} \leq c$. Hence, $\frac{\partial L}{\partial \tilde{U}_i^{(t)}} = 0$. Consequently, 
\begin{equation} \label{eq:fixed_point}
	\begin{split}
		\tilde{U}_i^{(t + 1)} &= \tilde{U}_i^{(t)} - \gamma \frac{\partial L}{\partial \tilde{U}_i^{(t)}} \\
									&= \tilde{U}_i^{(t)}.
	\end{split}
\end{equation}
By \cref{eq:fixed_point}, $\kappa_i^{(t')} = 0$ for all $t' \geq t$.

\subsection{Proof of \cref{eq:universal_space}}
 It suffices to prove that any $f_V \in \hilb_V$ is square-integrable  with respect to  any probability measure. Since $\phi_{ib}$ is a continuous function on a compact set, there exists a $0 < M_{ib} < \infty$ such that $|\phi_{ib}|$ is bounded by $M_{ib}$. Without loss of generality, assume $V = \{1, \cdots, q\}$. Then, there exists coefficients $c_{b_1, \cdots, b_q} \in \R$ such that
\begin{equation*}
	\begin{split}
		f_V(x_V) &= \sum_{b_1 \in [B_1]} \cdots \sum_{b_{q} \in [B_q]} c_{b_1, \cdots, b_q}  \prod_{i = 1}^q  \phi_{ i b_i}(x_i) \\
					& \leq \sum_{b_1 \in [B_1]} \cdots \sum_{b_{q} \in [B_q]} c_{b_1, \cdots, b_q} M_*^q \\
					& < \infty
	\end{split}
\end{equation*}
for all $x_V$, where $M_* = \max_{i \in [p]} \max_{b \in [B_i]} M_{ib} < \infty$ since $B_i < \infty$. Hence, for any probability measure $\mu$, 
\begin{equation}
\begin{split}
	\int |f_V(x_V)|^2 d\mu &< \int \left( \sum_{b_1 \in [B_1]} \cdots \sum_{b_{q} \in [B_q]} c_{b_1, \cdots, b_q} M_*^q \right)^2 d\mu \\
					& = \left( \sum_{b_1 \in [B_1]} \cdots \sum_{b_{q} \in [B_q]} c_{b_1, \cdots, b_q} M_*^q \right)^2 \\
					& < \infty.
\end{split}
\end{equation}

\subsection{Proof of \cref{thm:theo_proj}} \label{A:theo_proj_proof}
Let 
\begin{equation}
	 \begin{split} 
	 	\tilde{f}_{ij} &=  f_{\{ i, j \}}^{\mu_{\tens}} - [\Psi_{ij}^i \Phi_i +  \Psi_{ij}^j \Phi_j + \Psi_{ij}^{0}] \\
	 	\tilde{f}_i  &=f_{\{ i \}}^{\mu_{\tens}} + \sum_{j > i} \Psi_{ij}^i \Phi_i + \sum_{j < i} \Psi_{ji}^{i} \Phi_i(x_i)  \\
	 	\tilde{f}_{\emptyset}  &= f_{\emptyset}^{\mu_{\tens}} + \sum_{i < j} \Psi_{ij}^{0}.
	 \end{split}
\end{equation}
We start by proving that $f =\tilde{f}_{\emptyset} + \sum_{i=1}^p \tilde{f}_i  + \sum_{i,j=1}^p \tilde{f}_{ij}$. Expanding each component,
\begin{equation*}
	\begin{split}
		\tilde{f}_{\emptyset} + \sum_{i} \tilde{f}_i + \sum_{i < j} \tilde{f}_{ij} &= \tilde{f}_{\emptyset} + \sum_{i} \tilde{f}_i + \sum_{i < j} \left[f_{\{ i, j \}}^{\mu_{\tens}} - [\Psi_{ij}^i \Phi_i +  \Psi_{ij}^j \Phi_j + \Psi_{ij}^{0}] \right] \\
		&=  f_{\emptyset}^{\mu_{\tens}} + \sum_{i} \tilde{f}_i + \sum_{i < j} \left[f_{\{ i, j \}}^{\mu_{\tens}} - [\Psi_{ij}^i \Phi_i +  \Psi_{ij}^j \Phi_j] \right] \\
		&=  f_{\emptyset}^{\mu_{\tens}} + \sum_{i}  \left[ f_{\{ i \}}^{\mu_{\tens}} + \sum_{j > i} \Psi_{ij}^i \Phi_i + \sum_{j < i} \Psi_{ji}^i \Phi_i \right] + \\ & \qquad \sum_{i < j} \left[f_{\{ i, j \}}^{\mu_{\tens}} - [\Psi_{ij}^i \Phi_i +  \Psi_{ij}^j \Phi_j] \right] \\
		&=  f_{\emptyset}^{\mu_{\tens}} + \sum_{i} f_{\{ i \}}^{\mu_{\tens}} + \sum_{i < j}f_{\{ i, j \}}^{\mu_{\tens}} + \\
			& \qquad  \sum_{i}  \sum_{j > i} \Psi_{ij}^i \Phi_i +  \sum_{i} \sum_{j < i} \Psi_{ji}^i \Phi_i -  \sum_{i < j} \left[\Psi_{ij}^i \Phi_i +  \Psi_{ij}^j \Phi_j \right]  \\
		&= f +  \sum_{i}  \sum_{j > i} \Psi_{ij}^i \Phi_i +  \sum_{i} \sum_{j < i} \Psi_{ji}^i \Phi_i - \sum_{i}  \sum_{j > i} \left[\Psi_{ij}^i \Phi_i +  \Psi_{ij}^j \Phi_j \right] \\
		&= f + \sum_{i} \sum_{j < i} \Psi_{ji}^i \Phi_i -  \sum_{i}  \sum_{j > i} \Psi_{ij}^j \Phi_j \\
		&= f + \sum_{i} \sum_{j < i} \Psi_{ji}^i \Phi_i -  \sum_{j}  \sum_{j < i} \Psi_{ji}^i \Phi_j \\
		&= f.
	\end{split}
\end{equation*}
We now prove that there exists unique coefficients, $\Psi_{ij}^i \in \R^{1 \times B_i}, \Psi_{ij}^j  \in \R^{1 \times B_j}, \Psi_{ij}^{0} \in \R$, such that $\tilde{f}_{ij}$ belongs to the orthogonal complement of the Hilbert space $\hilb^{\text{add}}_{\{i, j\}} \coloneqq \text{span}\{1, \{\phi_{ib}\}_{b=1}^{B_i}, \{\phi_{jb} \}_{b=1}^{B_j} \} =  \hilb_{\emptyset} \bigoplus \hilb_{\{i\}} \bigoplus  \hilb_{\{j\}}$. Recall that  $$\hilb_{\{i,j\}} = \text{span}\{1, \{\phi_{ib}\}_{b=1}^{B_i}, \{\phi_{jb}\}_{b=1}^{B_j} \},  \{ \phi_{ib}\phi_{jb'} \}_{b \in [B_i], b' \in [B_j]}\}.$$ Then, $f_{\{ i, j \}}^{\mu_{\tens}}  \in \hilb_{\{i,j\}}$ and $\hilb^{\text{add}}_{\{i, j\}}$ is a closed convex subspace of $\hilb_{\{i,j\}}$. Therefore, by the Hilbert Projection Theorem, there exists unique $\bar{f}_{ij} \in \hilb^{\text{add}}_{\{i, j\}}$ and ${f}^{\perp}_{ij} \in \hilb_{\{i,j\}}$ such that
\begin{equation} \label{eq:hilb_proj}
	\begin{split}
		f_{\{ i, j \}}^{\mu_{\tens}} &= \bar{f}_{ij} + {f}^{\perp}_{ij} \quad \text{s.t.} \\
		& \inner{g}{{f}^{\perp}_{ij}}_{\mu} = 0 \quad \forall g \in \hilb^{\text{add}}_{\{i, j\}}.
	\end{split}
\end{equation}
Since $\text{span}\{1, \{\phi_{ib}\}_{b=1}^{B_i}, \{\phi_{jb} \}_{b=1}^{B_j} \}$ is a linearly independent basis of $\hilb^{\text{add}}_{\{i, j\}}$, there exists unique coefficients, $\Psi_{ij}^i \in \R^{1 \times B_i}, \Psi_{ij}^j  \in \R^{1 \times B_j}, \Psi_{ij}^{0} \in \R$, such that $\bar{f}_{ij} = \Psi_{ij}^i \Phi_i^T(x_i) +  \Psi_{ij}^j \Phi_j^T(x_j) + \Psi_{ij}^{0}$. 
 
To complete the proof, we need to show that $\tilde{f}_{ij} =  f_{\{ i, j \}}^{\mu}, \tilde{f}_i  =  f_{\{ i\}}^{\mu},   \tilde{f}_{\emptyset} = f_{\emptyset}^{\mu}$. It suffices to show that  
\begin{equation} \label{eq:check_orth}
	\begin{split}
		& \int_{x_i} \tilde{f}_i \ d\mu_i = 0  \\
		& \int_{x_i, x_j} \tilde{f}_{ij} \ d\mu_i = 0  \\
		& \int_{x_i, x_j} \tilde{f}_i  \tilde{f}_{ij} \ d\mu(x_i, x_j) = 0.
	\end{split}
\end{equation}
The last two equalities in \cref{eq:check_orth} follow directly from \cref{eq:hilb_proj}. For the first equality in \cref{eq:check_orth}, notice that 
\begin{equation*}
	\begin{split}
		 \int_{x_i} \tilde{f}_i \ d\mu_i  &= \E_{\mu_i} \tilde{f}_i \\
		 											&= \E_{\mu_i} \left[  f_{\{ i \}}^{\mu_{\tens}} + \sum_{j > i} \Psi_{ij}^i \Phi_i + \sum_{j < i} \Psi_{ji}^{i} \Phi_i \right] \\
		 											&=  \E_{\mu_i}  f_{\{ i \}}^{\mu_{\tens}} + \sum_{j > i} \E_{\mu_i}[\Psi_{ij}^i \Phi_i] + \sum_{j < i} \E_{\mu_i} [\Psi_{ji}^{i} \Phi_i] \\
		 											&= \sum_{j > i} \Psi_{ij}^i \E_{\mu_i}[\Phi_i] + \sum_{j < i}  \Psi_{ji}^{i} \E_{\mu_i}[\Phi_i]  \\
		 											&= 0,
	\end{split}
\end{equation*}
where the last equation follows from the fact that the components of $\Phi_i$ span $\hilb_{\{i\}}^o$ (and hence are all zero mean).

\subsection{Proof of \cref{prop:change_basis_correct}}
As shown in the proof of  \cref{thm:theo_proj}, $\Psi_{ij}^i \in \R^{1 \times B_i}, \Psi_{ij}^j  \in \R^{1 \times B_j}, \Psi_{ij}^{0} \in \R$ equal the unique set of coefficients such that $\bar{f}_{ij} = \Psi_{ij}^i \Phi_i +  \Psi_{ij}^j \Phi_j + \Psi_{ij}^{0}$ for $\bar{f}_{ij}$ defined in \cref{eq:hilb_proj} and also shown below:
\begin{equation*} 
	\begin{split}
		f_{\{ i, j \}}^{\mu_{\tens}} &= \bar{f}_{ij} + {f}^{\perp}_{ij} \quad \text{s.t.} \\
		& \inner{g}{{f}^{\perp}_{ij}}_{\mu} = 0 \quad \forall g \in \hilb^{\text{add}}_{\{i, j\}}.
	\end{split}
\end{equation*}
Let $y_{ij}^{(w)} = f_{\{ i, j \}}^{\mu_{\tens}}(x_i^{(m)}, x_j^{(w)})$ and $\epsilon_{ij}^{(w)} = {f}^{\perp}_{ij}(x_i^{(w)}, x_j^{(w)})$, where $x^{(w)} \iid \mu$. Then, 
\begin{equation} \label{eq:rand_des_ols}
	\begin{split}
		y_{ij}^{(w)} &= \Psi_{ij}^i \Phi_i(x_i^{(w)}) +  \Psi_{ij}^j \Phi_j(x_j^{(w)}) + \Psi_{ij}^{0} +  \epsilon_{ij}^{(w)} \quad x^{(w)} \iid \mu.
	\end{split}
\end{equation}
Then, \cref{eq:rand_des_ols} is a special case of the random design linear model under misspecification studied in \citet{hsu2014random}. Hence, by \citet[Theorem 11]{hsu2014random}  we can consistently recover $\Psi_{ij}^i$, $\Psi_{ij}^j$, $ \Psi_{ij}^{0}$ by using ordinary least-squares. Hence  \cref{algo:change_basis} recovers  $\Psi_{ij}^i$, $\Psi_{ij}^j$, $ \Psi_{ij}^{0}$ as $W \rightarrow \infty$.

\section{Literature Review} \label{A:lit_review} 
\noindent \textit{Finite Basis Expansion Methods.} \citet{stone1994} introduced the hierarchical functional decomposition and derived statistical rates of convergence by approximating $\hilb$ using a finite B-spline tensor product basis. \citet{projection_anova} later extended this result to general tensor product families such as wavelets, polynomials, etc. There have been a number of specific Bayesian and frequentist methods that fall within the general class of models described in \citet{projection_anova}; see, for example, \citet{sparse_bayes_pairwise, spike_slab_additive, Curtis13fastbayesian, dunson_gp_inter, bayes_smooth}. Unfortunately, since these methods explicitly generate the tensor product basis, they are computationally intractable as $p$ increases beyond a few hundred or thousand covariates. \edit{In \citet{radchenko_var}, the authors consider the $Q=2$ setting, and develop the VANISH algorithm to fit nonlinear interaction models under a heredity constraint (i.e., interaction terms are only added if the main effects are selected). They authors use a finite basis to model the main and interaction effects but do not assume a tensor product basis. Unfortunately, their method does not scale well with larger $p$ since the runtime is $O(p^2)$; see Step 0 of the VANISH algorithm on page 6. In  \citet{haris_herid}, the authors generalize VANISH (and several other algorithms) by  using an alternating directions method of multipliers algorithm to fit interactions.}
 
Linear models trivially fall within this class as well. For $Q=1$, the Lasso and the many related techniques provide fast variable selection and estimation in high-dimensional linear models \citep{atom_pursuit, dantzig, pruning_lasso}. For $Q=2$, the hierarchical Lasso \citep{lass_heirch} extends the Lasso to model interactions, and there have been many variants of this model; see, for example, \citet{lim_heirch_lasso, backtrack_heirch_lasso}. However, these methods take at least $O(p^2)$ time since they explicitly model all main and interaction effects. Other linear interaction methods assume that the interactions have a low-rank structure. This structure helps both statistically and computationally; see, for example, \citet{fact_machines, dunson_factor}. However, this low-rank structure in the interaction effects might not always hold in practice. \\

\noindent \textit{Two-Stage \& Forward-Stage Approaches.} Instead of modeling interactions jointly, a common heuristic (similar in spirit to forward stepwise regression) is greedily adding interactions such as in multivariate additive regression splines (MARS) or GA2M \citep{intelligible_model_caruana}. \edit{The iFORM algorithm proposed in \citet{hao_high} is the middle-ground between MARS and fitting a model with all interactions terms included at the start. Specifically, iFORM starts with the empty model, and then proceeds by adding one more predictor at a time, where all interactions between the current active set of predictors are considered.} Another common approach is performing computationally cheap variable selection methods designed for generalized additive models (e.g., Lasso or SpAM \citep{spam}) to identify a sparse set of relevant variables. By restricting to a small set of variables, one can then apply more computationally intensive interactions techniques such as RKHS ANOVA methods.

The approaches above requires some form of strong-hierarchy, namely that all interactions have non-zero main effects, to consistently identify the correct set of variables. While some problems have strong main effects, in other applications this may not be the case. For example, in genome-wide associate studies, fitting an additive-only model to predict an individual's height from genetics only has an $R^2$ of about $5\%$ even though height is well-predicted by parents' heights (thought to be between $80\%-90\%$) \citep{missing_herit}. This discrepancy, more generally called the problem of \emph{missing heritability}, remains an open challenge in biology for understanding complex diseases based on genetics. One explanation for missing heritability is not modeling genetic interactions \citep{missing_herit, gwas_interaction, epistatis_gene2, gpu_epistasis}. In other words, the main effects might be weak, or in the extreme case some genes might only have interaction effects. Hence, from a purely variable selection standpoint, modeling interactions could help better identify genes that are risk-factors for certain diseases. 

In the orthogonal $\mu$ case, the statistical benefit from modeling interactions can be easily seen from the decomposition in \cref{eq:sig_var_decomp}. Suppose $Q=2$ and that main effects total signal variance equals $5$, the pairwise signal variance equals equals $90$, and the noise variance equals $5$. Then, the $R^2$ for an additive-only model is $5\%$ while the $R^2$ for interaction model is $90\%$.  Since the achievable signal increases (and necessarily the effective noise variance decreases), performing variable selection in a lower signal-to-noise regime might offset the statistical price of modeling more parameters. \\

\noindent \edit{\textit{Tree-Based Approaches.}  Tree-based methods (e.g., random forest, CART, gradient boosting) are often used for black-box prediction tasks. While these methods sometimes provide variable importance measures, it is unclear how to access the effects from the fitted prediction function and perform variable selection. Nevertheless, some authors have adapted tree-based methods to estimate effects and perform variable selection. For example, in \citet{bart_var}, the authors modify the Bayesian additive regression trees method in \citet{bart} by placing Dirichlet priors on the splitting proportions of the regression tree prior to induce sparsity. While \citet{bart_var} perform variable selection by looking at posterior inclusion probabilities, it is unclear how to access the interaction effects from the fit. Other authors have adapted tree-based methods to estimate heterogeneous treatment effects (i.e., interactions between a treatment and set of covariates). For example, in \citet{iter_trees}, the authors modify the CART splitting rule to get better estimates of heterogeneous treatment effects. Similarly, in \citet{tree_medicine}, the authors use a model-based recursive partitioning to identify covariates that most strongly interact with a particular treatment. In general, estimating treatment-by-covariate interactions is less computationally demanding since there are only $O(p)$ number of pairwise interactions between a single treatment and all covariates.  \\ }


\noindent \edit{\textit{Kernel Methods.} Many of the functional ANOVA methods described in \cref{sec:related_anova_work} use kernels to model nonlinear interaction effects. Unfortunately, as we discuss in \cref{sec:related_anova_work}, these methods are computationally intractable for moderate $p$ when $Q > 1$.  Some kernel methods used to identify interactions fall under the general area of "multiple kernel learning," where the goal is to learn some weighted combination of kernels \citep{mult_kern_one, mult_kern_bach}. \emph{Hierarchical kernel learning}, for example, is a multiple kernel learning method that  learns nonlinear interactions via hierarchy conditions encoded in a directed acyclic graph \citep{back_kern_heir}. This hierarchy condition translates into a method similar to the greedy forward-stage methods discussed above where higher-order interactions are added only when all lower-interactions are present. For $R$ selected kernels, the runtime stated in \citet{back_kern_heir} is $O(N^3R+ N^2Rp^2 +N^2R^2p)$. The quadratic dependence on $p$ makes this method unsuitable for larger $p$ problems. \citet{comp_kernel} considers a greedy kernel search based on adding and multiplying a base set of kernels together. Since the  multiplication of two kernel corresponds to an interaction, this method has the flexibility to model interaction effects. However, the focus in \citet{comp_kernel} is on prediction and understanding the structure of the fitted kernel instead of estimating effects and performing variable selection. \\}

\noindent \textit{Comparison with \citet{kit}.} The structure of the SKIM-FA prior in \cref{eq:skim_fa_weight} generalizes the prior used in \citet{kit} to handle non-linear effects. However, in \citet{kit}, the authors use a regularized horseshoe prior to achieve sparsity in $\kappa$. While a regularized horseshoe prior does not lead to exact sparsity in $\kappa$, the authors in \citet{kit} were still able to develop an $O(p)$ variable selection procedure by exploiting strong-hierarchy, namely that interactions only occur among selected main effects.  In the current work, we do not make any strong-hierarchy assumption. Hence, to develop an $O(p)$ variable selection procedure, we need \emph{exact} sparsity in $\kappa$; see \cref{sec:alg_implement} for details.

In terms of computational complexity, \citet{kit} fit linear interaction models in $O(pN^2 + N^3)$ time per iteration, which has the same asymptotic complexity as \cref{algo:learn_skim_hyp}. However, they use Hamiltonian Monte Carlo (HMC) to perform inference. Each HMC step requires computing and inverting an $N\times N$ kernel matrix many times. Hence, their method takes hours to complete when $p$ and $N$ are larger than 500. Due to this computational intensity, we do not benchmark against their method.


\section{Zero Mean Kernels and Finite-Basis Functions} \label{A:model_implement_details}
In this section, we show how we construct $k_i$, i.e., the reproducing kernel for $\hilb_{\{i\}}^o$. We construct $k_i$ by first generating a finite-dimensional basis for $\hilb_{\{i\}}$. Then, we normalize each basis function to be zero mean and unit variance so that the normalized basis functions span $\hilb_{\{i\}}^o$. For a more general approach to construct zero mean kernels (e.g., even when $\hilb_{\{i\}}$ is infinite-dimensional) see \citet{anova_zero_mean_sens}. \\

\noindent \textit{Construction of $\hilb_{\{i\}}^o$.}
For each covariate dimension $i$, consider a set of linearly independent basis functions $\{ \phi_{ib} \}_{b=1}^{B_i}$ such that
\begin{equation*}
	\begin{split}
	& \hilb_{\{i\}} = \text{span} \{1, \phi_{i1}, \cdots, \phi_{iB_i} \}. 
	\end{split}
\end{equation*}
Let $\tilde{\phi}_{ib} = \frac{\phi_{ib} - \E_{\mu}[\phi_{ib}]}{\sqrt{\text{Var}_{\mu}[\phi_{ib}] }}$. Then, 
\begin{equation*}
	\begin{split}
	& \hilb_{\{i\}}^o = \text{span} \{\tilde{\phi}_{i1}, \cdots, \tilde{\phi}_{iB_i} \}, \quad \Phi_i \coloneqq [\tilde{\phi}_{i1}, \cdots, \tilde{\phi}_{iB_i}].
	\end{split}
\end{equation*}
Hence, $k_i(x_i, \tilde{x}_i) =\Phi_i(x_i)^T \Phi_i(\tilde{x}_i)$ is the reproducing kernel for $\hilb_i^o$. In many instances, we do not actually know the joint distribution of the covariates. In this case, we approximate $\mu$ with the empirical distribution $\hat{\mu}$ of the datapoints:
\begin{equation*}
	\begin{split}
	 \tilde{\phi}_{ib} &= \frac{\phi_{ib} - \E_{\hat{\hat{\mu}}}[\phi_{ib}]}{\sqrt{\text{Var}_{\hat{\mu}}[\phi_{ib}] }} \quad \text{s.t.} \\
	& \quad \hat{\mu} = \frac{1}{N} \sum_{n=1}^N \delta_{x^{(n)}} \\
	& \quad \E_{\hat{\hat{\mu}}}[\phi_{ib}] = \frac{1}{N} \sum_{n=1}^N \phi_{ib}(x_i^{(n)}) \\
	& \quad\text{Var}_{\hat{\mu}}[\phi_{ib}]= \frac{1}{N} \sum_{n=1}^N \phi_{ib}^2(x_i^{(n)}) -  \E_{\hat{\hat{\mu}}}[\phi_{ib}] .
	\end{split}
\end{equation*}
%

\noindent \textit{Note}. for the experiments in \cref{sec:experiments}, we use a natural cubic spline basis with 5 knots at the quantiles to generate each $\hilb_{\{i\}}$ for SKIM-FA and SPAM-2Stage.  \\

\noindent \textit{Practical Considerations for Picking Basis Functions.} While any set of basis functions can be used to generate $\hilb_{\{i\}}^o$ in principle, we provide several suggestions below; see Chapter of \citet{esl} for a more in-depth review. 
\begin{itemize}
	\item If $x_i$ is a seasonal covariate, use a wavelet basis to generate   $\hilb_{\{i\}}^o$.
	\item If $x_i$ is categorical, let $\hilb_{\{i\}}^o$ equal the one-hot encoding of $x_i$.
	\item if $x_i$ is continuous, use a polynomial spline basis. 
		\begin{itemize}
			\item If extrapolation beyond the data is a concern, use a natural cubic spline basis (which enforces linearity beyond the boundary knots).
		\end{itemize}
\end{itemize}

\section{Additional Algorithmic Details} \label{A:add_algo_details} 
To fit SKIM-FA, we set the number of iterations $T=2000$,  learning rate $\gamma=0.1$, and cross-validation batch size $M = 0.2 N$ in \cref{algo:learn_skim_hyp}. We let the truncation level $c$ in  \cref{algo:learn_skim_hyp} depend on the iteration number $t$ for $1 \leq t \leq T$. Empirically, we find gradually increasing $c$ as a function of $t$ works well for accurately selecting the correct covariates and inducing sparsity in $\kappa$. As outlined in \cref{algo:trun_schedule}, we suggest the following schedule for $c$:

 \begin{algorithm}
    \caption{Scheduler for Truncation Level $c$}
    \label{algo:trun_schedule}
    \begin{algorithmic}[1] 
         \Procedure{TruncScheduler}{$U^{(t)}$, $c_{t-1}$, $t$, $r=.01$, $\gamma=.75$} 
         \If{$t < 500$}
         	\Return 0
         \EndIf
         
         \If{$t = 500$}
       		\Return $q_{25}(U^{(t)}_1, \cdots, U^{(t)}_p)$ \Comment{Take the 25th\% of the components in $U^{(t)}$}
       	\EndIf
       	\If{$t > 500$}
		\Return $\max(\min( (1+r) c_{t-1}, \ \gamma), \ c_{t-1})$
		\EndIf
        \EndProcedure
    \end{algorithmic}
\end{algorithm}

When $t<500$, $c=0$ in \cref{algo:trun_schedule}. Hence, since $\kappa_i^{(t)} = \max(U_i^{(t)} - c, 0)$, $\kappa_i^{(t)} \neq 0$ for $t \leq 500$ and $i \in [p]$. At iteration 500, we drop the bottom 25th\% of covariates (determined by their importance measure $U_i^{(t)}$). Specifically, $\kappa^{(500)}$ has 25\% of its entries equal to zero. For subsequent iterations, we take the previous trunction level $c_t$ and increase that level by a factor of $(1+r)$ until $c_t$ reaches $\gamma$. If $c_{t-1}$ is larger than $\gamma$, we set $c_t$ equal to $c_{t-1}$.

\section{SKIM-FA Extensions} \label{A:skim_extend}
\color{editcolor}

\subsection{Beyond Gaussian Responses}
Throughout we have assumed that $y \mid x$ is drawn from a Gaussian distribution. In general, we may assume that the response belongs to an exponential family, which allows us to model, for example, count, binary, and exponential response data. For non-Gaussian responses, there does not exist an analytical solution to \cref{eq:pen_ls}. Nevertheless, a combination of reweighted least squares and the Newton Raphson method can be used to iteratively solve \cref{eq:pen_ls} when $\hilb$ is an RKHS; see \citet{gen_kern_machines} for details. Since the results in \citet{gen_kern_machines} are not unique to the specific kernel used, we can extend SKIM-FA beyond Gaussian errors. However, from an implementation standpoint, this extension might be challenging since we must take gradients of the kernel hyperparameters during the Newton Raphson optimization steps. Hence, a similar framework used in \citet{laplace_integrated} might be needed for the practical implementation of SKIM-FA to general responses.

\subsection{Change-of-Basis Formula For General $Q$ and Arbitrary Measures}
We discuss how to extend the change-of-basis formula in  \cref{algo:change_basis} to general $Q$. The general idea is as follows:
\begin{enumerate}
	\item (PROJECT) Project each $Q$ way interaction onto the space spanned by all lower-order interactions
	\item (UPDATE) Subtract out the lower-order variation from the $Q$ way interactions, and add back this projected variation to the lower-order interactions
	\item (RECURSE) Repeat Step 1 and Step 2 for the $Q-1$ way interactions, then $Q-2$ interactions, until the highest order interaction is the constant term
\end{enumerate}
Line 6 in \cref{algo:change_basis} is the analogue of the PROJECT step, and Lines 7-9 in  \cref{algo:change_basis} is the analogue of the UPDATE STEP. Under \cref{assum:tens_space}, we describe how the PROJECT-UPDATE-RECURSE methodology can be used to move between two arbitrary functional ANOVA decompositions with respect to $\mu^{\prime}$ and $\mu$. 

Suppose we are given the functional ANOVA decomposition with respect to $\mu^{\prime}$: $f = \sum_{V: |V| \leq Q} f_V^{\prime}$, where $f_V^{\prime} \in \hilb_{V, \mu^{\prime}}^o$. Given each $f_V^{\prime}$, we would like re-express $f$ as $\sum_{V: |V| \leq Q} f_V$, where $f_V \in \hilb_{V, \mu}^o$. Such a decomposition exists since $ \hilb^o_{V, \mu_{\prime}} =  \hilb^o_{V, \mu}$ by  \cref{assum:tens_space} and \cref{eq:universal_space}. We define the projection operator of a function with interactions in $A \subset [p], |A| \leq Q$ below:
\begin{equation}
	\mathrm{proj}_{\mathcal{F}_A, \mu}[f_A] \coloneqq \sum_{V: V \subsetneq A} g_{V_A}, \quad g_{V_A} \in \hilb_{V,  \mu}^o, \quad \mathcal{F}_A = \bigoplus_{V: V \subsetneq A} \hilb_{V, \mu}^o
\end{equation}
The UPDATE step involves adding and subtracting $g_{V_A}$ from the interactions. Since each component $g_{V_A}$ in $\mathrm{proj}_{\mathcal{F}_A, \mu}[f_A]$ is unique by the Hilbert Projection Theorem, this procedure is well defined. We summarize our algorithm in \cref{algo:change_basis_general}. The proof of correctness, namely that \cref{algo:change_basis_general} recovers the functional ANOVA decomposition of $f$ with respect to $\mu$ follows by using the same proof strategy in  \cref{thm:theo_proj}.

\begin{algorithm}[h]
    \caption{Change of Basis Formula for General $Q$ and Measures}
    \label{algo:change_basis_general}
    \begin{algorithmic}[1] 
         \Procedure{ReExpressANOVA}{$\sum_{V: |V| \leq Q} f_V^{\prime}$, $\mu$}  
         	\State Let $I$ equal the highest order interaction in $\sum_{V: |V| \leq Q} f_V^{\prime}$
         	 \If{$I = 0$}
         	\Return $f_{\emptyset}^{\prime}$
         \EndIf
			\State For all $A \subset [p], |A| = I$, compute  $\mathrm{proj}_{\mathcal{F}_A, \mu}[f_A^{\prime}]$.
			\State  For all $A \subset [p], |A| = I$, let $f_A = f_A^{\prime} - \mathrm{proj}_{\mathcal{F}_A, \mu}[f_A^{\prime}]$ \Comment{Update higher-order interaction effects}
			\State For all $V \subset [p], |V| < I$, let $f_V^{\prime} = f_V^{\prime} + \sum_{A: |A| = I, V \subsetneq A} g_{V_A}$ \Comment{Update lower-order interaction effects}
            \State \textbf{return} $\sum_{A: |A| = I} f_A$ + ReExpressANOVA($\sum_{V: |V| \leq I-1} f_V^{\prime}$, $\mu$) \Comment{RECURSE step}
        \EndProcedure
    \end{algorithmic}
\end{algorithm}

\noindent \textit{Computing the Projection Operator via Monte-Carlo.} We show how to compute $\mathrm{proj}_{\mathcal{F}_A, \mu}[f_A^{\prime}]$ via Monte-Carlo as we do in \cref{algo:change_basis} when $Q = 2$. To this end, let the components of $\Phi_A = \bigotimes_{j \in A} \Phi_j$. Let $d^* = \max_{V: |V| \leq Q} \mathrm{dim}(\Phi_A)$.  For $1 \leq w \leq W$ randomly sample $x^{(w)} \iid \mu$, where $W > d^*$. Define
\begin{equation}
	X_A = [\Phi_V(x^{(1)}_V) \cdots \Phi_V(x^{(W)}_V)]^T_{V: V \subsetneq A}
\end{equation}
where $\Phi_{\emptyset} = 1$. Let $f_{A, W}^{\prime} = [f_{A}^{\prime}(x^{(1)}_V) \cdots f_{A}^{\prime}(x^{(W)}_V)]^T$. Then,
\begin{equation}
	\hat{g}_{V_A}(\cdot) = \hat{\Psi}_{V}^T \Phi_V(\cdot), \quad \text{where} \quad  [\hat{\Psi}_{V}]^T_{V: V \subsetneq A} = (X_{A}^T X_{A})^{-1} X_{A}^T f_{A, W}^{\prime}.
\end{equation}
By following nearly an identical proof of \cref{prop:change_basis_correct}, $\hat{g}_{V_A}(\cdot) \rightarrow g_{V_A}(\cdot)$ as $W \rightarrow \infty$.

\subsection{Consistency Guarantees}
Proving selection consistency is beyond the scope of the current paper. Given selection consistency, however, estimation consistency follows from the work in \citet{projection_anova}, where the author examines the consistency properties of fitting functional ANOVA models via a finite-dimensional tensor product basis (i.e., as we do in SKIM-FA). 

To make this connection more concrete, suppose the number of correct covariates $S$ is fixed (and does not depend on $p$ or $N$). If selection consistency holds, then SKIM-FA consistently recovers $S$ with probability one as $N,p \rightarrow \infty$. To simplify the analysis, suppose we use sample splitting, where the first $\frac{N}{2}$ datapoints are used for selection, and the remaining $\frac{N}{2}$ datapoints are used to re-estimate the effects among the selected covariates. For any desired probability, $N$ can be chosen sufficiently large such that SKIM-FA selects all $S$ correct covariates exceeding the chosen probability. Since $S$ is fixed, the results in \citet{projection_anova} apply when estimating the effects on the held-out set of  $\frac{N}{2}$ datapoints. In particular, \citet{projection_anova} provides the rate at which $N$ must grow as a function of the size and smoothness of the tensor product basis generated from the $S$ selected covariates; see Theorem 3 and Corollary 2 of \citet{projection_anova} for different rates of convergence. 

\color{black}

\section{MARS ANOVA Procedure} \label{A:mars_anova} 

We show how to perform the functional ANOVA decomposition of $\hat{f}$ with respect to $\hat{\mu}_{\tens} = \hat{\mu}_1 \tens \cdots \tens \hat{\mu}_p$, where $\hat{f}$ denotes the regression function fit from MARS and $\mu_i$ the empirical distribution of covariate $i$: $\hat{\mu}_i = \frac{1}{N} \sum_{n=1}^N \delta_{x^{(n)}_i}$. Under $\hat{\mu}_{\tens}$, the functional ANOVA decomposition of $\hat{f}$ equals
\begin{equation} \label{eq:anova_mars_expect}
	\begin{split}
		\hat{f}_{\emptyset} &= \E_{\hat{\mu}_{\tens}}[\hat{f}] \\
		\hat{f}_{\{i\}}(x_i) &= \E_{\hat{\mu}_{\tens}}[\hat{f} \mid x_i = x_i] - \hat{f}_{\emptyset} \\
		\hat{f}_{\{i, j\}}(x_i, x_j) &= \E_{\hat{\mu}_{\tens}}[\hat{f} \mid x_i = x_i, x_j = x_j] - \hat{f}_{\{i\}}(x_i) - \hat{f}_{\{j\}}(x_i) - \hat{f}_{\emptyset},
	\end{split}
\end{equation}
which is also shown in \citet[Equation 5]{anova_zero_mean_sens}. We show how to compute each of the expectations in \cref{eq:anova_mars_expect}. The intercept $\hat{f}_{\emptyset}$ equals the sample average of the fitted values (i.e., $\hat{f}$ applied to each of the $N$ training datapoints). Let $X$ denote the $N \times p$  matrix of training data. Let $X^{i}$ equal the matrix obtained by setting all values in the $i$th column of $X$ equal to $x_i$ and the remaining columns unchanged. Then,  
\begin{equation*}
	 \E_{\hat{\mu}_{\tens}}[\hat{f} \mid x_i = x_i]  = \frac{1}{N} \sum_{n=1}^N \hat{f}(X^{i}_n),
\end{equation*}
where $X^{i}_n$ is the $n$th row of $X^{i}_n$. Similarly, let $X^{ij}$ equal the matrix obtained by setting all values in the $i$th and $j$th columns of $X$ equal to $x_i$ and $x_j$ respectively, and the remaining columns unchanged. Then,  
\begin{equation*}
	\E_{\hat{\mu}_{\tens}}[\hat{f} \mid x_i = x_i, x_j = x_j] = \frac{1}{N} \sum_{n=1}^N \hat{f}(X^{ij}_n).
\end{equation*}

\section{Additional Experimental Details} \label{A:add_exp_detail}

\subsection{Fitting benchmark methods}
\noindent \textit{SPAM-2Stage}: we perform variable selection by fitting a sparse additive model (SpAM)  \citep{spam} to the data. We use the \texttt{sam} package in \texttt{R}. Since \texttt{sam}  does not provide a default way to select the $L_1$ regularization strength, we use 5-fold cross-validation. For estimation, we generate all main and interaction effects among the subset of covariates selected by SpAM. We calculate these effects by taking pairwise products of  univariate basis functions generated from a natural cubic spline basis with 5 total knots; see \cref{A:model_implement_details} for details. We estimate the basis coefficients (and hence effects) using ridge regression, where again we use 5-fold cross-validation to pick the $L_2$ regularization strength. \\
  
 	\noindent \textit{Multivariate Additive Regression Splines (MARS)}: we use the \texttt{python} implementation of MARS \citep{mars} in \texttt{py-earth}.  \\
 		
 	\noindent \textit{Hierarchical Lasso (HierLasso)}: we use the implementation of HierLasso \citep{lim_heirch_lasso} in the authors' \texttt{R} package \texttt{glinternet}. Since \citet{lim_heirch_lasso} use cross-validation to pick the $L_1$ regularization strength, we similarly use 5-fold cross-validation. \\
 	
 	\noindent \textit{Pairs Lasso}: we fit the Lasso on the expanded set of features $\{x_i \}_{i=1}^p$ and $\{x_ix_j\}_{i,j=1}^p$. We fit the Lasso using the \texttt{python} package \texttt{sklearn}, and use 5-fold cross-validation to select the $L_1$ regularization strength.
 	
\subsection{Evaluation Criteria}
\noindent \textit{Variable Selection Evaluation Metrics.} We consider both the power to select correct covariates and avoid incorrect ones. \emph{\# Correct Selected} counts the number  of covariates correctly selected by the method. Higher is better. \emph{\# Wrong Selected}  counts the number of covariates incorrectly selected by the method (i.e., Type I error). Lower is better. \emph{\# Correct Not Selected} counts the number of covariates that belong to the true model but were not selected by the method (i.e., Type II error). Lower is better. \\

\noindent \textit{Estimation Evaluation Metrics.} We evaluate how well a method estimates main effects and interaction effects. Instead of  looking only at the total mean squared estimation error, we break this error into multiple buckets to understand what bucket drives the majority of the error. Lower is better for all of the following quantities. \emph{Correct Selected SSE (Main)} takes the sum of squared errors (SSE) between each estimated main effect component and true main effect component. This sum equals $\sum_{i \in S_1} \| f_i^* - \hat{f}_i \|_{\mu}^2$, where $S_1$ is the set of correctly identified main effects, $\hat{f}_i $ is the estimated main effect, and $f_i^*$ is the true main effect. \emph{Correct Not Selected SSE (Main)} takes the sum of squared norms of main effects not selected. This sum equals $\sum_{i \in S_2} \| f_i^* \|_{\mu}^2$, where $S_2$ is the set of correct  main effects not selected. \emph{Wrong Selected SSE (Main)} takes the sum of  squared norms of  main effect components incorrectly selected. This sum equals $\sum_{i \in S_3} \| \hat{f}_i \|_{\mu}^2$, where $S_3$ is the set of incorrect  main effects selected. \emph{Correct Selected SSE (Pair)}, \emph{Correct Not Selected SSE (Pair)}, and \emph{Wrong Selected SSE (Pair)} are the same as the analogous main effect metrics but instead considers interaction effects. \emph{Total SSE} equals the sum of the 6 buckets above and \emph{Total SSE / Signal Variance} equals the relative estimation error, i.e., Total SSE divided by the true signal variance.

 \section{Additional Experimental Results} \label{A:add_exp}

\begin{landscape}
\begin{Table}[]
\centering
\begin{tabular}{@{}ccccc@{}}
\toprule
\textbf{Method} & \textbf{\# Covariates} & \textbf{\# Correct Selected} & \textbf{\# Wrong Selected} & \textbf{\# Correct Not Selected} \\ \midrule
HierLasso       & 250        & 4                   & 0                          & 1                   \\
SKIM -FA           & 250        & 4                   & 0                          & 1                   \\
Pairs Lasso     & 250        & 4                   & 5                          & 1                   \\
SPAM-2Stage     & 250        & 5                   & 52                         & 0                   \\
MARS            & 250        & 5                   & 58                         & 0                   \\ \hline
                &            &                     &                            &                     \\
SKIM-FA            & 500        & 5                   & 13                         & 0                   \\
SPAM-2Stage     & 500        & 5                   & 28                         & 0                   \\
Pairs Lasso     & 500        & 4                   & 39                         & 1                   \\
HierLasso       & 500        & 4                   & 48                         & 1                   \\
MARS            & 500        & 5                   & 64                         & 0                   \\ \hline
                &            &                     &                            &                     \\
SKIM-FA            & 1000       & 3                   & 0                          & 2                   \\
HierLasso       & 1000       & 4                   & 5                          & 1                   \\
Pairs Lasso     & 1000       & 4                   & 6                          & 1                   \\
SPAM-2Stage     & 1000       & 5                   & 15                         & 0                   \\
MARS            & 1000       & 5                   & 70                         & 0                   \\ \bottomrule
\end{tabular}
\caption{Variable Selection Performance for Main Effects Only Setting.} \label{tab:var_select_main_only}
\end{Table}
\end{landscape}

\begin{landscape}
\begin{Table}[]
\centering
\begin{tabular}{cccccccccc}
\hline
\textbf{Method} & \textbf{p} & \textbf{\begin{tabular}[c]{@{}c@{}}Correct \\ Selected \\ SSE \\ (Main)\end{tabular}} & \textbf{\begin{tabular}[c]{@{}c@{}}Correct \\ Not \\ Selected \\ SSE \\ (Main)\end{tabular}} & \textbf{\begin{tabular}[c]{@{}c@{}}Wrong\\ Selected \\ SSE\\ (Main)\end{tabular}} & \textbf{\begin{tabular}[c]{@{}c@{}}Correct\\ Selected \\ SSE \\ (Pair)\end{tabular}} & \textbf{\begin{tabular}[c]{@{}c@{}}Correct\\ Not \\ Selected \\ SSE\\ (Pair)\end{tabular}} & \textbf{\begin{tabular}[c]{@{}c@{}}Wrong\\ Selected \\ SSE\\ (Pair)\end{tabular}} & \textbf{\begin{tabular}[c]{@{}c@{}}Total \\ SSE\end{tabular}} & \textbf{\begin{tabular}[c]{@{}c@{}}Total SSE\\ ÷ \\ Signal \\ Variance\end{tabular}} \\ \hline
MARS-EMP & 250 & 0.31 & 0.00 & 2.11 & 0.00 & 0.00 & 2.24 & 4.66 & 0.23 \\
SPAM-2Stage & 250 & 2.77 & 0.00 & 1.99 & 0.00 & 0.00 & 0.08 & 4.84 & 0.24 \\
SKIM-FA & 250 & 2.75 & 4.02 & 0.00 & 0.00 & 0.00 & 0.39 & 7.16 & 0.36 \\
MARS-VANILLA & 250 & 73.17 & 0.00 & 2.37 & 0.00 & 0.00 & 9.89 & 85.43 & 4.27 \\ \hline
 &  &  &  &  &  &  &  &  &  \\
SPAM-2Stage & 500 & 2.82 & 0.00 & 1.25 & 0.00 & 0.00 & 0.04 & 4.11 & 0.21 \\
SKIM-FA & 500 & 2.75 & 0.00 & 0.49 & 0.00 & 0.00 & 1.40 & 4.64 & 0.23 \\
MARS-EMP & 500 & 0.38 & 0.00 & 2.37 & 0.00 & 0.00 & 2.19 & 4.95 & 0.25 \\
MARS-VANILLA & 500 & 35.67 & 0.00 & 3.62 & 0.00 & 0.00 & 9.22 & 48.51 & 2.43 \\ \hline
 &  &  &  &  &  &  &  &  &  \\
SPAM-2Stage & 1000 & 2.67 & 0.00 & 0.78 & 0.00 & 0.00 & 0.02 & 3.46 & 0.17 \\
MARS-EMP & 1000 & 0.45 & 0.00 & 2.68 & 0.00 & 0.00 & 2.39 & 5.51 & 0.28 \\
SKIM-FA & 1000 & 2.70 & 8.10 & 0.00 & 0.00 & 0.00 & 0.24 & 11.03 & 0.55 \\
MARS-VANILLA & 1000 & 16.14 & 0.00 & 1.56 & 0.00 & 0.00 & 10.33 & 28.02 & 1.40 \\ \hline
\end{tabular}
\caption{Estimation Performance for Main Effects Only Setting.}
\label{tab:est_main_only} 
\end{Table}
\end{landscape}

\begin{landscape}
\begin{Table}[]
\centering
\begin{tabular}{@{}ccccc@{}}
\toprule
\textbf{Method} & \textbf{\# of Covariates} & \textbf{\# Correct Selected} & \textbf{\# Wrong Selected} & \textbf{\# Correct Not Selected} \\ \midrule
SKIM-FA         & 250                       & 5                            & 1                          & 0                                  \\
HierLasso       & 250                       & 5                            & 25                         & 0                                  \\
SPAM-2Stage     & 250                       & 5                            & 37                         & 0                                  \\
MARS            & 250                       & 5                            & 84                         & 0                                  \\
Pairs Lasso     & 250                       & 5                            & 89                         & 0                                  \\ \hline
                &            &                     &                            &                     \\

SKIM-FA         & 500                       & 5                            & 0                          & 0                                  \\
SPAM-2Stage     & 500                       & 5                            & 29                         & 0                                  \\
HierLasso       & 500                       & 5                            & 30                         & 0                                  \\
MARS            & 500                       & 5                            & 69                         & 0                                  \\
Pairs Lasso     & 500                       & 5                            & 182                        & 0                                  \\ \hline
                &            &                     &                            &                     \\

SKIM-FA         & 1000                      & 5                            & 0                          & 0                                  \\
SPAM-2Stage     & 1000                      & 5                            & 15                         & 0                                  \\
HierLasso       & 1000                      & 5                            & 40                         & 0                                  \\
MARS            & 1000                      & 5                            & 71                         & 0                                  \\
Pairs Lasso     & 1000                      & 5                            & 213                        & 0                                  \\ \bottomrule
\end{tabular}
\caption{Variable Selection Performance for Equal Main and Interaction Effects Setting.}  \label{tab:var_select_equal}
\end{Table}
\end{landscape}

\begin{landscape}
\begin{Table}[]
\centering
\begin{tabular}{@{}cccccccccc@{}}
\toprule
\textbf{Method} & \textbf{p} & \textbf{\begin{tabular}[c]{@{}c@{}}Correct \\ Selected \\ SSE \\ (Main)\end{tabular}} & \textbf{\begin{tabular}[c]{@{}c@{}}Correct \\ Not \\ Selected \\ SSE \\ (Main)\end{tabular}} & \textbf{\begin{tabular}[c]{@{}c@{}}Wrong\\ Selected \\ SSE\\ (Main)\end{tabular}} & \textbf{\begin{tabular}[c]{@{}c@{}}Correct\\ Selected \\ SSE \\ (Pair)\end{tabular}} & \textbf{\begin{tabular}[c]{@{}c@{}}Correct\\ Not \\ Selected \\ SSE\\ (Pair)\end{tabular}} & \textbf{\begin{tabular}[c]{@{}c@{}}Wrong\\ Selected \\ SSE\\ (Pair)\end{tabular}} & \textbf{\begin{tabular}[c]{@{}c@{}}Total \\ SSE\end{tabular}} & \textbf{\begin{tabular}[c]{@{}c@{}}Total SSE\\ ÷ \\ Signal \\ Variance\end{tabular}} \\ \midrule
SKIM-FA & 250 & 1.62 & 0.00 & 0.08 & 0.52 & 0.00 & 0.17 & 2.39 & 0.12 \\
SPAM-2Stage & 250 & 1.63 & 0.00 & 1.72 & 8.84 & 0.00 & 0.11 & 12.30 & 0.62 \\
MARS-EMP & 250 & 0.71 & 0.00 & 4.44 & 2.17 & 0.00 & 5.69 & 13.01 & 0.65 \\
MARS-VANILLA & 250 & 24.91 & 0.00 & 5.28 & 17.13 & 0.00 & 18.03 & 65.35 & 3.27 \\ \hline
 &  &  &  &  &  &  &  &  &  \\ 
SKIM-FA & 500 & 1.52 & 0.00 & 0.00 & 0.41 & 0.00 & 0.00 & 1.93 & 0.10 \\
SPAM-2Stage & 500 & 1.62 & 0.00 & 3.74 & 2.16 & 0.00 & 5.47 & 12.99 & 0.65 \\
MARS-EMP & 500 & 0.71 & 0.00 & 4.69 & 1.63 & 0.96 & 6.57 & 14.56 & 0.73 \\
MARS-VANILLA & 500 & 11.36 & 0.00 & 13.22 & 15.62 & 0.96 & 23.55 & 64.71 & 3.24 \\ \hline
 &  &  &  &  &  &  &  &  &  \\
SKIM-FA & 1000 & 1.54 & 0.00 & 0.00 & 0.29 & 0.00 & 0.00 & 1.82 & 0.09 \\
SPAM-2Stage & 1000 & 1.67 & 0.00 & 1.07 & 0.41 & 0.00 & 2.16 & 5.31 & 0.27 \\
MARS-EMP & 1000 & 0.61 & 0.00 & 3.84 & 1.70 & 0.00 & 2.52 & 8.67 & 0.43 \\
MARS-VANILLA & 1000 & 454.88 & 0.00 & 3.16 & 21.46 & 0.00 & 13.22 & 492.72 & 24.64 \\ \bottomrule
\end{tabular}
\caption{Estimation Performance for Equal Main and Interaction Effects Setting.} \label{tab:est_equal}
\label{tab:my-table}
\end{Table}
\end{landscape}

\begin{landscape}
\begin{Table}[]
\centering
\begin{tabular}{@{}ccccc@{}}
\toprule
\textbf{Method} & \textbf{\# Covariates} & \textbf{\# Correct Selected} & \textbf{\# Wrong Selected} & \textbf{\# Correct Not Selected} \\ \midrule
SKIM-FA         & 250                    & 5                            & 6                          & 0                                \\
MARS            & 250                    & 5                            & 75                         & 0                                \\
SPAM-2Stage     & 250                    & 4                            & 77                         & 1                                \\
Pairs Lasso     & 250                    & 5                            & 123                        & 0                                \\
HierLasso       & 250                    & 5                            & 160                        & 0                                \\ \hline
                &                        &                              &                            &                                  \\
SKIM-FA         & 500                    & 5                            & 16                         & 0                                \\
SPAM-2Stage     & 500                    & 1                            & 21                         & 4                                \\
HierLasso       & 500                    & 5                            & 62                         & 0                                \\
Pairs Lasso     & 500                    & 5                            & 85                         & 0                                \\
MARS            & 500                    & 2                            & 132                        & 3                                \\ \hline
                &                        &                              &                            &                                  \\
SKIM-FA         & 1000                   & 5                            & 9                          & 0                                \\
SPAM-2Stage     & 1000                   & 1                            & 41                         & 4                                \\
MARS            & 1000                   & 5                            & 75                         & 0                                \\
HierLasso       & 1000                   & 5                            & 120                        & 0                                \\
Pairs Lasso     & 1000                   & 5                            & 144                        & 0                                \\ \bottomrule
\end{tabular}
\caption{Variable Selection Performance for Weak Main Effects Setting.}
\label{tab:var_select_main_weak}
\end{Table}
\end{landscape}

\begin{landscape}
\begin{Table}[]
\centering
\begin{tabular}{@{}cccccccccc@{}}
\toprule
\textbf{Method} & \textbf{p} & \textbf{\begin{tabular}[c]{@{}c@{}}Correct \\ Selected \\ SSE \\ (Main)\end{tabular}} & \textbf{\begin{tabular}[c]{@{}c@{}}Correct \\ Not \\ Selected \\ SSE \\ (Main)\end{tabular}} & \textbf{\begin{tabular}[c]{@{}c@{}}Wrong\\ Selected \\ SSE\\ (Main)\end{tabular}} & \textbf{\begin{tabular}[c]{@{}c@{}}Correct\\ Selected \\ SSE \\ (Pair)\end{tabular}} & \textbf{\begin{tabular}[c]{@{}c@{}}Correct\\ Not \\ Selected \\ SSE\\ (Pair)\end{tabular}} & \textbf{\begin{tabular}[c]{@{}c@{}}Wrong\\ Selected \\ SSE\\ (Pair)\end{tabular}} & \textbf{\begin{tabular}[c]{@{}c@{}}Total \\ SSE\end{tabular}} & \textbf{\begin{tabular}[c]{@{}c@{}}Total SSE\\ ÷ \\ Signal \\ Variance\end{tabular}} \\ \midrule
SKIM-FA & 250 & 0.45 & 0.00 & 0.95 & 0.73 & 0.00 & 0.77 & 2.89 & 0.14 \\
MARS-EMP & 250 & 1.46 & 0.00 & 4.02 & 4.83 & 0.00 & 4.67 & 14.97 & 0.75 \\
SPAM-2Stage & 250 & 0.09 & 0.05 & 2.22 & 10.72 & 7.73 & 0.42 & 21.23 & 1.06 \\
MARS-VANILLA & 250 & 22497.35 & 0.00 & 7.31 & 148073.29 & 0.00 & 18.55 & 170596.50 & 8529.83 \\ \hline
 &  &  &  &  &  &  &  &  &  \\
SKIM-FA & 500 & 0.69 & 0.00 & 2.05 & 1.50 & 0.00 & 1.37 & 5.61 & 0.28 \\
SPAM-2Stage & 500 & 0.27 & 0.20 & 4.09 & 0.00 & 19.46 & 0.08 & 24.11 & 1.21 \\
MARS-EMP & 500 & 0.41 & 0.15 & 21.92 & 0.00 & 19.46 & 15.56 & 57.51 & 2.88 \\
MARS-VANILLA & 500 & 0.10 & 0.15 & 323788.65 & 0.00 & 19.46 & 324588.33 & 648396.70 & 32419.83 \\ \hline
 &  &  &  &  &  &  &  &  &  \\
SKIM-FA & 1000 & 0.72 & 0.00 & 1.37 & 0.61 & 0.00 & 0.63 & 3.33 & 0.17 \\
MARS-EMP & 1000 & 0.67 & 0.00 & 5.86 & 3.37 & 0.00 & 5.63 & 15.52 & 0.78 \\
SPAM-2Stage & 1000 & 0.16 & 0.20 & 6.69 & 0.00 & 18.33 & 0.31 & 25.69 & 1.28 \\
MARS-VANILLA & 1000 & 23.62 & 0.00 & 3.18 & 23.16 & 0.00 & 15.43 & 65.39 & 3.27 \\ \bottomrule
\end{tabular}
\caption{Estimation Performance for Weak Main Effects Setting.}
\label{tab:est_main_weak}
\end{Table}
\end{landscape}

\begin{Table}[]
\centering
\begin{tabular}{@{}cc@{}}
\toprule
\textbf{Effect}       & \textbf{Signal Variance} \\ \midrule
Hour                  & 0.382             \\
Air Temp.             & 0.104             \\
Humidity              & 0.024             \\
Windspeed             & 0.002             \\
Hour x Air Temp.      & 0.047             \\
Hour x Humidity       & 0.01              \\
Hour x Windspeed      & 0.002             \\
Air Temp. x Humidity  & 0.012             \\
Air Temp. x Windspeed & 0.005             \\
Humidity x Windspeed  & 0.003             \\ \bottomrule
\end{tabular}
\caption{Proxy Ground Truth Effects and Signal Variances for the Bike Sharing Data Set.}
\label{tab:ground_truth_effects}
\end{Table}

\begin{Table}[]
\centering
\begin{tabular}{@{}cccc@{}}
\toprule
\textbf{Method} & \textbf{\# Covariates} & \textbf{\# Original Selected} & \textbf{\# Wrong Selected} \\ \midrule
SKIM-FA         & 250                    & 2                             & 0                          \\
HierLasso       & 250                    & 3                             & 7                          \\
Pairs Lasso     & 250                    & 3                             & 29                         \\
MARS            & 250                    & 3                             & 96                         \\
SPAM-2Stage     & 250                    & 4                             & 97                         \\ \hline \\
                
SKIM-FA         & 500                    & 2                             & 0                          \\
HierLasso       & 500                    & 3                             & 8                          \\
SPAM-2Stage     & 500                    & 3                             & 22                         \\
Pairs Lasso     & 500                    & 3                             & 39                         \\
MARS            & 500                    & 4                             & 109                        \\  \hline \\

SKIM-FA         & 1000                   & 3                             & 0                          \\
HierLasso       & 1000                   & 3                             & 5                          \\
SPAM-2Stage     & 1000                   & 3                             & 8                          \\
Pairs Lasso     & 1000                   & 3                             & 76                         \\
MARS            & 1000                   & 3                             & 119                        \\ \bottomrule
\end{tabular}
\caption{Variable Selection Performance for the Bike Sharing Data Set.}
\label{tab:bike_var_perform_full}
\end{Table}

\begin{Table}[]
\centering
\resizebox{\textwidth}{!}{%
\begin{tabular}{@{}ccccccccc@{}}
\toprule
\textbf{Method} & \textbf{\# Noise} & \textbf{\begin{tabular}[c]{@{}c@{}}Correct \\ Selected \\ SSE \\ (Main)\end{tabular}} & \textbf{\begin{tabular}[c]{@{}c@{}}Correct \\ Not \\ Selected \\ SSE \\ (Main)\end{tabular}} & \textbf{\begin{tabular}[c]{@{}c@{}}Wrong\\ Selected \\ SSE\\ (Main)\end{tabular}} & \textbf{\begin{tabular}[c]{@{}c@{}}Correct\\ Selected \\ SSE \\ (Pair)\end{tabular}} & \textbf{\begin{tabular}[c]{@{}c@{}}Correct\\ Not \\ Selected \\ SSE\\ (Pair)\end{tabular}} & \textbf{\begin{tabular}[c]{@{}c@{}}Wrong\\ Selected \\ SSE\\ (Pair)\end{tabular}} & \textbf{\begin{tabular}[c]{@{}c@{}}Total \\ SSE\end{tabular}} \\ \midrule

SKIM-FA & 250 & 0.15 & 0.027 & 0 & 0.019 & 0.038 & 0 & 0.233 \\
SPAM-2Stage & 250 & 0.149 & 0 & 0.172 & 0.091 & 0 & 0.01 & 0.422 \\
MARS-EMP & 250 & 0.209 & 0.002 & 0.476 & 0.052 & 0.026 & 0.344 & 1.11 \\
MARS-Vanilla & 250 & 6.522 & 0.002 & 1.644 & 1.036 & 0.026 & 2.2 & 11.431 \\ \hline \\

SKIM-FA & 500 & 0.148 & 0.027 & 0 & 0.019 & 0.038 & 0 & 0.231 \\
SPAM-2Stage & 500 & 0.15 & 0.002 & 0.057 & 0.081 & 0.009 & 0.002 & 0.302 \\
MARS-EMP & 500 & 0.225 & 0 & 0.529 & 0.052 & 0.026 & 0.3 & 1.131 \\
MARS-Vanilla & 500 & 5.564 & 0 & 0.5 & 1.037 & 0.026 & 2.085 & 9.212 \\ \hline \\

SKIM-FA & 1000 & 0.145 & 0.002 & 0 & 0.107 & 0.009 & 0 & 0.263 \\
SPAM-2Stage & 1000 & 0.149 & 0.002 & 0.027 & 0.081 & 0.009 & 0.000 & 0.269 \\
MARS-EMP & 1000 & 0.214 & 0.002 & 0.485 & 0.054 & 0.026 & 0.245 & 1.026 \\
MARS-Vanilla & 1000 & 6.556 & 0.002 & 0.796 & 0.947 & 0.026 & 1.882 & 10.209 \\ \bottomrule
\end{tabular}
}
\caption{Estimation Performance for the Bike Sharing Data Set.}
\label{tab:bike_est_perform_full}
\end{Table}

\subsection{Appending Irrelevant but Real Covariates to the Bike Sharing Data Set} \label{A:bike_fake_real}

In \cref{sec:experiments}, we appended fake covariates drawn from a Uniform(0, 1) distribution to the Bike Sharing Data Set for various choices of $p_{\mathrm{noise}}$. In many applications, however, covariates are correlated and this correlation structure might affect the performance of a method.  To create a design matrix with a non-trivial correlation structure, we start by taking a completely different data set, namely the SECOM data set from the UCI Machine Learning repository which contains 591 covariates related to semi-conductor manufacturing.\footnote{We only consider 432 continuous covariates (with non-zero variance) in the SECOM data set.} Then, we append these covariates to the Bike Sharing data set. Since these two data sets are independent, the covariates in the SECOM data set play the same role as the synthetic fake covariates in \cref{sec:experiments} (i.e., should not be selected) but now have a real correlation structure. \cref{tab:bike_var_perform_real_append} and \cref{tab:bike_est_perform_real_append} summarize how each method performs in terms of variable selection and estimation, respectively.

\begin{Table}[]
\centering
\label{tab:bike_var_perform_real_append}
\begin{tabular}{@{}cccc@{}}
\toprule
\textbf{Method} & \textbf{\# Covariates} & \textbf{\# Original Selected} & \textbf{\# Wrong Selected} \\ \midrule
\textbf{SKIM-FA}         & 432                   & 2                             & 0                          \\
HierLasso       & 432                   & 3                             & 1                          \\
SPAM-2Stage     & 432                   & 0                             & 0                          \\
Pairs Lasso     & 432                   & 3                             & 14                         \\
MARS            & 432                   & 3                             & 97                        \\ \bottomrule
\end{tabular}
\caption{Variable Selection Performance for the Bike Sharing-SECOM Data Set.}
\end{Table}

\begin{Table}[]
\centering
\resizebox{\textwidth}{!}{%
\begin{tabular}{@{}ccccccccc@{}}
\toprule
\textbf{Method} & \textbf{\# Noise} & \textbf{\begin{tabular}[c]{@{}c@{}}Correct \\ Selected \\ SSE \\ (Main)\end{tabular}} & \textbf{\begin{tabular}[c]{@{}c@{}}Correct \\ Not \\ Selected \\ SSE \\ (Main)\end{tabular}} & \textbf{\begin{tabular}[c]{@{}c@{}}Wrong\\ Selected \\ SSE\\ (Main)\end{tabular}} & \textbf{\begin{tabular}[c]{@{}c@{}}Correct\\ Selected \\ SSE \\ (Pair)\end{tabular}} & \textbf{\begin{tabular}[c]{@{}c@{}}Correct\\ Not \\ Selected \\ SSE\\ (Pair)\end{tabular}} & \textbf{\begin{tabular}[c]{@{}c@{}}Wrong\\ Selected \\ SSE\\ (Pair)\end{tabular}} & \textbf{\begin{tabular}[c]{@{}c@{}}Total \\ SSE\end{tabular}} \\ \midrule
\textbf{SKIM-FA} & 432 & 0.137 & 0.026 & 0 & 0.016 & 0.029 & 0 & 0.208 \\
SPAM-2Stage & 432 & 0 & 0.549 & 0 & 0 & 0.074 & 0 & 0.623 \\
MARS-EMP & 432 & 0.212 & 0.001 & 7.416 & 0.049 &  0.018 & 6.113 & 13.810 \\
MARS-Vanilla & 432 & 3.876 & 0.001 & 85.369 & 1.704 & 0.0178 & 104.405 &  195.373 \\ \bottomrule
\end{tabular}
}
\caption{Estimation Performance for the Bike Sharing-SECOM Data Set.}
\label{tab:bike_est_perform_real_append}
\end{Table}

\subsection{Impact of  Correlated Predictors on the Functional ANOVA for the Bike Sharing Data Set} \label{A:corr_pred_bike}

We perform the same analysis as in \cref{sec:anova_sens} but for the Bike Sharing data set in \cref{fig:bike_anova}. Unlike the Concrete Compressive Strength data set, however, we do not see a large difference between the two functional ANOVA decompositions for the Bike Sharing data set in  \cref{fig:bike_anova}.

\begin{figure}[]
        \centering
        \begin{subfigure}[b]{0.45\textwidth}
            \centering
            \includegraphics[width=\textwidth]{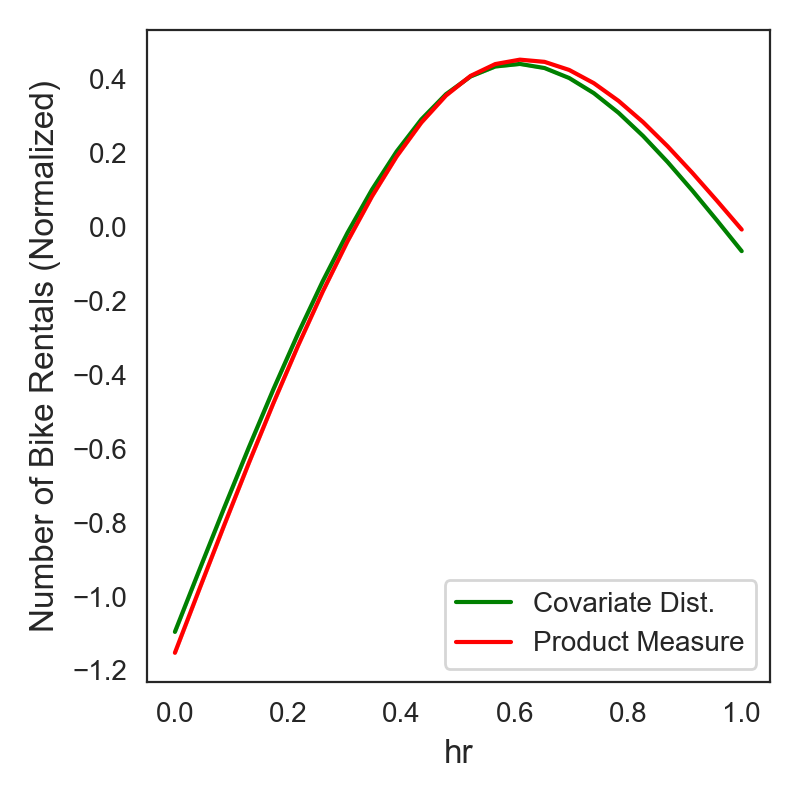}
            \caption[Network2]%
            {{\small  Main Effect of Hour of the Day on Rentals}}    
        \end{subfigure}
        \hspace{1cm}
        \begin{subfigure}[b]{0.45\textwidth}  
            \centering 
            \includegraphics[width=\textwidth]{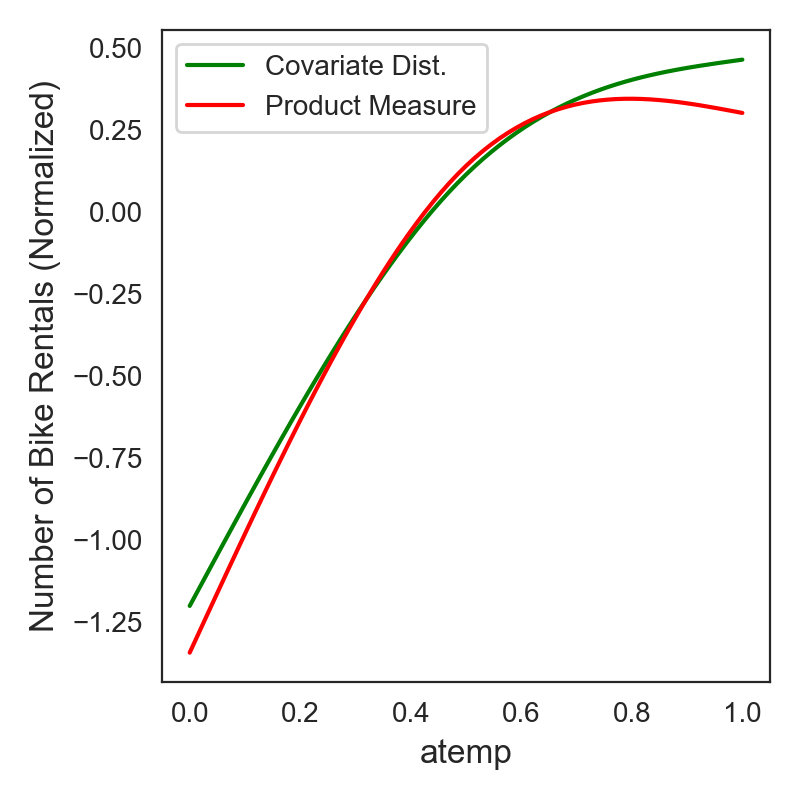}
            \caption[]%
            {{\small Main Effect of Hour of Temperature}}    
        \end{subfigure}
        \hspace{1cm}
        \begin{subfigure}[b]{0.45\textwidth}  
            \centering 
            \includegraphics[width=\textwidth]{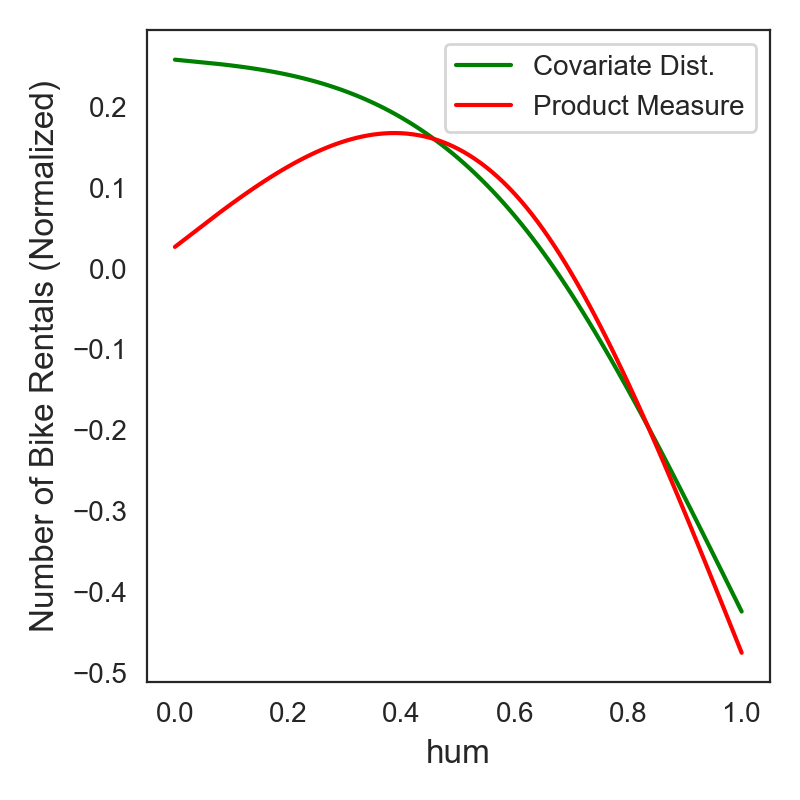}
            \caption[]%
            {{\small Main Effect of Humidity}}    
        \end{subfigure}
        \hspace{1cm}
        \begin{subfigure}[b]{0.45\textwidth}  
            \centering 
            \includegraphics[width=\textwidth]{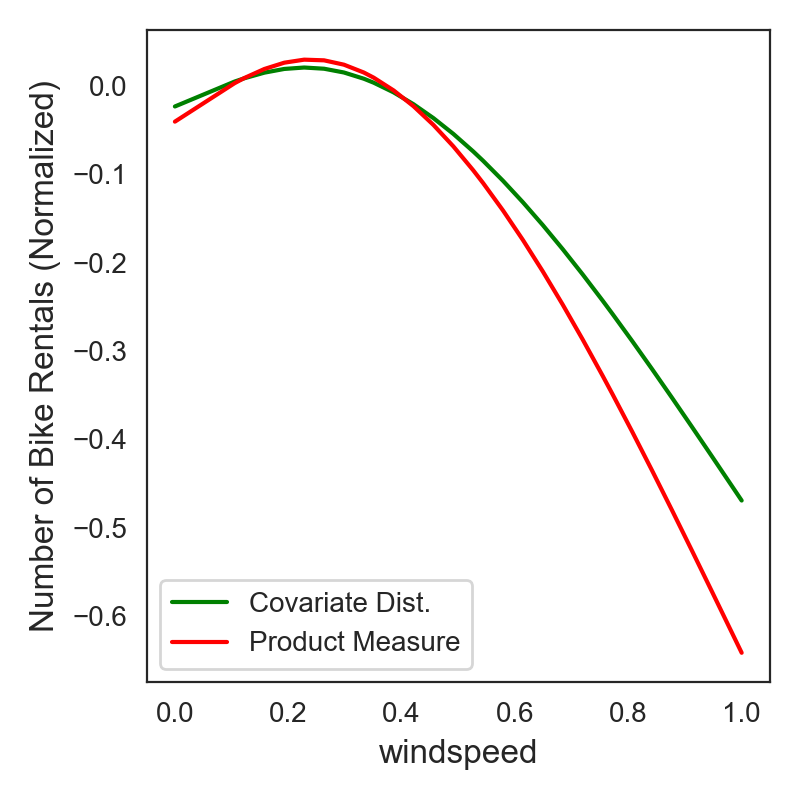}
            \caption[]%
            {{\small Main Effect of Windspeed}}    
        \end{subfigure}
                \caption{Effect of Correlated Predictors on the Bike Sharing Data Set} \label{fig:bike_anova}

\end{figure}

\color{editcolor}
\subsection{Obesity Gene-Expression and SNP Data Set: Additional SKIM-FA Fit Details} \label{A:obese_details}

Since the number of datapoints is small ($N=87$), the number of datapoints in between any adjacent knots in a spline basis is small. Hence, we instead fit a linear interaction model for this data set using SKIM-FA. Since we fit a linear interaction model, we can examine the regression coefficients to understand the fit. \cref{tab:skim_main_obese} summarizes the 31 variables selected by SKIM-FA and their estimated main effects.

We report the 10 strongest interaction effects in \cref{tab:skim_iter_obese}. We see that SKP1 has the largest number of strong interaction effects follows by IRSP2.  

\begin{Table}[b]
\centering
\begin{tabular}{@{}cc@{}}
\toprule
\textbf{Effect} & \textbf{Coeff.} \\ \midrule
GSC & 1.40 \\
IRS2 & -1.04 \\
EOGT & 1.03 \\
SNORA71B & -1.00 \\
SPATA20 & -0.95 \\
AK299501 & -0.80 \\
CLIC5 & 0.77 \\
SKP1 & 0.77 \\
SETDB2 & -0.77 \\
BTN2A3P & 0.76 \\
SAP130 & 0.76 \\
KISS1R & 0.69 \\
ZBED3 & -0.66 \\
AQP11 & -0.64 \\
IYD & 0.59 \\
LOC100287177 & -0.56 \\
TMEM74B & 0.47 \\
ATP2B3 & -0.44 \\
LOC283070 & 0.42 \\
IRX1 & 0.39 \\
RPA4 & -0.34 \\
TSHZ3 & 0.31 \\
INTS4 & -0.30 \\
ALDH1A2 & -0.30 \\
PCDH8 & 0.30 \\
FBN2 & 0.29 \\
KLK7 & -0.28 \\
FBXL12 & -0.27 \\
SEMA6B & 0.24 \\
SLC25A13 & -0.00 \\ \bottomrule
\end{tabular}
\caption{Main Effects Selected by SKIM-FA on the Obesity Gene-Expression and SNP Data Set}
\label{tab:skim_main_obese}
\end{Table}

\begin{Table}[b]
\centering
\begin{tabular}{@{}cc@{}}
\toprule
\textbf{Effect} & \textbf{Coeff.} \\ \midrule
(SKP1, SETDB2) & -0.14 \\
(SKP1, ZBED3) & 0.13 \\
(RPA4, SETDB2) & 0.12 \\
(IRS2, RPA4) & 0.11 \\
(IRS2, SNORA71B) & -0.11 \\
(SKP1, RPA4) & -0.10 \\
(SNORA71B, RPA4) & -0.08 \\
(RPA4, ZBED3) & 0.08 \\
(SKP1, IRS2) & -0.06 \\
(RPA4, SAP130) & 0.06 \\ \bottomrule
\end{tabular}%
\caption{Interaction Effects Selected by SKIM-FA on the Obesity Gene-Expression and SNP Data Set (10 Strongest Interactions Shown)}
\label{tab:skim_iter_obese}
\end{Table}

\subsection{Sensitivity to Non-Compactness and Sparse Interaction Effects}
To test the sensitivity of SKIM-FA to the compactness assumption in \cref{thm:unique_fanova} we instead draw covariates each independently from a Gaussian distribution. Since a Gaussian distribution has support on all of $\R$, the covariates do not belong to a compact set. On an unbounded set, the exponential function has an infinite mean. Hence, we replace the exponential trend in \cref{sec:syn_exp} with a cubic trend.  On an unbounded set, to model a sine trend, we would need to use a wavelet basis. Since we have only have support for a spline basis in the current version of our package, we replace the sine trend in \cref{sec:syn_exp} with a leaky rectified linear unit (ReLU) trend, which is often used in neural networks. We keep the linear, logistic, and quadatic effects. 

We let $y$ depend on the first 5 covariates; the remaining $995$ covariates are taken as noise covariates that we do not want to select. Hence, we consider a total of $p=1000$ covariates. We let the main effects equal the sum of the 5 trends discussed above, where the $i$th trend is applied to covariate $i$. To additionally test the impact of not having all interactions present, we only consider 5 out of the 10 possible pairwise interactions: linear-logistic, leaky ReLU-linear, leaky ReLU-quadratic, logistic-quadratic, and cubic-logistic. We select a noise variance such that the $R^2 = \frac{\sigma^2_{\text{signal}}}{\sigma^2_{\text{signal}} + \sigma^2_{\text{noise}}} = 0.8$, where $\sigma^2_{\text{signal}} = \inner{f^*}{f^*}_{\mu}$. We generate a total of $N=1000$ datapoints.

In terms of variable selection performance, SKIM-FA selects all 5 true covariates and 0 incorrect covariates. We summarize the estimation performance below:
\begin{itemize}
	\item Corrected Selected SSE (Main): .95
	\item Corrected Not Selected SSE (Main): 0
	\item Wrong Selected SSE (Main): 0
	\item Correct Selected SSE (Pair):  2.17
	\item Correct Not Selected SSE (Pair): 0
	\item Wrong Selected SSE (Pair): .94
	\item Total SSE: 4.05
\end{itemize}
Since SKIM-FA considers all pairwise interactions among selected covariates, SKIM-FA estimates 5 incorrect interactions. However, the total variance of these 5 incorrect interactions estimated by SKIM-FA is only .94. Hence, SKIM-FA shrinks all 5 incorrect interactions close to 0. 

Each true main and pairwise effect has variance 2. Since all covariates are independent, the total signal variance of the main and pairwise interaction effects is each 10. Hence, the normalized SSE for the true main effects is .95/10 = .095 and 2.17/10 = .217 for the true pairwise effects. 

\color{black}


%
%
%
%

\clearpage
\newpage

\bibliographystyle{rss}
\bibliography{references}

\end{document}